\RequirePackage[l2tabu,orthodox]{nag}
\documentclass
[11pt,letterpaper]
{article}

\usepackage[utf8]{inputenc}
\usepackage{etex}
\usepackage{xspace,enumerate}
\usepackage[dvipsnames]{xcolor}
\usepackage[T1]{fontenc}
\usepackage[full]{textcomp}
\usepackage[american]{babel}
\usepackage{mathtools}
\usepackage{titling}
\usepackage{amsthm}
\newtheorem{theorem}{Theorem}[section]
\newtheorem{notation}{Notation}[section]
\newtheorem*{theorem*}{Theorem}

\newtheorem{proposition}[theorem]{Proposition}
\newtheorem*{proposition*}{Proposition}
\newtheorem{lemma}[theorem]{Lemma}
\newtheorem*{lemma*}{Lemma}
\newtheorem{corollary}[theorem]{Corollary}
\newtheorem*{conjecture*}{Conjecture}
\newtheorem{fact}[theorem]{Fact}
\newtheorem*{fact*}{Fact}

\newtheorem*{hypothesis*}{Hypothesis}

\theoremstyle{definition}
\newtheorem{definition}[theorem]{Definition}
\newtheorem*{definition*}{Definition}

\newtheorem{algorithm}[theorem]{Algorithm}

\theoremstyle{remark}

\newtheorem*{claim*}{Claim}
\newtheorem{remark}[theorem]{Remark}
\newtheorem*{remark*}{Remark}

\newtheorem*{observation*}{Observation}

\usepackage[
letterpaper,
top=1.2in,
bottom=1.2in,
left=1in,
right=1in]{geometry}
\usepackage{newpxtext} 
\usepackage{textcomp} 
\usepackage[varg,bigdelims]{newpxmath}
\usepackage[scr=rsfso]{mathalfa}
\usepackage{bm} 
\let\mathbb\varmathbb
\usepackage{microtype}
\usepackage[
pagebackref,
colorlinks=true,
urlcolor=blue,
linkcolor=blue,
citecolor=OliveGreen,
]{hyperref}
\usepackage[capitalise,nameinlink]{cleveref}
\crefname{lemma}{Lemma}{Lemmas}
\crefname{fact}{Fact}{Facts}
\crefname{theorem}{Theorem}{Theorems}
\crefname{corollary}{Corollary}{Corollaries}
\crefname{claim}{Claim}{Claims}
\crefname{example}{Example}{Examples}
\crefname{algorithm}{Algorithm}{Algorithms}
\crefname{problem}{Problem}{Problems}
\crefname{definition}{Definition}{Definitions}
\usepackage{paralist}
\usepackage{turnstile}
\usepackage{mdframed}
\usepackage{tikz}
\usepackage{caption}
\DeclareCaptionType{Algorithm}
\usepackage{newfloat}

\newcommand{\Authornotecolored}[3]{}
\newcommand{\Authorcomment}[2]{}
\newcommand{\Authorfnote}[2]{}

\definecolor{forestgreen(traditional)}{rgb}{0.0, 0.27, 0.13}

\usepackage{boxedminipage}
\newcommand{\paren}[1]{(#1)}

\newcommand{\Paren}[1]{\left(#1\right)}

\newcommand{\Brac}[1]{\left[#1\right]}

\newcommand{\Abs}[1]{\left\lvert#1\right\rvert}

\newcommand{\set}[1]{\{#1\}}
\newcommand{\Set}[1]{\left\{#1\right\}}

\newcommand{\norm}[1]{\lVert#1\rVert}
\newcommand{\Norm}[1]{\left\lVert#1\right\rVert}

\newcommand{\iprod}[1]{\langle#1\rangle}
\newcommand{\Iprod}[1]{\left\langle#1\right\rangle}

\newcommand{\Esymb}{\mathbb{E}}
\newcommand{\Psymb}{\mathbb{P}}

\DeclareMathOperator*{\E}{\Esymb}

\DeclareMathOperator*{\ProbOp}{\Psymb}
\renewcommand{\Pr}{\ProbOp}

\newcommand{\mper}{\,.}
\newcommand{\mcom}{\,,}
\newcommand\bdot\bullet

\DeclareMathOperator{\poly}{poly}

\DeclareMathOperator{\dist}{dist}

\newcommand{\Hoelder}{H\"{o}lder\xspace}
\newcommand{\Holder}{\Hoelder}

\newcommand{\Z}{\mathbb Z}
\newcommand{\N}{\mathbb N}
\newcommand{\R}{\mathbb R}

\newcommand{\cA}{\mathcal A}
\newcommand{\cB}{\mathcal B}

\newcommand{\cD}{\mathcal D}

\newcommand{\cL}{\mathcal L}

\newcommand{\cN}{\mathcal N}
\newcommand{\cO}{\mathcal O}

\def\QQ{\mathbf{Q}}

\newcommand{\bbQ}{\mathbb Q}

\newcommand{\bbE}{\mathbb E}

\newcommand{\wh}{\widehat}

\renewcommand{\leq}{\leqslant}
\renewcommand{\le}{\leqslant}
\renewcommand{\geq}{\geqslant}
\renewcommand{\ge}{\geqslant}
\let\epsilon=\varepsilon
\numberwithin{equation}{section}
\newcommand\MYcurrentlabel{xxx}
\newcommand{\MYstore}[2]{%
  \global\expandafter \def \csname MYMEMORY #1 \endcsname{#2}%
}
\newcommand{\MYload}[1]{%
  \csname MYMEMORY #1 \endcsname%
}
\newcommand{\MYnewlabel}[1]{%
  \renewcommand\MYcurrentlabel{#1}%
  \MYoldlabel{#1}%
}
\newcommand{\MYdummylabel}[1]{}
\newcommand{\torestate}[1]{%
  \let\MYoldlabel\label%
  \let\label\MYnewlabel%
  #1%
  \MYstore{\MYcurrentlabel}{#1}%
  \let\label\MYoldlabel%
}
\newcommand{\restatetheorem}[1]{%
  \let\MYoldlabel\label
  \let\label\MYdummylabel
  \begin{theorem*}[Restatement of \cref{#1}]
    \MYload{#1}
  \end{theorem*}
  \let\label\MYoldlabel
}
\newcommand{\restatelemma}[1]{%
  \let\MYoldlabel\label
  \let\label\MYdummylabel
  \begin{lemma*}[Restatement of \cref{#1}]
    \MYload{#1}
  \end{lemma*}
  \let\label\MYoldlabel
}
\newcommand{\restateprop}[1]{%
  \let\MYoldlabel\label
  \let\label\MYdummylabel
  \begin{proposition*}[Restatement of \cref{#1}]
    \MYload{#1}
  \end{proposition*}
  \let\label\MYoldlabel
}
\newcommand{\restatefact}[1]{%
  \let\MYoldlabel\label
  \let\label\MYdummylabel
  \begin{fact*}[Restatement of \prettyref{#1}]
    \MYload{#1}
  \end{fact*}
  \let\label\MYoldlabel
}
\newcommand{\restate}[1]{%
  \let\MYoldlabel\label
  \let\label\MYdummylabel
  \MYload{#1}
  \let\label\MYoldlabel
}

\allowdisplaybreaks
\newcommand*{\tr}{\mathrm{tr}}

\DeclareMathOperator{\pE}{\widetilde{\mathbb{E}}}

\newcommand{\1}{\bm{1}}

\newcommand{\Proj}{\mathcal{P}}

\def\norm#1{\left\| #1 \right\|}

\def \dtv{d_{\mathsf{TV}}}

\def\tzeta{\tilde{\zeta}}

\newcommand*{\threefrac}[3]{%
  \ensuremath{%
    \vcenter{%
      \halign{\hfil$\,##\,$\hfil\cr
        \scriptstyle{#1}\cr
        \noalign{\kern\threefracLineSep}%
        \hline
        \noalign{\kern\threefracLineSep}%
        \scriptstyle{#2}\cr
        \noalign{\kern\threefracLineSep}%
        \hline
        \noalign{\kern\threefracLineSep}%
        \scriptstyle{#3}\cr
      }%
    }%
  }%
}
\newcommand*{\threefracLineSep}{.4ex}

\title{
  List-decodable Covariance Estimation 
}

\author{
    Misha Ivkov\thanks{Carnegie Mellon University \& Stanford University.} \\ mishai@stanford.edu  \and Pravesh K. Kothari \thanks{Carnegie Mellon University. Supported by NSF CAREER Award \#2047933, a Alfred P. Sloan Research Fellowship and a Google Research Scholar Award.} \\praveshk@cs.cmu.edu
}
\begin{document}

\pagestyle{empty}


\maketitle
\thispagestyle{empty} 


\begin{abstract}
We give the first polynomial time algorithm for \emph{list-decodable covariance estimation}. For any $\alpha > 0$, our algorithm takes input a sample $Y \subseteq \R^d$ of size $n\geq d^{\poly(1/\alpha)}$ obtained by adversarially corrupting an $(1-\alpha)n$ points in an i.i.d. sample $X$ of size $n$ from the Gaussian distribution with unknown mean $\mu_*$ and covariance $\Sigma_*$. In $n^{\poly(1/\alpha)}$ time, it outputs a constant-size list of $k = k(\alpha)= (1/\alpha)^{\poly(1/\alpha)}$ candidate parameters that, with high probability, contains a $(\hat{\mu},\hat{\Sigma})$ such that the total variation distance $TV(\cN(\mu_*,\Sigma_*),\cN(\hat{\mu},\hat{\Sigma}))<1-O_{\alpha}(1)$. This is the statistically strongest notion of distance and implies multiplicative spectral and relative Frobenius distance approximation with dimension independent error. Our algorithm works more generally for $(1-\alpha)$-corruptions of any distribution $D$ that possesses low-degree sum-of-squares certificates of two natural analytic properties: 1) anti-concentration of one-dimensional marginals and 2) hypercontractivity of degree 2 polynomials. 

Prior to our work, the only known results for estimating covariance in the list-decodable setting were for the special cases of list-decodable linear regression and subspace recovery~\cite{DBLP:conf/nips/KarmalkarKK19,RY19,DBLP:conf/soda/BakshiK21,raghavendra2020list}. The best-known algorithms for both these problems only yield a weak recovery guarantee that needs super-polynomial time for any sub-constant (in dimension $d$) target error for the parameters in natural norms. As a corollary, our result yields the first polynomial time \emph{exact} algorithm for list-decodable linear regression and subspace recovery that, in particular, obtain $2^{-\poly(d)}$ error in polynomial-time in the underlying dimension. List-decodable setting also generalizes the problem of robust clustering non-spherical mixtures in the strong contamination model~\cite{bakshi2020mixture,DHKK20} and the state of the art~\cite{bakshi2020mixture} for this latter problem needs $d^{k^{O(k)}}$ samples and tolerates an $\epsilon \ll k^{-O(k)}$ fraction outliers. Our result implies an algorithm with an improved running time and sample bound of $d^{\poly(k)}$ that handles a larger $\epsilon \ll 1/\poly(k)$ fraction of outliers.

\end{abstract}

\clearpage


  \microtypesetup{protrusion=false}
  \setcounter{tocdepth}{1}
  \tableofcontents{}

  \microtypesetup{protrusion=true}

\clearpage

\pagestyle{plain}
\setcounter{page}{1}

\renewcommand{\dist}{\mathsf{parameter\text{-}distance}}
\section{Introduction}
\label{sec:intro}
Can we accurately estimate the mean and covariance of a high-dimensional probability distribution $D$ from a input sample with outliers? What properties of $D$ allow the \emph{robust estimation} of such basic parameters to be statistically and computationally tractable? 



When outliers form a small constant (say $\leq 10\%$) fraction of the input data, we now have a good first-cut understanding of efficient robust estimation of basic parameters of distributions. The works~\cite{DBLP:conf/focs/DiakonikolasKK016,DBLP:conf/focs/LaiRV16} invented the first polynomial time algorithms for the problem with dimension-independent error guarantees and invigorated the now active field of \emph{high-dimensional robust statistics}. The ensuing follow-ups provide optimal guarantees for estimating the mean~\cite{KothariSteinhardt17,DBLP:journals/corr/abs-1711-11581,HopkinsLi17,diakonikolas2018list}, covariance and higher moments~\cite{DBLP:journals/corr/abs-1711-11581} of a broad class of distributions while tolerating a small constant fraction of outliers. This progress has resulted in new broadly applicable techniques, abstracted out properties of the distributions\footnote{In this case, efficiently verifiable certificates of upper bounds on directional moments.} that make efficient robust estimation possible and even inspired progress on related problems such as finding optimal estimators for mean~\cite{MR4102693,DBLP:conf/stoc/CherapanamjeriH20,DBLP:conf/colt/CherapanamjeriF19} and covariance of heavy-tailed distributions.   


In contrast, much less is understood in the setting where a \emph{majority} of the input data are outliers. Since unique recovery of parameters is clearly impossible in this setting, the goal is to compute a dimension-independent constant size \emph{list} of candidate parameters one of which is close to those of the unknown distribution. This model was introduced by Blum, Balcan and Vempala~\cite{DBLP:conf/stoc/BalcanBV08} to study an agnostic variant of clustering\footnote{The ``inliers'' correspond to one of the clusters.} without separation assumptions on the underlying input data. Indeed, list-decodable learning implies clustering algorithms \emph{without any separation assumptions} and allows \emph{partial cluster recovery} even when outliers obliterate multiple clusters completely. 

The recent effort in designing algorithms that tolerate such overwhelming fraction of outliers began with the influential work of Charikar, Steinhardt and Valiant~\cite{DBLP:conf/stoc/CharikarSV17}. In addition to the applications above, they argued that list-decodable learning is a natural model for learning from untrusted data and showed applications to semi-verified learning. Their work gave the first non-trivial guarantees for \emph{list-decodable mean estimation} for distributions with \emph{spherical} covariances (i.e. multiples of identity). Subsequent works obtained stronger guarantees for spherical Gaussians~\cite{diakonikolas2018list} and more generally, Poincaré distributions~\cite{KothariSteinhardt17} with corollaries~\cite{HopkinsLi17,KothariSteinhardt17} to clustering spherical mixtures at the statistically minimum mean separation. A recent sequence of works have even sped-up these results to almost linear time in certain settings~\cite{DBLP:conf/focs/CherapanamjeriM20,DBLP:conf/nips/DiakonikolasKK20,DBLP:journals/corr/abs-2011-09973}. 

\paragraph{List-decodable covariance estimation} Despite this progress on mean estimation, the problem of \emph{covariance estimation} in the list-decodable setting has turned out to be significantly more challenging. Prior works~\cite{DBLP:conf/nips/KarmalkarKK19,RY19} built a  framework via the sum-of-squares method for list-decodable learning to make progress in two special cases: \emph{list-decodable linear regression}~\cite{DBLP:conf/nips/KarmalkarKK19,RY19} (corresponds to the case where the unknown covariance is spherical in a subspace of co-dimension 1) and \emph{list-decodable subspace recovery}~\cite{bakshi2020list,DBLP:conf/soda/BakshiK21} (unknown covariance is spherical in an arbitrary subspace) by introducing the new tool of \emph{certifiable anti-concentration}. However, there is an inherent bottleneck in their approach (see Section~\ref{overview:slack-term}) that leads to significantly weak error guarantees even for the special cases they study: the best known algorithms need \emph{super-polynomial} time for any \emph{sub-constant} error in the underlying dimension and do not appear to extend to settings when the unknown covariance has eigenvalues of different scales\footnote{For e.g., covariances such as $I+\log d \cdot vv^{\top}-uu^{\top}$ for unknown orthogonal unit vectors $u,v$.}. Recent works for the special case of clustering non-spherical Gaussian mixtures~\cite{bakshi2020mixture,DHKK20} managed to wriggle out of the difficulty\footnote{With substantial technical effort -- see Section 2.2 of the overview in~\cite{bakshi2020mixture} for a discussion.} by crucially relying on the input data being generated from a  mixture of $k$ Gaussians where every pair is \emph{separated in total variation distance} and the fraction of outliers is at most $\leq k^{-\poly(k)}$. 

\paragraph{This Work} In this work, we design the first polynomial time algorithm for list-decodable covariance estimation. As immediate corollaries, we also obtain the first polynomial time \emph{exact} algorithms for list-decodable linear regression and subspace recovery in $\R^d$ obtaining constant size lists of candidates one of which achieves as small as $2^{-\poly(d)}$ error in any natural norm and a $d^{\poly(k)}$ time algorithm for clustering non-spherical mixtures in the presence of $\epsilon = 1/\poly(k)$ fraction outliers (the best known prior work needs $d^{k^{\poly(k)}}$ running time and sample complexity and tolerates $\leq k^{-O(k)}$-fraction outliers). 

Our list-decodable covariance estimation algorithm relies on the coalescence of a number of sophisticated tools developed in robust statistics over the past few years. This includes the algorithmic certificates for basic probabilistic phenomenon such as certifiable subgaussianity~\cite{DBLP:journals/corr/abs-1711-11581}, certifiable hypercontractivity~\cite{DBLP:conf/soda/KauersOTZ14,bakshi2020mixture,DHKK20} and certifiable anti-concentration~\cite{DBLP:conf/nips/KarmalkarKK19,RY19} and the sum-of-squares framework for robust statistics. 

The main idea that allows us to finally obtain an algorithm for all covariances is to abandon the previous ``one-shot rounding" approach in related list-decodable learning algorithms and instead settle for a \emph{coarse spectral recovery} guarantee via rounding a sum-of-squares relaxation to obtain a combination of multiplicative approximation for large eigenvalues and additive approximation for small eigenvalues. We then give an iterated pruning method that, that instead of relying on the strong certifiable anti-concentration (the bottleneck in the previous works that holds only for Gaussian-like distributions) property, crucially only needs the significantly milder \emph{Paley-Zygmund} anti-concentration inequality that holds for all subgaussian distributions. Our final algorithm is obtained by an interleaved iteration of coarse spectral recovery, a new ``subgaussian restriction'' subroutine and the pruning procedure based on mild anti-concentration. We expect our algorithmic primitives to be useful in improving robust estimation algorithms that rely on certifiable anti-concentration to faster methods that apply to broader family of distributions.

\subsection{Our Results} 
We now describe our main results in more detail. 
Our results formally hold in the following \emph{strong contamination} model for list-decodable learning. 

\begin{definition}[Strong Contamination Model for List-Decodable Learning] \label{def:strong-contamination}
In the strong contamination model, a $(1-\alpha)$-corrupted sample of size $n$ from a distribution $D$ is generated by choosing an i.i.d. sample $X$ of size $n$ from $D$, and, adversarially switching any $(1-\alpha)n$ points to obtain $Y$.  
\end{definition}
\begin{remark} This is the harshest studied model for robust estimation (and also used in~\cite{DBLP:conf/colt/RaghavendraY20}). It generalizes the more commonly studied list-decodable learning model~\cite{DBLP:conf/stoc/CharikarSV17,DBLP:conf/nips/KarmalkarKK19,RY19,DBLP:conf/soda/BakshiK21,DBLP:journals/corr/abs-2106-09689,DBLP:conf/focs/CherapanamjeriM20,DBLP:conf/nips/DiakonikolasKK20,DBLP:journals/corr/abs-2011-09973} where the input sample $Y$ is obtained by \emph{adding} $(1-\alpha)n$ outliers to an i.i.d. sample of size $\alpha n$. In contrast, our model above allows both \emph{adding and deleting} points from an independent sample of size $n$: our input $Y$ can be generated by first selecting an arbitrarily ``biased'' subset of $\alpha n$ points from an i.i.d. sample $X$ and then adding $(1-\alpha)n$ outliers. Our motivations for working with the harsher model are natural: we'd like to design algorithms that provide strong recovery guarantees under weakest possible modeling assumptions. A concrete advantage of our choice (see Corollary~\ref{cor:clustering-non-spherical}) is that the resulting algorithms (unlike standard list-decodable learning) imply significantly improved algorithms that are more sample efficient, faster, handle larger outlier rates for robust clustering of non-spherical mixtures in the \emph{strong contamination model}. 
\end{remark}

\paragraph{Main Result} Our estimation guarantees are in the following notion of distance between parameters. As we explain, this captures the information-theoretically strongest possible estimation guarantees in our setting.
\begin{definition}[Parameter Distance] \label{def:param-distances}
We say that the distance $\dist((\mu,\Sigma),(\mu',\Sigma'))$ between two sets of mean-covariance pairs is at most $\Delta$ if the following three parameter distance bounds hold. 
\begin{enumerate}
\item \textbf{Mahalanobis Mean Closeness:} $\forall v \in \R^d$,  $\langle \mu - \mu', v \rangle^2 \leq \Delta v^{\top} (\Sigma+\Sigma') v$,
\item \textbf{Multiplicative Spectral Closeness: } $\forall v \in \R^d$, $\frac{1}{\Delta} v^{\top} \Sigma' v \leq v^{\top} \Sigma v \leq \Delta v^{\top}\Sigma' v$, and, 
\item \textbf{Relative-Frobenius Closeness: } $\Norm{ \Sigma^{\dagger/2} \Sigma' \Sigma^{\dagger/2}-I}_F \leq \Delta$.
\end{enumerate}
\end{definition}
If $\dist \leq \Delta$, we can conclude that the total variation distance between the corresponding Gaussians is at most $1-\exp(-\Delta^{O(1)})$ (see Fact~\ref{fact:tv-vs-param-for-gaussians}). As a result, obtaining recovery guarantees in $\dist$ translates into bounds on the total variation error with no dimension dependence. Total variation is the strongest possible (and arguably, the ``right'') notion of distance in this context and is the metric of choice in prior works on robust mean and covariance estimation~\cite{DBLP:conf/focs/DiakonikolasKK016}. Our main result is the following theorem that gives a polynomial time algorithm for list-decodable learning of mean and covariance of an unknown Gaussian distribution. 

\begin{theorem} \label{thm:main-intro-Gaussians}
For every $d \in \N$, there is a $n^{\poly(1/\alpha)}$ time algorithm\footnote{Our algorithm works in the standard word RAM model. The running time of our algorithm is polynomial in the total bit complexity of the input \emph{and of the unknown $\Sigma_*$}. The dependence on the bit-complexity of $\Sigma_*$ is necessary, see Section~\ref{sec:numerical-issues} for a discussion. We note that our algorithm can also be formalized in the idealized ``real RAM'' model~\cite{MR1479636,DBLP:conf/focs/0001HM20} of real computation that is implicitly used in prior works but we choose not to do this.} that takes input a $(1-\alpha)$-corrupted sample of size $n \geq d^{\alpha^{-O(1)}}$ from a $d$-dimensional Gaussian distribution with mean $\mu_*$ and covariance $\Sigma_*$ and outputs a list of $2^{O(1/\alpha^{O(1)})}$-parameters such that with probability at least $0.99$ over the draw of the uncorrupted sample $X$ and the randomness of the algorithm, there is a $(\hat{\mu},\hat{\Sigma})$ in the list satisfying:
$\dist((\hat{\mu},\hat{\Sigma}), (\mu_*,\Sigma_*)) \leq \Delta$,
for $\Delta = \poly(1/\alpha)$. As a corollary, we obtain that $\dtv(\cN(\mu_*,\Sigma_*), \cN(\mu,\Sigma') \leq 1-\exp(-\alpha^{-O(1)})$ where $\dtv(\cdot , \cdot)$ denotes the total variation or statistical distance between two probability distributions. 
\end{theorem}
\begin{remark}
Observe, that our algorithm needs no assumptions on the unknown covariance $\Sigma_*$. In particular, $\Sigma_*$ can have $\exp(d)$ large condition number and can be rank deficient. In fact, if $\Sigma_*$ is singular, our algorithm must construct a candidate $\hat{\Sigma}$ with the \emph{same} range space as $\Sigma_*$ and thus, must recover the low-rank structure in $\Sigma_*$ exactly. For the numerical issues that arise in obtaining this strong guarantee and how we handle them, we direct the reader to the discussion in Section~\ref{sec:preliminaries}.
\end{remark}

Our main algorithmic innovation is a list-decodable learning algorithm for covariance estimation that achieves a \emph{multiplicative spectral approximation} to the unknown covariance with a \emph{dimension-independent} multiplicative factor. Even as a function of $\alpha$, our guarantees are tight\footnote{Given only a $(1-\alpha)$-corrupted sample, we cannot distinguish between the 1-D Gaussians $\cN(0,1)$ and $\cN(0,\alpha^2)$ the variances of which are $1/\alpha^2$ multiplicatively far and $\Omega(1/\alpha)$ additively far.} up to constant factors in the exponent of $\alpha$ in $\Delta$. We then use this estimate to obtain the stronger relative Frobenius distance recovery guarantees. Our result for mean estimation then follows by using our estimates to ``isotropize'' (and thus effectively make the covariance almost spherical) and applying a list-decodable mean estimation algorithm~\cite{KothariSteinhardt17} for covariances of bounded spectral norm. 

\vspace{-3mm}
\paragraph{Running Time and List Size: } For any constant $\alpha$, the running time and sample complexity of our algorithm is polynomial in the underlying dimension. The dependence on $\alpha$ of the running time and sample complexity is exponential. This appears necessary. As we explain below (see Corollary~\ref{cor:clustering-non-spherical}), our list-decodable covariance estimation algorithm implies an algorithm for robustly clustering well-separated (in total variation distance) mixtures of Gaussians. Even in the application to this special-case and to the setting \emph{without any outliers}, known statistical query lower bounds~\cite{DBLP:conf/focs/DiakonikolasKS17,DBLP:journals/corr/abs-2106-09689} suggest a lower bound of $d^{\Omega(1/\alpha)}$ time that matches our guarantees up to the exponent of the polynomial of $1/\alpha$. In terms of the list-size, our algorithm returns a list of size $2^{\poly(1/\alpha)}$. This is a dimension-independent constant but can likely be improved to the optimal $O(1/\alpha)$ bound. 

\vspace{-3mm}
\paragraph{List-decodable learning of all ``reasonable'' distributions} Our algorithm more generally works for any distribution $D$ on $\R^d$ as long as it satisfies two natural analytic properties of probability distributions identified in the context of robust non-spherical clustering in ~\cite{bakshi2020mixture}. Informally speaking, these properties ask for low-degree sum-of-squares certificates of \emph{anti-concentration} and \emph{hypercontractivity of degree 2 polynomials} of the distribution $D$ (we postpone formal definitions to Section~\ref{sec:analytic-props}). While certifiable hypercontractivity of degree 2 polynomials is known to be true for uniform distribution on product domains (such as discrete/solid hypercube), we only have verified certifiable anti-concentration property for rotationally invariant distributions such as Gaussian distributions and affine transforms of uniform distribution on the unit sphere~\cite{DBLP:conf/nips/KarmalkarKK19,RY19,bakshi2020mixture}. Our algorithm thus succeeds as is (and does not require the knowledge of moments of underlying distribution) for all such distributions. We believe that finding natural analytic properties that govern the success of algorithms adds to our understanding of robust estimation in general. 

\begin{theorem}[See Theorem~\ref{thm:list-decodable-cov-mean-estimation-section} for a detailed version] \label{thm:main-intro}
For any $\alpha >0$, there is a $n^{\poly(1/\alpha)}$ time algorithm that takes input a $Y \subseteq \R^d$ of size $n$ and outputs a $\cL$ list of size $2^{\poly(1/\alpha)}$ of estimates $(\hat{\mu},\hat{\Sigma})$ with the following guarantee. Suppose there is an i.i.d. sample $X$ of size $n \geq n_0 = d^{\poly(1/\alpha)}$ from a certifiably $(C,\alpha^3/2C)$-anti-concentrated distribution $D$ with mean $\mu_*$ and covariance $\Sigma_*$ with $C$-certifiably hypercontractive degree 2 polynomials such that $|Y \cap X| = \alpha n$. Then, with probability at least $0.99$ over the randomness of the algorithm, there exists a candidate $(\hat{\mu},\hat{\Sigma})$ in the list $\cL$ such that $\dist( (\hat{\mu},\hat{\Sigma}), (\mu_*,\Sigma_*)) \leq \poly(1/\alpha)$. 
\end{theorem}

As an immediate consequence of our algorithm for list-decodable covariance estimation, we obtain improved guarantees for the previously studied problems of list-decodable linear regression and subspace recovery and clustering non-spherical mixtures. 

\paragraph{Applications to Linear Regression} In list-decodable linear regression, we are given a $(1-\alpha)$-corruption of a system of linear equations $\iprod{x_i, \ell_*} = b_i$ where each $x_i$ is drawn from Gaussian distribution and $\ell_*$ is an unknown unit vector. Introducing the key new tool of \emph{certifiable anti-concentration}, Karmarkar, Klivans and Kothari~\cite{DBLP:conf/nips/KarmalkarKK19} and Raghavendra and Yau~\cite{RY19} gave an algorithm\footnote{We note that these algorithms can handle random additive noise in the equations of variance $\ll \alpha$.} for this problem with a running time of $n^{O(1/(\eta^4 \alpha^4)}$ time to produce a list of size $O(1/\alpha)$~\cite{DBLP:conf/nips/KarmalkarKK19} (the list size is a slightly larger bound of $O(1/\alpha^{\log (1/\alpha)})$ in ~\cite{RY19}) that contains a $\hat{\ell}$ that is $\norm{\hat{\ell}-\ell_*}_2 \leq \eta$. This running time was improved to $n^{O(\log (1/\eta)+1/\alpha^4)}$ (at the cost of a larger list size of $\alpha^{O(\log 1/\eta)}$) via a general \textit{error reduction within SoS} method by Bakshi and Kothari~\cite{bakshi2020list}. Note that both results assume that the covariance of $x_i$s is known to be $I$ and more importantly, for any target sub-constant error $\eta \rightarrow 0$ as $d \rightarrow \infty$, the running time required is super-polynomial. This is in fact the consequence of the recovery guarantees being in a norm weaker than total variation. 

Observe that the coefficients of the uncorrupted set of equations $(x_i,b_i)$ are distributed as $d+1$-dimensional Gaussian with mean $0$ and covariance matrix $I$ restricted to a subspace of co-dimension $1$ -- namely, the one orthogonal to the vector $(\ell_*,-1)$. Thus, list-decoding linear equations above is equivalent to list-decoding the (kernel of) the covariance. Our multiplicative spectral guarantees above for covariance estimation immediately yields an algorithm that can obtain an error as low as $\eta = 2^{-\poly(d)}$ in polynomial time. In fact, our algorithm is \emph{exact} in the sense that the sample complexity does not depend on the target error $\eta$ and the estimation error is entirely because of finite numerical precision in computing the output. In addition, unlike prior works, our algorithm does not need to know the covariance of $x_i$s or the length of the unknown vector $\ell_*$. 

\begin{corollary}[Exact Algorithm for list-decodable linear regression]
For any $\alpha >0$ and target error $\eta$, there is a $n^{\poly(1/\alpha)} \poly \log (1/\eta)$-time algorithm for list-decodable linear regression that succeeds with probability at least $0.99$ whenever $n \geq d^{\poly(1/\alpha)}$ and produces a list of $2^{\poly(1/\alpha)}$ candidate vectors such that there is an $\hat{\ell}$ in the list satisfying $\Norm{\hat{\ell}-\ell_*}_2 \leq \eta$. 
\end{corollary} 

\paragraph{Applications to Subspace Recovery} In list-decodable subspace recovery, we are given $(1-\alpha)n$ corrupted samples from $\cN(0,\Pi)$ where $\Pi$ is a projection matrix to a subspace of $\R^d$. In~\cite{raghavendra2020list}, the authors gave an algorithm that takes such a set of points and in $n^{O(1)}$ time and $d^{O(1)}$ samples, finds a list of constant size containing a candidate $\hat{\Pi}$ such that $\Norm{\hat{\Pi}-\Pi}_F^2 \leq O(1/\alpha^5)$. The work~\cite{DBLP:conf/soda/BakshiK21} obtained the stronger guarantee of $\Norm{\hat{\Pi}-\Pi}_F^2 \leq \eta$ for arbitrarily small $\eta$ in time $n^{O(\log 1/\eta)/\alpha^4}$ time whenever $n \geq d^{O(1/\alpha^4)}$. However, even this improved algorithm requires super-polynomial time to achieve any sub-constant recovery error. 

As a direct corollary of our stronger multiplicative spectral approximation guarantee, we immediately obtain the first \emph{exact} algorithm for list-decodable subspace recovery, that, in particular allows achieving even exponentially small errors in polynomial time.  

\begin{corollary}[Exact Algorithm for list-decodable subspace recovery]
For any $\alpha >0$ and target error $\eta$, there is a $n^{\poly(1/\alpha)} \poly \log (1/\eta)$-time algorithm for list-decodable subspace recovery that succeeds with probability at least $0.99$ whenever $n \geq d^{\poly(1/\alpha)}$ and produces a list of $2^{\poly(1/\alpha)}$ candidate projection matrices such that there is an $\hat{\Pi}$ in the list satisfying $\Norm{\hat{\Pi}-\Pi_*}_F \leq \eta$. 
\end{corollary} 

\paragraph{Applications to Robust Clustering of Non-Spherical Mixtures} Our work immediately improves the best known prior algorithms for robust clustering of non-spherical mixtures in the ``small outlier regime''. Specifically, the goal in this problem is to take input an $\epsilon$-corrupted sample from a mixture of $k$ Gaussians with equal weights with means and covariances $\mu_i,\Sigma_i$ such that each pair is $\Delta$-separated in parameter distance (equivalent to separated on total variation distance as discussed above), and output an estimate $\hat{\mu}_i,\hat{\Sigma}_i$ of the parameters of each component that are close in parameter distance (Definition~\ref{def:param-distances}). Two recent works obtained the first efficient algorithms for solving this problem. Specifically, the algorithm in~\cite{bakshi2020mixture} obtains a $n = d^{k^{O(k)}}$-sample and $n^{k^{O(k)}}$ time algorithm for $\Delta =  k^{\poly(k)}$-separated mixtures to obtain $k^{\poly(k)} \epsilon$-close estimates in parameter distance as long as $\epsilon \ll k^{-O(k)}$. Their algorithm succeeds more generally for mixtures of all ``reasonable distributions'' discussed above.  The work of ~\cite{DHKK20} obtains a $d^{F(k)}$ sample and $n^{F(k)}$ time algorithm that tolerates a $\epsilon \ll 1/F(k)$ of outliers for the same problem when the components are Gaussians with $\Delta = F(k)$-separation where $F(k)$ is at most a $\poly(k)$ size tower of exponentials in $k$. While both algorithms are polynomial time for a fixed $k$, their running times and sample complexity are exponentially larger than the potentially optimal bound of $d^{\poly(k)}$ (that matches the SQ lower bounds in~\cite{DBLP:conf/focs/DiakonikolasKS17}). Progress in obtaining clustering algorithms for non-spherical mixtures is a key component in the recent resolution of the problem of robust learning of mixtures of arbitrary Gaussians~\cite{DBLP:conf/stoc/LiuM21,DBLP:journals/corr/abs-2012-02119}.

By combining our list-decodable covariance estimation algorithm (here, our algorithm running in the strong contamination model of list-decodable learning is important) with a clustering algorithm with known approximate parameters (based on the partial clustering framework of ~\cite{bakshi2020mixture}) and a verification subroutine from \cite{bakshi2020mixture}, we obtain the following improved algorithm on three fronts: 1) the algorithm applies to arbitrary weighted mixtures of Gaussians, 2) handles as large as $\epsilon \leq O(p_{min}/k)$ fraction outliers (note that $\epsilon \ll p_{min}$ is information theoretically necessary, and 3) needs sample and running time scaling as $d^{\poly(1/p_{min})}$ -- $d^{\poly(k)}$ for the equiweighted case. We present a detailed proof sketch in Section~\ref{sec:clustering-non-spherical}. 

\begin{corollary}[Improved Algorithms for Clustering Non-Spherical Mixtures, See Theorem~\ref{thm:clustering-non-spherical-section}] \label{cor:clustering-non-spherical}
Let $d, k \in \N$ and $\epsilon \ll O(p_{min}/k)$. For any $\eta >0$, there is an algorithm that takes input an $\epsilon$-corrupted sample $Y = \{y_1, y_2, \ldots, y_n\} \subseteq \bbQ^d$ drawn from $\sum_i p_i \cN(\mu_i,\Sigma_i)$ for $p_i \geq p_{min}$ for each $i$ and with probability at least $0.99$, outputs estimates $\hat{\mu}_i,\hat{\Sigma}_i$ such that $(\hat{\mu}_i,\hat{\Sigma}_i)$ are $O(k \epsilon)$-close to $\mu_i,\Sigma_i$ for each $i$. The algorithm needs $n \geq n_0 = d^{\poly(k)}/\epsilon^2$ samples and runs in time $n^{\poly(k/\eta)}$. 
\end{corollary}


\subsection{Comparison with Related Work} \label{sec:comparison-with-related-work}
\paragraph{Clustering Well-Separated Non-Spherical Mixtures} List-decodable (mean and) covariance estimation significantly generalizes the problem of \emph{robust clustering} of non-spherical mixture models. In the \emph{robust clustering} problem, the input is an $1-\alpha = \epsilon$-corruption (Definition~\ref{def:strong-contamination}) of an i.i.d. sample from a mixture of $k$ distributions (say Gaussians). By viewing any cluster as ``inliers'' and all the other points as ``outliers'', this corresponds to the setting of $\alpha = 1/2k$ if $\epsilon \leq 1/2k$. Recent works~\cite{bakshi2020mixture,DHKK20} gave an $n^{\poly(k)}$ time  ($n^{f(k)}$ time in $\cite{DHKK20}$ where $f(k)$ is a polynomial size tower of exponentials in $k$) algorithm for clustering equiweighted \emph{non-spherical} mixtures of $k$ Gaussians by relying on the new tools of certifiable anti-concentration~\cite{DBLP:conf/nips/KarmalkarKK19,RY19} and certifiable hypercontractivity of degree 2 polynomials. 

Such clustering algorithms need two crucial assumptions: 1) every pair of the $k$ components is separated in total variation distance by $1-\exp(-k^{O(k)})$ and 2) the fraction of outliers $\epsilon \ll k^{-O(k)}$. In contrast, even when specialized to clustering, our algorithm for list-decodable covariance estimation succeeds (and gives total variation distance guarantee) without any separation assumptions and handles as large as $1-1/k$ fraction outliers -- enough to obliterate all but one clusters. Indeed, this \emph{agnostic} clustering application was the main motivation in the initial work of Balcan, Blum and Vempala~\cite{DBLP:conf/stoc/BalcanBV08} that defined and studied the list-decodable learning model. 

The significantly more general list-decodable setting makes the approach in~\cite{bakshi2020mixture,DHKK20} inapplicable. Let us briefly explain why: the key idea in ~\cite{bakshi2020mixture,DHKK20} is to give a \emph{sum-of-squares proof} that if $w$ indicates a subset $C$ of the input points that satisfy some ``Gaussian-like'' properties, then $C$ cannot simultaneously have a large intersection with two different components. This fact crucially needs (even for its information-theoretic truth) that every pair of components is pairwise well-separated. Indeed, the rounding algorithm in~\cite{bakshi2020mixture,DHKK20} comes up with an approximation to the ground-truth clustering of the input points -- a goal that is not meaningful in the  setting of list-decodable covariance estimation. 

\paragraph{Learning Arbitrary Gaussian Mixtures} Our work is related (but incomparable and complementary, in both results and techniques) to the recent resolution of the problem of \emph{robust learning} of a mixture of $k$-arbitrary Gaussians~\cite{DBLP:conf/stoc/LiuM21,DBLP:journals/corr/abs-2012-02119}. When viewed from our vantage point, these works give a polynomial time algorithm (for any fixed $k$) to learn the parameters of a mixture of $k$ Gaussians given an $\epsilon$-corrupted input sample. The algorithms of~\cite{DBLP:conf/stoc/LiuM21,DBLP:journals/corr/abs-2012-02119} do not need strong separation assumptions but crucially need that the fraction of outliers is small (at most $\sim \exp(-k!)$). In that setting, their algorithm recovers estimates of the components that are close (within some $\epsilon^{F(k)}$ in~\cite{DBLP:journals/corr/abs-2012-02119}) to those of the unknown mixture. On the other hand, our algorithm for list-decodable covariance estimation must handle an overwhelming  $1-\alpha \sim 1-1/k$-fraction outliers and list-decodes to an error guarantee of $1-\theta_{\alpha}$ in total variation distance where $\theta_{\alpha}$ is a function only of $\alpha$ and bounded away from $0$ for all $\alpha>0$. This is essentially the best possible guarantee in our setting as it is statistically impossible to obtain a total variation error $< 1-\alpha$. 

Indeed, our techniques are significantly (and necessarily so) different from those in~\cite{DBLP:conf/stoc/LiuM21,DBLP:journals/corr/abs-2012-02119}. In fact, the algorithms in~\cite{DBLP:conf/stoc/LiuM21,DBLP:journals/corr/abs-2012-02119} use robust clustering algorithms from~\cite{bakshi2020mixture,DHKK20} as a first step with their key new algorithmic components coming \emph{after} the clustering step. We note that using our new list-decodable covariance estimation algorithm in lieu of the clustering algorithm in the first step both simplifies and speeds up that step in their proof. 


\section{Technical Overview}
\label{sec:overview}
In this section, we give a high-level overview of our algorithm and the main ideas that go into improving on the approaches from prior works. 

Let $X \subseteq \R^d$ be an i.i.d. sample from $\cN(\mu_*,\Sigma_*)$. Let $Y \subseteq \R^d$ be obtained by taking any $(1-\alpha)$-corruption (i.e. corrupting an $1-\alpha$ fraction of the points) of $X$. The goal of our algorithm is to take input any such $Y$ and come up with a list of candidate hypotheses $(\hat{\mu}_i,\hat{\Sigma}_i)$ of size some fixed dimension-independent constant, such that there an $i$ satisfying $\dtv(\cN(\mu_*,\Sigma_*), \cN(\hat{\mu},\hat{\Sigma})) < 1-\theta_{\alpha}$ where $\theta_{\alpha}$ is a function only of $\alpha$ bounded away from $0$ for all $\alpha>0$. By Fact~\ref{fact:tv-vs-param-for-gaussians}, it is enough to obtain a list of parameters that contains $(\hat{\mu},\hat{\Sigma})$ satisfying $\dist \leq \poly(1/\alpha)$ (see Definition~\ref{def:param-distances}).

In this overview, we will assume that $\mu_*$ is $0$. This is essentially without the loss of any generality. If $X$ is an i.i.d. sample from $\cN(\mu_*\Sigma_*)$, then, $\frac{x-x'}{\sqrt{2}}$ is distributed as $\cN(0,\Sigma_*)$. Thus, if we start by taking a random matching of $Y$ and applying the scaled difference transform above for pairs in the matching, we can simulate access to a $(1-\alpha^2)$-corrupted sample from $\cN(0,\Sigma_*)$. We will further restrict attention to obtaining mutiplicative spectral guarantee in our estimate -- the key component of our algorithm that requires the introduction of several new ideas. 

\paragraph{The standard approach for list-decodable learning} Let's start with the approach in prior works~\cite{DBLP:conf/nips/KarmalkarKK19,RY19,DBLP:conf/soda/BakshiK21,raghavendra2020list} on list-decodable linear regression and subspace recovery. The algorithms in both those works find and round a solution to the sum-of-squares relaxation of a system of polynomial constraints (see Section~\ref{sec:coarse-spectral-recovery} for the system we use) that encode the task of finding a subset, say $C$, of $Y$ of size $\alpha n$ (indicated by $0$-$1$ variables $w_1, w_2, \ldots, w_n$) that satisfies two relevant properties of Gaussian distributions: anti-concentration and hypercontractivity. In order to impose such properties as constraints, we use the standard (see discussion on succinct representation of constraints in Chapter 4 of ~\cite{TCS-086}) technique (from~\cite{DBLP:journals/corr/abs-1711-11581,HopkinsLi17}) of constraint compression by relying on sum-of-squares proofs.

\paragraph{Analysis by relating total variation distance to parameter distance} To understand the main idea in their analyses, consider a ``real world'' solution to the constraint system. Such a solution is simply a subset $C$ of $Y$ of size $\alpha n$. The conceptual crux of the algorithm in~\cite{DBLP:conf/nips/KarmalkarKK19,RY19,DBLP:conf/soda/BakshiK21,raghavendra2020list} is the following observation: If $|C \cap X| \geq \alpha |Y \cap X|$ -- i.e., the set $C$ intersects the ``inlier'' part of $Y$ that comes from the $X$ in $\alpha$ fraction of its points -- then, the empirical parameters of $C$ must be close to that of $X$. This is a basic statement in statistics that relates a non-trivial bound on the total variation distance (which corresponds to intersection when specialized to uniform distributions on two sets of points) to the closeness of the corresponding parameters. In fact, their analyses can be directly used to infer the following purely information-theoretic result: \emph{if the uniform distributions on $C$ and $X$ are both anti-concentrated and have hypercontractive degree 2 polynomials, and $|C \cap X| \geq \alpha |C|$, then, the covariance of $C$ multiplicatively approximates the covariance of $X$.} The anti-concentration property implies that if $C' \subseteq C$ is a subset of arbitrarily small but fixed constant fraction of $C$, then, the ``variance'' $\E_{x \sim C'} \iprod{x,v}^2$ in any direction $v$ on the subset $C'$ must be within a constant factor of $\E_{x \sim C} \iprod{x,v}^2$ -- the variance in the same direction on the whole subset $C$. Such a property can be used to show \emph{statistical identifiability} of a small list -- we direct the reader to the technical overview sections of~\cite{DBLP:conf/nips/KarmalkarKK19,DBLP:conf/soda/BakshiK21,bakshi2020mixture} that provide an essentially complete proof of such a statistical identifiability result with related discussions. While this is exactly the statement we want, such an argument does not yield an efficient algorithm~\footnote{This, by itself, is not surprising. Statistical identifiability in parameter estimation is often significantly simpler to establish than the task of finding efficient algorithms.}. 

\paragraph{Formalizing identifiability in low-degree sum-of-squares proof system} In order to obtain efficient algorithms, the works above formalize the above information-theoretic reasoning into the \emph{low-degree sum-of-squares proof system}. 

In order to work in the low-degree sum-of-squares proof system, we need to work with sum-of-squares certificates for anti-concentration and hypercontractivity inequalities. Informally speaking, this strategy involves creating Boolean indicator variables $w_1, w_2, \ldots, w_n$ that identify a subset of the input corrupted sample $Y$ of size $n$ and force that the subset of points indicated by $w$ admit SoS certificates of hypercontractivity and anti-concentration (i.e., satisfy the two relevant properties of Gaussian distributions that an ``uncorrupted'' part of $Y$ is promised to satisfy). Now, notice that there can be multiple solutions to this relaxation even for the unrelaxed polynomial formulation since $Y$ could be a disjoint union of $1/\alpha$ different subsets of size about $\alpha n$ such that each of these subsets provide a feasible assignment for $w$s. The solution to the relaxation yields a ``pseudo-distribution'' -- for the sake of exposition in this section, the reader can think of a pseudo-distribution as a probability distribution supported on $w$s that describes subsets of $Y$ that satisfy the two relevant properties of Gaussians we imposed as constraints.   

We must now give a rounding algorithm to take such a pseudo-distribution and produce a small list of parameters, one of which is close to the ground truth. For this goal, we might want to replicate the above information-theoretic strategy and argue that the parameters of $w$s in the pseudo-distribution must be close to that of $X$. Such a statement would of course require that the \emph{on average}, a subset $C$ indicated by $w$s in the support of the pseudo-distribution intersects substantially in $X$ (as otherwise, there's no reason for $C$ to have any information about parameters of $X$). Such a statement does not generically hold for all pseudo-distributions (since we can, in general, have solutions $w$s that are entirely supported on the ``outlier part'' in $Y$). But prior works~\cite{DBLP:conf/nips/KarmalkarKK19,RY19} show that certain ``spread-out-ness'' constraints (formulated as minimizing surrogates for entropy of pseudo-distributions) imply that on average, $C$ indicated by $w$ in the support of the pseudo-distribution does intersect in about $\alpha$ fraction of its points with $X$. 

At this point, we might naturally want to replicate the information-theoretic strategy above that infers closeness of parameters of $C$ and $X$ from large intersection between them. But this creates a major technical difficulty in prior works that while tackled with some effort in special cases with weaker notions of error, prevents applications to general covariances. Let us explain this issue a little more:

A concrete way to analyze the pseudo-distribution is to work the following variables $w$ that allow capturing the intersection of $w$s in the support of the pseudo-distribution with the unknown, uncorrupted $X \cap Y \subseteq Y$.  
Let $w_i' = w_i \cdot \1(x_i = y_i)$. Then, notice that $w_i'$ is the indicator of indices of points in the intersection $C \cap X$. Following on the information-theoretic strategy above, we'd like to argue that the variance of points indicated by $w'$ in any direction $v$ is multiplicatively close to that of $X$ in the same direction. Instead of ``real-world'' anti-concentration, this time, we must use a low-degree sum-of-squares certificate for anti-concentration only. The low-degree sum-of-squares certificate for anti-concentration from~\cite{DBLP:conf/nips/KarmalkarKK19,RY19} allows us to obtain a claim of the following form in degree $O(1/\delta^2)$ (which translates into a running time of $n^{O(1/\delta^2)}$. 

\begin{equation} \label{eq:anti-conc-consequence}
\frac{1}{n} \sum_{i} w_i' \iprod{x_i,v}^2 \geq \delta^2 \Paren{\frac{1}{n} \sum_i w_i' - O(\delta)} \frac{1}{n} \sum_i \iprod{x_i,v}^2\mper
\end{equation}
Informally speaking, the LHS counts the contribution to the variance in the direction $v$ from the points in $C \cap X$. The RHS, on the other hand, is a scaling of the variance of $X$ in the direction $v$ with the major difference from the real world version is the presence of the additive $-O(\delta)$ slack in the right hand side in the multiplier to the variance.

\subsection{Key Bottleneck: Exponential Dependence on Condition Number}
\label{overview:slack-term}
The expression above reveals a ``gap'' between low-degree certificates for anti-concentration inequality vs ``real world'' anti-concentration: the guarantee above is meaningful only when $\frac{1}{n} \sum_i w_i' \gg \delta$. 
Further, the sum-of-squares degree dependence of $O(1/\delta^2)$ for such a certificate happens to be \emph{tight}\footnote{The certificates rely on the univariate polynomial approximators for indicator functions of $\delta$-length interval around $0$ over standard Gaussian distributions. Such a polynomial can be shown to need degree $O(1/\delta^2)$ by standard techniques in approximation theory. See the recent \href{https://simons.berkeley.edu/talks/recent-progress-algorithmic-robust-statistics-sum-squares-method}{talk} for a research direction on potentially stronger certificates that could escape such lower bounds.} -- translating into a running time cost of $n^{O(1/\delta^2)}$.
This might appear innocuous -- after all, if $w$ was indeed an indicator of $Y \cap X$, the associated $\frac{1}{n} \sum_i w_i' \geq \alpha$ so simply choosing $\delta \ll \alpha$ should work. But this is misleading. If we were to analyze, for example, the average mean and covariance under the (pseudo-) distribution, we need that $\frac{1}{n} \sum_i w_i' \geq \delta$ to hold \emph{pointwise} in the support of the distribution. But this is of course not enforceable as a constraint. In ``real world'', we could analyze the distribution by going over all $w$s in the support of it and splitting into two cases depending on whether a given $w$ satisfies the above large intersection condition. But such an argument involves an if-then statement that is hard to formulate as a low-degree sum-of-squares proof (indeed, a version of this argument is precisely what is used in~\cite{bakshi2020mixture} but its success strongly relies on $Y$ being a sample from a Mixture of Gaussians with only a small fraction of outliers) and the source of all the trouble in the list-decodable setting. We remark that this issue of ``pointwise'' facts provable in low-degree sum-of-squares proof system also arises in the recent work~\cite{DBLP:conf/alt/KothariMZ22} on using the sum-of-squares relaxation for robust moment estimation to obtain optimal error guarantees for Gaussian distributions.

In order to get around this issue of additive slack and obtain a meaningful anti-concentration inequality from \eqref{eq:anti-conc-consequence}, we need an \emph{upper-bound} on $\frac{1}{n} \sum_i \iprod{x_i,v}^2$. For example, if we knew and encoded into our constraint system that the unknown covariance $\Sigma_*$ has all its eigenvalues at most $K$, then, we can conclude from \eqref{eq:anti-conc-consequence} that:

\begin{equation} \label{eq:anti-conc-consequence-second}
\frac{1}{n} \sum_{i} w_i \iprod{x_i,v}^2 + O(K \delta) \Norm{v}_2^2 \geq \delta^2 \Paren{\frac{1}{n} \sum_i w_i} \frac{1}{n} \sum_i \iprod{x_i,v}^2\mper
\end{equation}

That is, we must incur an additive \emph{slack} in the anti-concentration inequality that scales with the \emph{largest eigenvalue} of $\Sigma_*$. Observe that in order for this guarantee to be meaningful, we'd have to choose $\delta \ll 1/\sqrt{K}$ and thus our running time must scale \emph{exponentially in the condition number} $\Sigma_*$. 

This is the key reason for the weaker error guarantees in the prior results on list-decodable learning, and in fact prevents any meaningful guarantee (even of the sort known for regression and subspace recovery) without a known bound on the condition number of $\Sigma_*$. In fact, approaches based on prior works appear to fail even if the unknown covariance matrix has eigenvalues of two different scales, such as $I + (\log d) uu^{\top} - vv^{\top}$ for pairwise orthonormal unit vectors $u,v$. 

This appears to be a fundamental issue in using anti-concentration inequalities within the SoS framework. For e.g., in their initial version of the algorithm for clustering non-spherical mixtures~\cite{bakshi2020mixture}, the authors made assumptions on the condition number of the covariances (with running time growing exponentially in $\log \kappa$) of the components in order to obtain their guarantees with substantial effort invested into getting around the issue of additive slack in certificates of anticoncentration~\footnote{See also the discussion on the need for \emph{a priori upper bounds} in ~\cite{DHKK20}.}. They later managed to find an iterative ``bootstrapping'' technique that crucially relied on the strong separation assumptions available in the clustering setting in order to get an assumption-free robust clustering algorithm.  

In the list-decodable setting, there are no such assumptions to work with and as a result the follow-up work~\cite{DBLP:conf/soda/BakshiK21} on list-decodable subspace recovery only obtains the weaker error guarantees discussed before. Our key innovation is a new algorithmic strategy that circumvents the issues with certifiable anti-concentration. Our algorithm is based on iterative use of three algorithmic components (that make an essential use of the \emph{isotropic position}) that may be useful in robust estimation in general. We explain these new components and how they fit together next. 

\subsection{Coarse Spectral Recovery: Theorem~\ref{thm:recover-large-evs}}
Let's depart from the approach above and abandon the goal of recovering multiplicative spectral approximation to $\Sigma_*$ \emph{in one shot} as in the prior works. Instead, we will shoot for a \emph{coarse spectral recovery} algorithm where we obtain 1) multiplicative approximation for the large eigenvalues, and, 2) additive approximation for all small eigenvalues of $\Sigma_*$. 

Of course, the algorithm does not know the scale of the eigenvalues of $\Sigma_*$. So what does ``large'' mean? We will say that a quadratic form on $v$ of $\Sigma_*$ is large if it is at least $\poly(\alpha)$ relative to that of the empirical covariance of the input \emph{corrupted} sample. That is, $v^{\top}\Sigma_*v \geq \poly(\alpha) \E_{y \sim Y} \iprod{v,y}^2$. Observe that $Y$ has $(1-\alpha)n$ outliers that can be arbitrarily large and can completely drown out all eigenvalues of $\Sigma_*$. In that case, coarse spectral recovery will only achieve a vacuous guarantee and we'd have to make progress via a different route. 

To obtain such a ``multiplicative+additive'' guarantee, we introduce a new ``bootstrapping'' technique that obtains low-degree sum-of-squares certificates of \emph{Frobenius norm} error bounds \emph{restricted to the subspace} where $\Sigma_*$ has small eigenvalues. Combining with multiplicative spectral recovery bounds on large eigenvalues then yields the required guarantee. This technique requires the use of \emph{certifiable hypercontractivity} of degree 2 polynomials (in our constraint system) in addition to certifiable anti-concentration. This is in contrast to prior works~\cite{DBLP:conf/soda/RaghavendraY20,DBLP:conf/nips/KarmalkarKK19}) that only need certifiable anti-concentration for such a goal. 

We will make ``conditional progress'' via this method. Specifically, we argue that if every eigenvalue of every candidate $\hat{\Sigma}$ in the list generated in the coarse spectral recovery step is large (compared to that of the corrupted sample), we show that one of them must be a good multiplicative approximation that we desire. As we argued before, the additive loss is a direct consequence of the ``slack'' term in certificates of anticoncentration. In order to make progress, we will rely on subroutines ``outside of the SoS'' system. 

\paragraph{Naive Pruning of $Y$} When there is a candidate $\hat{\Sigma}$ in the list recovered in coarse spectral recovery that does have small eigenvalues, we will make progress by a new pruning step. Here's the key intuition: Suppose that a candidate $\hat{\Sigma}$ has a small eigenvalue in direction $v$. Suppose further that $\hat{\Sigma}$ is a ``good'' candidate (i.e. the one that achieves the multiplicative and additive guarantee w.r.t. the unknown covariance $\Sigma_*$). Then, the variance of points in the inlier $X \cap Y$ (i.e. the intersection with the uncorrupted sample) in the direction $v$ must be small because of $X$ being anticoncentrated in all directions including $v$. Thus, if we \emph{prune} the $y \in Y$ such that the projection in direction $v$ is large (i.e., $\iprod{v,y} \gg v^{\top} \hat{\Sigma}v$), we will remove only a small fraction of inliers $X \cap Y$. Thus, our pruning functions as a biased filter that removes mostly outliers along with a small fraction of inliers. Of course, we do not know which candidate is good a priori, so we simply run the process on every candidate and obtain a different pruned subset of $Y$ for each.  

How do we make progress in this step? Observe that we have no guarantees on $Y$ to begin with. And for all we know, the variance of $Y$ in direction $v$ is large (compared to $v^{\top} \hat{\Sigma}v$) because of say, just a single large $y$. If we were to repeatedly apply coarse spectral recovery with pruning, we'd have run $O(n)$ interleavings and in each such step expand the list by a factor of $\poly(1/\alpha)$ leading to an exponential run time and list-size. 

In order to escape this pitfall, we must find a way to guarantee that when we do try to remove points in $Y$ that are too large, we end up chopping off a constant fraction of $Y$ without hurting the inliers. It turns out that this is true if the low-order empirical moments of $Y$ are \emph{subgaussian}. 

\paragraph{Pruning with Subgaussianity: The Power of Mild Anti-Concentration (see Section~\ref{sec:splitting})} 
Informally, the pruning step removes too few points only if a small fraction of $y$ contribute most of the variance in a given direction $v$. But such ``fat-tails" cannot exist if low-order moments of $Y$ are subgaussian. In fact, we prove that if $Y$ has subgaussian moments of $O(1)$-degree, then, we can prune away a $\poly(\alpha)$ fraction of points in $Y$ while ensuring that for any good candidate $\hat{\Sigma}$ the fraction of inliers in the pruned points is small. Specifically, say $Y$ is normalized to be isotropic: that is $\E_{y \sim Y} yy^{\top} = I$. Then, using the Paley-Zygmund anti-concentration inequality, if the $4th$ moment of $Y$ in the direction $v$ is at most $\Delta$, then, there must be a $\Omega(1/\Delta)$ points in $Y$ such that $|\iprod{v,y}|$ exceeds say $1/2$. On the other hand, if $\hat{\Sigma}$ were a good candidate and additively approximates all small eigenvalues of $\Sigma_*$, then, $X$ itself would be subgaussian and given that $v^{\top}\hat{\Sigma}v$ is small, the fraction of points $x$ such that $|\iprod{v,x}|$ is large should be tiny. Thus, if we pruned away points $y$ where $|\iprod{v,y}|$ is larger than, say, $1/2$, we'd have removed a $\poly(\alpha)$ fraction of points in $Y$ while essentially leaving the inliers untouched. 

Note that in this step, we are using subgaussianity to infer a mild anti-concentration of the \emph{entire corrupted sample}. We have strong anti-concentration guarantees for the inliers (and the original uncorrupted sample) but a priori, no guarantees for $Y$ that includes a majority of outliers. The key gain from our splitting subroutine is making progress \emph{without needing to certify strong anti-concentration properties} -- by resorting to mild anti-concentration that can be inferred only from \emph{moment upper bounds}. Such an anti-concentration, by itself, is not enough to do list-decodable learning, but it's enough for our splitting algorithm to succeed.

Of course, our splitting subroutine relied on the corrupted sample $Y$ being subgaussian -- which, of course, it needn't be. Let us describe how we perform a different pruning to ensure this property. 

\subsection{Subgaussian Restriction: Theorem~\ref{thm:subgaussian-restriction}}
The goal of our subgaussian restriction algorithm is to take input a subset of points $Y$ and prune away some points, if needed, to ensure low-order subgaussianity: $\E_{y \sim Y} \iprod{y,v}^{2t}\leq (C't)^t \Paren{\E_{y \sim Y} \iprod{y,v}^{2}}^t$ for some appropriate constant $C'$ and $t = O(1/\alpha)$. Additionally, we must ensure that the points we prune away are overwhelmingly just the outliers. 

Observe that the requirement above is \emph{linearly invariant} and significantly stronger than the variant (accomplished as a step in list-decodable mean estimation algorithms such as~\cite{KothariSteinhardt17}) where we only want that $\E_{y \sim Y} \iprod{y,v}^{2t}\leq (C't)^t \Norm{v}_2^2$. 

In order to accomplish this goal, we give a new algorithm based on a natural sum-of-squares relaxation that maximizes $\E_{y \sim Y} \iprod{y,v}^{2t}$ over vectors $v \in \R^d$ (after putting $Y$ into isotropic position). Our algorithm uses the solution (if the relaxation is infeasible, we are sure that $Y$ is already subgaussian as we'd want) to the relaxation to effectively assign to each $y \in Y$, the ``weight'' $\pE[\iprod{v,y}^{2t}]$. 

We then use a ``reverse Markov'' inequality to argue that there is a threshold $Z$ such that the set of all $y \in Y$ such that $\pE[\iprod{v,y}^{2t}] \geq Z$ must be significantly larger than what we expect for $X$ and thus, we have found a ``outlier-dominated'' portion of $Y$ that can be pruned away. We then iterate on the resulting pruned $Y$. 

Observe that unlike our splitting algorithm above, we do not guarantee pruning away  non-trivial fraction of points in $Y$. Instead, our idea can be summarized as saying that if the adversary added outliers so as to maintain subgaussianity, then, our basic splitting procedure above makes progress. If not, then, we must make progress in the subgaussian restriction phase.

\subsection{Combining the Subroutines for a Multiplicative Spectral Guarantee: Theorem~\ref{thm:list-decoding-covariance-spectral-guarantee}}
Given the above pieces, here's how we can combine them all together: at all points, our algorithm maintains a list of candidate covariances along with ``current witness sets'' -- these are subsets of $Y$ obtained by applying pruning and splitting steps above assuming that the corresponding candidate was a ``good'' candidate. If a candidate $\hat{\Sigma}$ has all its eigenvalues large compared to that of its current witness, we can label it ``final''. If not, we first apply subgaussian restriction to its witness set, put the resulting $Y'$ in isotropic position and run coarse spectral recovery. This gives a new list of candidates. If any of the candidates has a small eigenvalue, the splitting step above prunes away $\poly(\alpha)$ fraction of the points from its witness set. On the other hand, we know that for a good candidate, we never prune away too many of the inliers. Thus, once the size of the witness set drops below, say, $\alpha n/2$, we can reject that candidate completely. Thus, for a ``good'' initial candidate $\hat{\Sigma}$, we must eventually end up with a ``witness'' set $Y'$ such that $\hat{\Sigma}$ has none of its eigenvalues are small relative to that of the $\E_{y \sim Y'} y'{y'}^{\top}$. Altogether, this gives us an algorithm with a ``recursion depth'' of $1/\poly(\alpha)$ and at ``generation'', we increase the list-size by a factor $1/\poly(\alpha)$. Altogether, the algorithm thus runs in time $n^{\poly(1/\alpha)}$ and produces a list of size $(1/\alpha)^{\poly(1/\alpha)}$.


\section{Preliminaries}
\label{sec:preliminaries}
Throughout this paper, we will use $X$ to denote an i.i.d. (uncorrupted) sample of $n$ points in $\R^d$ and $Y$ to denote its $(1-\alpha)$-corruption. For any finite set $S$ of points, we will  $\E_{s \sim S} f(s)$ to denote the empirical average of $f(s)$ as $s$ varies uniformly over $S$. 

\subsection{Computational Model and Numerical Inputs}  \label{sec:numerical-issues}
Our algorithms work in the standard word RAM model\footnote{Works in statistical learning theory often (and sometimes implicitly) present algorithmic guarantees in the real RAM model~\cite{MR1479636}. This model deems a certain (carefully chosen) list of operations on real numbers to be doable in a single step (this includes arithmetic operations). There are a few different choices considered in prior works for such operations, see the discussion in the paragraph titled "formally modeling real RAM algorithms" on Page 1 in ~\cite{DBLP:conf/focs/0001HM20}. Depending on the choice of the list of such operations, our algorithm can be implemented in this model. However, a blackbox running-time preserving translation from algorithms that work in the real RAM model to the one in the standard word RAM model is not known (see the recent work of Erickson on a smoothed version of such a statement~\cite{DBLP:conf/focs/0001HM20}). So we choose the more direct route of working in the word RAM model in this work.}. We assume that the inputs are rational numbers and the running time  of the algorithm is a function total bit complexity of the representation of the input. This does require a moment's thought since a draw from the standard Gaussian, for e.g., is irrational with probability $1$. In this work, we will assume that we have access to a \emph{bit-oracle} for the input irrational number that can furnish as many bits of precision as our algorithm desires. The complexity of the algorithm grows as the number of bits of precision it demands increases.  

\paragraph{Numerical Issues in Vanilla Covariance Estimation in Total Variation Distance} Following prior works on learning mixtures of Gaussians, our recovery guarantee for parameters of an unknown Gaussian distribution is naturally stated in \emph{total variation} distance -- the strongest possible notion of distance in our context. When the covariance $\Sigma_* \succeq 2^{-\poly(d)}I$, our analysis succeeds even on rational truncation of the inputs with essentially no change. 

When $\Sigma_*$ is singular, however, we need some care\footnote{We thank Sam Hopkins and Daniel Kane for discussions on computational models for statistical learning algorithms that motivated our formalization in this section.} in dealing with numerical issues. To see this, consider the basic task of estimating the covariance of an unknown Gaussian distribution $\cN(0,\Sigma)$ on $\R^d$ from independent samples $x_1, x_2, \ldots, x_n$. This is a basic subroutine (that has nothing to do with robust estimation) in numerical algorithms. Standard matrix concentration results imply that for $n \gg d$, $(1 - O(d/n)) \Sigma \preceq \hat{\Sigma} = \frac{1}{n} \sum_{i \leq n} x_i x_i^{\top} \preceq (1 + O(d/n)) \Sigma$. Such a multiplicative guarantee in Löwner ordering is necessary to obtain any bound $<1$ on the total variation distance $TV(\cN(0,\hat{\Sigma}), \cN(0,\Sigma))$ between the unknown Gaussian distribution and the one we estimate from samples. When implemented in the word RAM model, the samples must be truncated to rationals and as a result, the estimated $\hat{\Sigma}$ will have rational entries. It is easy to construct examples where such a procedure (and in fact any rational $\hat{\Sigma}$) must be maximally far from the true covariance $(i.e. $1$)$ in total variation distance. For example, let $v = (\sqrt{3/5}, \sqrt{2/5}, 0, \ldots, 0)$ and let $\Sigma = I - vv^{\top}$. Then, for every $\hat{\Sigma}$ with rational entries,  the multiplicative guarantee above fails and, in fact, $TV(\cN(0,\hat{\Sigma}), \cN(0,\Sigma)) =1$.

The above example shows that it is provably impossible to output a rational $\wh{\Sigma}$ even for the basic task of estimating covariance from i.i.d. samples if we desire a multiplicative spectral guarantee. However, this also appears somewhat pathological since the ``hardness'' here seems to arise entirely from issues in representing the input (without any role of the algorithm itself).

\paragraph{Our Resolution:} In order to circumvent this issue, we will make an assumption that is arguably weakest possible: we will assume that there \emph{exists} a matrix with rational entries that satisfies the guarantees we want. That is, thus \emph{some} output -- however hard to compute -- could have satisfied the requirements of the algorithm. For covariance estimation above, this is essentially equivalent to assuming that the unknown covariance $\Sigma$ is rational entries. Note that the input sample points will still have irrational entries with probability $1$ and will be truncated to rationals. 

In this case, it is possible to recover the multiplicative spectral guarantees for the basic covariance estimation task. It is easy to prove (see Proposition~\ref{prop:smallest-non-zero-sing}) that the smallest non-zero singular value of a $d \times d$ matrix of $B$-bit rationals is at least $2^{-\poly(Bd)}$. Thus, if our estimate from the truncated samples happens to produce singular values $\ll 2^{-\poly(Bd)}$, we can hope to ``round them down'' and learn the kernel of the unknown covariance. Note of course that for this to be possible, the algorithm needs to know an \emph{a priori} bound on the bit complexity of the unknown $\Sigma_*$ as otherwise there is no way to find the right precision for input truncation. The rounding down step needs some care -- we formally perform it using the lattice basis reduction algorithm of~\cite{MR682664-Lenstra82} (see Section~\ref{sec:bit-complexity}).



\subsection{Sum-of-Squares Preliminaries}
We refer the reader to the monograph~\cite{TCS-086} and the lecture notes~\cite{BarakS16} for a detailed exposition 
of the sum-of-squares method and its usage in average-case algorithm design. A \emph{degree-$\ell$ pseudo-distribution} is a finitely-supported function $D:\R^n \rightarrow \R$ such that $\sum_{x} D(x) = 1$ and $\sum_{x} D(x) f(x)^2 \geq 0$ for every polynomial $f$ of degree at most $\ell/2$. We define the \emph{pseudo-expectation} of a function $f$ on $\R^d$ with respect to a pseudo-distribution $D$, denoted $\pE_{D(x)} f(x)$, as $\pE_{D(x)} f(x) = \sum_{x} D(x) f(x)$. 

The degree-$\ell$ pseudo-moment tensor of a pseudo-distribution $D$ is the tensor $\E_{D(x)} (1,x_1, x_2,\ldots, x_n)^{\otimes \ell}$ with entries corresponding to pseudo-expectations of monomials of degree at most $\ell$ in $x$. The set of all degree-$\ell$ moment tensors of degree $d$ pseudo-distributions is also closed and convex.

\begin{definition}[Constrained pseudo-distributions]
  Let $D$ be a degree-$\ell$ pseudo-distribution over $\R^n$.
  Let $\cA = \{f_1\ge 0, f_2\ge 0, \ldots, f_m\ge 0\}$ be a system of $m$ polynomial inequality constraints.
  We say that \emph{$D$ satisfies the system of constraints $\cA$ at degree $r$} (satisfies it $\eta$-approximately, respectively), if for every $S\subseteq[m]$ and every sum-of-squares polynomial $h$ with $\deg h + \sum_{i\in S} \max\set{\deg f_i,r}$, $\pE_{D} h \cdot \prod _{i\in S}f_i  \ge 0$ ($\pE_{D} h \cdot \prod _{i\in S}f_i \geq \norm{h}_2 \prod_{i \in S} \norm{f_i}_2 $ where $\norm{h}_2$ for any polynomial $h$ is the Euclidean norm of its coefficient vector. We say that $D$ satisfies (similarly for approximately satisfying) $\cA$ (without mentioning degree) if $D$ satisfies $\cA$ at degree $r$.
\end{definition}
\paragraph{Basic Facts about Pseudo-Distributions.}
\begin{fact}[Hölder's Inequality for Pseudo-Distributions] \label{fact:pseudo-expectation-holder}
Let $f,g$ be polynomials of degree at most $d$ in indeterminate $x \in \R^d$. 
Fix $t \in \N$. Then, for any degree $dt$ pseudo-distribution $\tzeta$,
$\pE_{\tzeta}[f^{t-1}g] \leq \paren{\pE_{\tzeta}[f^t]}^{\frac{t-1}{t}} \paren{\pE_{\tzeta}[g^t]}^{1/t}$.
\end{fact}
Observe that the special case of $t =2$ corresponds to the Cauchy-Schwarz inequality.
The following idea of \emph{reweighted} pseudo-distributions follows immediately from definitions and was first formalized  and used in~\cite{DBLP:conf/stoc/BarakKS17}).  
\begin{fact}[Reweightings~\cite{DBLP:conf/stoc/BarakKS17}] \label{fact:reweightings}
Let $D$ be a pseudo-distribution of degree $k$ satisfying a set of polynomial constraints $\cA$ in variable $x$. 
Let $p$ be a sum-of-squares polynomial of degree $t$ such that $\pE[p(x)] \neq 0$.
Let $D'$ be the pseudo-distribution defined so that for any polynomial $f$, $\pE_{D'}[f(x)] = \pE_{D}[ f(x)p(x)]/\pE_{D}[p(x)]$. Then, $D'$ is a pseudo-distribution of degree $k-t$ satisfying $\cA$. 
\end{fact}

\paragraph{Sum-of-squares proofs} A \emph{sum-of-squares proof} that the constraints $\{f_1 \geq 0, \ldots, f_m \geq 0\}$ imply the constraint $\{g \geq 0\}$ consists of  polynomials $(p_S)_{S \subseteq [m]}$ such that $g = \sum_{S \subseteq [m]} p_S \cdot \Pi_{i \in S} f_i$.

We say that this proof has \emph{degree $\ell$} if for every set $S \subseteq [m]$, the polynomial $p_S \Pi_{i \in S} f_i$ has degree at most $\ell$ and write: 
\begin{equation}
  \{f_i \geq 0 \mid i \leq r\} \sststile{\ell}{}\{g \geq 0\}
  \mper
\end{equation}



\begin{fact}[Soundness]
  \label{fact:sos-soundness}
  If $D$ satisfies $\cA$ for a degree-$\ell$ pseudo-distribution $D$ and there exists a sum-of-squares proof $\cA \sststile{r'}{} \cB$, then $D$ satisfies $\cB$ at degree $rr' +r'$.
\end{fact}


\begin{definition}[Total bit complexity of Sum-of-Squares Proofs]
Let $p_1, p_2, \ldots, p_m$ be polynomials in indeterminate $x$ with rational coefficients. 
For a polynomial $p$ with rational coefficients, we say that $\{p_i \geq 0\}$ derives $\{p\geq 0\}$ in degree $k$ and total bit complexity $B$ if $p = \sum_i q_i^2 + \sum_i r_i p_i$ where each $q_i^2,r_i$ are polynomials with rational coefficients of degree at most $k$ and $k-deg(p_i)$ for every $i$, and the total number number of bits required to describe all the coefficients of all the polynomials $q_i, r_i, p_i$ is at most $B$.
\end{definition}

There's an efficient separation oracle for moment tensors of pseudo-distributions that allows approximate optimization  of linear functions of pseudo-moment tensors approximately satisfying constraints. %
The \emph{degree-$\ell$ sum-of-squares algorithm} optimizes over the space of all degree-$\ell$ pseudo-distributions that approximately satisfy a given set of polynomial constraints:

\begin{fact}[Efficient Optimization over Pseudo-distributions \cite{MR939596-Shor87,parrilo2000structured,MR1748764-Nesterov00,MR1846160-Lasserre01}]
Let $\eta>0$. There exist an algorithm that for $n, m\in \N$ runs in time $(n+ m)^{O(\ell)} \poly \log 1/\eta$, takes input an explicitly bounded and satisfiable system of $m$ polynomial constraints $\cA$ in $n$ variables with rational coefficients and outputs a level-$\ell$ pseudo-distribution that satisfies $\cA$ $\eta$-approximately. \label{fact:eff-pseudo-distribution}
\end{fact}

\paragraph{Basic Sum-of-Squares Proofs}
\begin{fact}[Operator norm Bound]
\label{fact:operator_norm}
Let $A$ be a symmetric $d\times d$ matrix with rational entries with numerators and denominators upper-bounded by $2^B$ and $v$ be a vector in $\mathbb{R}^d$. 
Then, for every $\epsilon \geq 0$, 
\[
\sststile{2}{v} \Set{ v^{\top} A v \leq \|A\|_2\|v\|^2_2 + \epsilon}
\]
The total bit complexity of the proof is $\poly(B,d,\log 1/\epsilon)$.
\end{fact}



\begin{fact}[SoS Hölder's Inequality] \label{fact:sos-holder}
Let $f_i,g_i$ for $1 \leq i \leq s$ be indeterminates. 
Let $p$ be an even positive integer. 
Then, 
\[
\sststile{p^2}{f,g} \Set{  \Paren{\frac{1}{s} \sum_{i = 1}^s f_i g_i^{p-1}}^{p} \leq \Paren{\frac{1}{s} \sum_{i = 1}^s f_i^p} \Paren{\frac{1}{s} \sum_{i = 1}^s g_i^p}^{p-1}}\mper
\]
The total bit complexity of the sos proof is $s^{O(p)}$. 
\end{fact}
Observe that using $p = 2$ yields the SoS Cauchy-Schwarz inequality. 

\begin{fact}[SoS Almost Triangle Inequality] \label{fact:sos-almost-triangle}
Let $f_1, f_2, \ldots, f_r$ be indeterminates. Then,
\[
\sststile{2t}{f_1, f_2,\ldots,f_r} \Set{ \Paren{\sum_{i\leq r} f_i}^{2t} \leq r^{2t-1} \Paren{\sum_{i =1}^r f_i^{2t}}}\mper
\]
The total bit complexity of the sos proof is $r^{O(t)}$.
\end{fact}

\begin{fact}[SoS AM-GM Inequality, see Appendix A of~\cite{MR3388192-Barak15}] \label{fact:sos-am-gm}
Let $f_1, f_2,\ldots, f_m$ be indeterminates. Then, 
\[
\Set{f_i \geq 0\mid i \leq m} \sststile{m}{f_1, f_2,\ldots, f_m} \Set{ \Paren{\frac{1}{m} \sum_{i =1}^m f_i }^m \geq \Pi_{i \leq m} f_i} \mper
\]
The total bit complexity of the sos proof is $\exp(O(m))$.
\end{fact}


\begin{fact}[Univariate SoS Proofs]
Let $p = \sum_{i \leq k} \alpha_i x^i$ be a univariate polynomial of degree $k$ with rational coefficients $\alpha_i$ with numerators and denominators upper-bounded by $2^{B}$ for some $B \in \N$. For every $\epsilon \geq 0$:
$\sststile{d}{x} \Set{p(x) \geq 0}$ and the total bit complexity of the SoS proof is upper-bounded by $\poly(B,\log 1/\epsilon)$.
  \label{fact:univariate}
\end{fact}
\begin{fact}[Cancellation within SoS, Constant RHS~\cite{bakshi2020mixture}]
    \label{lem:sos-cancel-basic}
    Suppose $A$ is indeterminate and $t\ge 1$. Then, 
    \[\Set{A^{2t}\le 1}\sststile{2t}{A}\Set{A^2\le 1}\]
    Further, the total bit complexity of the SoS proof is at most $2^{O(t)}$.
\end{fact}

\begin{lemma}[Cancellation within SoS~\cite{bakshi2020mixture}]
\label{lem:sos-cancel}
Suppose $A$ and $C$ are indeterminates and $t\ge 1$. Then,
\[\Set{A\ge 0\cup A^t\le CA^{t-1}}\sststile{2t}{A,C}\Set{A^{2t}\le C^{2t}}.\] Further, the total bit complexity of the SoS proof is at most $2^{O(t)}$.
\end{lemma}


\begin{fact}[Frobenius-Operator Norm Bounds in SoS~\cite{bakshi2020mixture}]
\label{lem:frob-op-norm}
Suppose $A \in \bbQ^{d\times d}$ have entries of bit complexity at most $B$. Let $Q$ be a $d \times d$ matrix valued indeterminate. Then 
\[\sststile{2}{Q}\Set{\Norm{AQ}_F^2\le \Norm{A^{\top}A}_{op}\Norm{Q}_F^2}\]
The total bit complexity of the SoS proof is at most $O(B^2 d^2)$.
\end{fact}







\begin{fact}[Contraction and Frobenius Norms, Lemma 9.1 in~\cite{bakshi2020mixture}] \label{fact:sos-contraction}
Let $A,B$ be $d \times d$ matrix-valued indeterminates. 
Let $\beta$ be a scalar-valued indeterminate. 
Then, 
\[
\Set{ \beta (v^{\top}A^{\top}Av)^t \leq \Delta \norm{v}_2^{2t}} \sststile{}{} \Set{\beta \Norm{AQ}_F^{2t} \leq \Delta t^t \Norm{Q}_F^{2t}}\mcom
\]
and,
\[
\Set{ \beta (v^{\top}A^{\top}Av)^t \leq \Delta \norm{v}_2^{2t}} \sststile{}{} \Set{\beta \Norm{QA}_F^{2t} \leq \Delta t^t \Norm{Q}_F^{2t}}\mper
\]
\end{fact}

\begin{fact}[See Lemma A.5 in \cite{DBLP:journals/corr/abs-1711-11581}]
Let $\cA$ be a set of polynomial equality axioms in variable $x$ such that:
\[
\cA \sststile{x,u}{2t} \Set{p(x,u)\geq 0}\mcom
\]
for a polynomial $p$ with total degree at most $2t$.
Then, for any pseudo-distribution $D$ of degree $2t$ on $x$ satisfying $\cA$,
\[
\sststile{u}{2t} \Set{\pE_{D(x)}p(x,u) \geq 0}\mper
\] 
\label{fact:partial-pseudo-expecting}
\end{fact}

\subsection{Analytic Properties of Probability Distributions} \label{sec:analytic-props}
\paragraph{Certifiable Anti-Concentration} 
\begin{fact}[Univariate Approximator to Interval Indicator (see Lemma A.1 in~\cite{DBLP:conf/nips/KarmalkarKK19})]
For each $\delta >0$, there is a univariate polynomial $p_{\delta}$ and a sum-of-squares polynomial $S_{\delta}$ of degrees $\leq s(\delta) = O(1/\delta^2)$ both with rational coefficients with numerators and denominators upper bounded by $2^{O(s(\delta))}$ satisfying:
\begin{enumerate}
  \item $p_{\delta}(x)= p_{\delta}(-x)$ for every $x$. Thus the non-zero coefficients of $p_{\delta}$ are on even-power monomials in $x$.
  \item $p_{\delta}(x) \geq 1/2$ for all $x$ such that $|x|\leq \delta$.
  \item $\E_{x \sim \cN(0,1)} p_{\delta}^2(x)+ S_{\delta}(x)= C \delta$ for an absolute constant $C>0$.
\end{enumerate}
\end{fact}

\begin{corollary} \label{cor:cert-concentration-props-poly}
Let $\delta>0$, and $x \in \R^d$. Let $R,\Sigma$ be $d \times d$ symmetric matrix-valued indeterminates. 
Let $q_{\delta,\Sigma}(x,v)$ be the following polynomial in $d$-dimensional vector valued indeterminate $v$ (parameterized by $x$). 
\[
q_{\delta, \Sigma}(x, v) = (v^{\top}\Sigma v)^{s(\delta)/2}p_{\delta}(\frac{\iprod{x,v}}{\sqrt{v^{\top}\Sigma v}})
\] 
Then, $q_{\delta, \Sigma}(x, v)$ satisfies:
\begin{enumerate}
  \item $\sststile{2s(\delta)}{\Sigma,v} \Set{\iprod{x,v}^2 (v^{\top} \Sigma v)^{s(\delta)-1} + \delta^2 q_{\delta, \Sigma}(x, v)^2 - \delta^2 (v^{\top}\Sigma v)^{s(\delta)} = SoS(v, R)}$.
  \item $\Set{R^2 = \Sigma} \sststile{2s(\delta)}{v} \Set{\E_{x \sim \N(0,I)} q_{\delta,\Sigma}(R x,v)^2 \leq C\delta \paren{v^{\top}\Sigma v}^{s(\delta)}}$.
\end{enumerate}
Here, $SoS(v,R)$ denotes a sum-of-squares polynomial in indeterminates $v$ and $R$. Further, the total bit complexity of both the SoS proofs above is at most $d^{O(s(\delta))}$. \label{cor:affine-transform-poly}
\end{corollary}
We provide a proof of the above corollary for completeness in Section~\ref{sec:deferred-proofs} of the Appendix.

\begin{definition}[Certifiable Anti-Concentration]
A distribution $D$ on $\R^d$ with mean $0$ and covariance $\Sigma_*$ of rational entries with numerator and denominators upper-bounded by $2^{B}$ is said to be $s(\delta)$-certifiably $(C,\delta)$-anti-concentrated if for $q_{\delta,\Sigma_*}$ defined in Corollary~\ref{cor:affine-transform-poly} satisfies:
\begin{enumerate}
 \item $ \sststile{4s}{v} \Set{\iprod{x,v}^2 (v^{\top} \Sigma_* v)^{s(\delta)-1} + \delta^2 q_{\delta, \Sigma_*} (x, v)^2 - \delta^2 (v^{\top}\Sigma v)^{s(\delta)} = SoS(v, \Pi)}$,
  \item $\sststile{4s}{v} \Set{\E_{x \sim \N(0,I)} q_{\delta,\Sigma_*}(\Pi x,v)^2 \leq C\delta \paren{v^{\top}\Sigma_* v}^{s(\delta)}}$, and
\end{enumerate}
the total bit complexity of each of two SoS proofs above is at most $\poly(B,s(\delta))$. A set $X \subseteq \R^d$ is said to be $s(\delta)$-certifiably $(C,\delta)$-anti-concentrated if the uniform distribution on $X$ is $s(\delta)$-certifiably $(C,\delta)$-anti-concentrated.
\end{definition}



\begin{fact}[Certifiable Anti-concentration of Gaussians and Spherically Symmetric Distributions, Theorem 6.2 in~\cite{DBLP:conf/soda/BakshiK21}]
Gaussian distribution (with arbitrary covariances) and more generally, affine transforms of any spherically symmetric random variable $H$ on $\R^d$ with sub-exponentially distributed $\norm{H}_2^2$ is $s(\delta)$-certifiably $(C,\delta)$-anti-concentrated for $s(\delta) \leq O(1/\delta^2)$ and $C = O(1)$. 
\end{fact}

\paragraph{Certifiable Hypercontractivity of Degree 2 Polynomials} 
Next, we define \emph{certifiable hypercontractivity} of degree-$2$ polynomials that formulates (within SoS)
the fact that higher moments of degree-$2$ polynomials of distributions (such as Gaussians)
can be bounded in terms of appropriate powers of their 2nd moment.

\begin{definition}[Certifiable Hypercontractivity] \label{def:certifiable-hypercontractivity}
An isotropic distribution $\cD$ on $\R^d$ is said to be $h$-certifiably $C$-hypercontractive
if there is a degree-$h$ sum-of-squares proof of the following unconstrained polynomial inequality
in $d \times d$ matrix-valued indeterminate $Q$:
\[ \E_{x \sim \cD} \paren{x^{\top}Qx}^{h} \leq \paren{Ch}^{h} \Paren{\E_{x \sim \cD} \Paren{\paren{x^{\top}Qx}-\E_{x \sim \cD}[x^{\top}Qx]}^2}^{h/2}\mper
\]
A set of points $X \subseteq \R^d$ is said to be  $C$-certifiably hypercontractive
if the uniform distribution on $X$ is  $h$-certifiably $C$-hypercontractive.
\end{definition}
\begin{remark}
Certifiable hypercontractivity is sometimes also defined (such as in~\cite{bakshi2020mixture}) with the RHS above being $h/2$-th power of the 2nd moment of $x^{\top}Qx$ instead of variance as in the above definition. In that case, an additional property (called ``certifiable bounded variance'') is needed to obtain the statement in terms of the variance on the RHS above. We choose the simpler formulation with the RHS above directly stated in terms of the variance of $x^{\top}Qx$. 
\end{remark}

Observe that the definition above is invariant under linear transforms of the the random variable $x$. It can also be shown to be invariant under affine transforms of $x$ (see Lemma 2.3 in~\cite{DBLP:journals/corr/abs-2012-02119}). Hypercontractivity is an important notion in high-dimensional probability and analysis on product spaces~\cite{MR3443800-ODonnell14}. Kauers, O'Donnell, Tan and Zhou~\cite{DBLP:conf/soda/KauersOTZ14} showed certifiable hypercontractivity of Gaussians and more generally product distributions with subgaussian marginals. Certifiable hypercontractivity strictly generalizes the better known \emph{certifiable subgaussianity} property (formalized and studied first in~\cite{DBLP:journals/corr/abs-1711-11581}) that is the special case of certifiable hypercontractivity of (squares of) linear polynomials, or, equivalently, when $Q = vv^{\top}$ for a vector-valued indeterminate $v$.

\paragraph{Analytic Properties Under Sampling} 
The following lemma can be proven via similar, standard techniques as in several prior works~\cite{DBLP:conf/nips/KarmalkarKK19,DBLP:conf/soda/BakshiK21,RY19,raghavendra2020list}.
\begin{fact}[Certifiable Anti-concentration and Hypercontractivity Under Sampling (see for e.g. Section 8 in~\cite{bakshi2020mixture})] \label{fact:sampling}
Let $D$ be a $s(\delta)$-certifiably $(C,\delta)$-anti-concentrated distribution with mean $\mu_*$ and covariance $\Sigma_*$ with $B$ bit rational entries and $2t$-certifiably $C$-hypercontractive degree $2$ polynomials on $\R^d$ for every $t \in \N$. Let $X$ be an i.i.d. sample from $D$ with $n \geq n_0 = O(d^{s(\delta)})$ and let $\tilde{X}$ be obtained by truncating each entry of each $x \in X$ to a rational number of $\poly(Bd)$ bits. Then, with probability at least $1-1/d$ over the draw of $X$, 1) $X$ and $\tilde{X}$ are $s(\delta)$-certifiably $(2C,\delta)$-anti-concentrated with $s(\delta)$-certifiably $2C$-hypercontractive degree $2$ polynomials, 2) $(\E_{x \sim X} x-\mu_*)(\E_{x \sim X} x-\mu_*)^{\top} \preceq 0.01 \Sigma_*$, 3) $\E_{x \sim X} (x-\mu_*)(x-\mu_*)^{\top} \in [0.99,1.01]\Sigma_*$ and 4)$\Norm{\Sigma_*^{\dagger/2} \E_{x\sim x}(x-\mu_*)(x-\mu_*)^{\top}\Sigma_*^{\dagger/2}}_F \leq 0.1$.

Further, if all entries of $\Sigma_*$ are $B$-bit rational numbers, then, all the above facts are true for $\poly(d)$-bit precision truncations of points in an i.i.d. sample $X$. 
\end{fact}

\paragraph{Total Variation vs Parameter Distance for Gaussians} 
The total variation distance (a.k.a. statistical distance) between any two probability density functions $p,q$ on $\R^d$ is defined by $\dtv(p,q) = \frac{1}{2} \int |p(x)-q(x)| dx$. Then, $0 \leq \dtv(p,q) \leq 1$ for all probability density functions $p,q$. 

The following fact relates the total variation distance between a pair of Gaussians and an appropriate notion of distance between their parameters. A relationship of this form was recently proved by~\cite{devroye2018total} but their bounds are only meaningful in the regime where the total variation distance is at most some absolute constant $\ll 1$. Instead, we use the following result established in the recent works on clustering~\cite{bakshi2020mixture,DHKK20} that gives a meaningful parameter distance translation in the regime where the total variation distance is close to but bounded away from $1$. 

\begin{fact}[TV vs Parameter Distance for Gaussians, see Prop. A.1 in~\cite{bakshi2020mixture}] \label{fact:tv-vs-param-for-gaussians}
Fix $\Delta > 0$ and let $\mu,\mu'$ and $\Sigma, \Sigma' \succ 0$ satisfy:
\begin{enumerate}
\item \textbf{Mean Closeness: } for all $v \in \R^d$, $\iprod{\mu - \mu',v}^2_2 \leq \Delta^2 v^{\top} (\Sigma + \Sigma')v$.
\item \textbf{Spectral Closeness: } for all $v \in \R^d$ $\frac{1}{\Delta^2} v^{\top} \Sigma v \leq v^{\top} \Sigma' v \leq \Delta^ 2 v^{\top} \Sigma(r')v$. 
\item \textbf{Relative Frobenius Closeness: } $\Norm{ \Sigma^{\dagger/2} \Sigma' \Sigma^{\dagger/2} -I}_F^2 \leq \Delta^2 \cdot \Norm{ \Sigma^{\dagger} \Sigma'}^2_2$.
\end{enumerate}
Then, $\dtv(\cN(\mu,\Sigma), \cN(\mu',\Sigma')) \leq 1- \exp(-O(\Delta^2 \log \Delta))$.
\end{fact}


\section{Coarse Spectral Recovery} \label{sec:coarse-spectral-recovery}
In this section, we present the first component of our algorithm for list-decodable covariance estimation. This subroutine produces a list of candidate covariances that includes a candidate that multiplicatively approximates the spectrum of the unknown $\Sigma_*$ \emph{when restricted to the subspace with sufficiently large eigenvalues} (compared to that of the corrupted sample) while giving an additive error guarantee on all small eigenvectors of $\Sigma_*$. 

Formally, our algorithm succeeds whenever we are given an $(1-\alpha)$-corruption $Y$ of a \emph{good} set $X$ of points that we define next. Recall that for any finite set $X$, we use the notation $\E_{x \sim X}$ to mean average over uniform draw of $x$ from $X$. 
\begin{definition}[Good Set] \label{def:good-set-coarse}
For $d \in \N$, we say that a subset $X \subseteq \bbQ^d$ is a $(C,\delta)$-good set with mean $\E_{x \sim X} x = \mu_*$ and 2nd moment $\E_{x \sim X} xx^{\top} = \Sigma_*$\footnote{In our application of this subroutine, we can ensure that $X$ is a sample from a mean $0$ distribution. We invite the reader to think of the mean of $X$ to be exactly $0$ in a first reading. In this case, the 2nd moment of $X$ is the covariance of $X$. We continue to use the same notation for covariance and 2nd moment as in the small mean case (that is satisfied whp by a large enough random sample from a zero-mean distribution) the 2nd moment spectrally approximates the covariance within a factor of $1.01$ which is enough for our guarantees for covariance recovery.} if $X$ satisfies the following for $s(\delta) = O(1/\delta^2)$:
\begin{enumerate}
  \item \textbf{Small Mean: } $\mu_*\mu_*^{\top} \preceq 0.1 \E_{x \sim X} (x-\mu_*)(x-\mu_*)^{\top}$.
  \item \textbf{Anti-Concentration: } For $\beta \geq \delta$, $v \in \R^d$, $\Pr_{x \sim X}[ \iprod{x,v}^2 \leq \frac{\beta}{2} v^{\top} \Sigma_* v] \leq \beta$.  
    \item \textbf{Certifiable Anti-Concentration: } $X$ is $s$-certifiably $(C,\delta)$-anti-concentrated for $s=s(\delta)$. 
  \item \textbf{Hypercontractivity: } $X$ has $2s(\delta)$-certifiably $C$-hypercontractive degree $2$ polynomials.
\end{enumerate}
\end{definition}

The following theorem is the main result of this section. 
\begin{theorem}[Coarse Spectral Recovery]
  \label{thm:recover-large-evs}
Let $1 \geq \alpha, \nu > 0$. For every $t\in \N$, there is an algorithm that takes input a collection of $n$ points $Y \subseteq \bbQ^d$ such that $\frac{1}{n} \sum_i y_i y_i^{\top} = (1 \pm 2^{-d}) I$ and outputs a list of positive semidefinite matrices $\widehat{\Sigma_1}, \widehat{\Sigma}_2, \ldots, \widehat{\Sigma}_k \in \bbQ^{d \times d}$ for $k = O(1/\alpha^{t+2}) \cdot \log (1/\nu)$ with the following guarantee: 

For $\delta = \alpha^3/2C$, suppose there is a $(C,\delta)$-\emph{good} set of points $X = \{x_1, x_2, \ldots, x_n\} \subseteq \bbQ^d$ satisfying $\E_{x \sim X} xx^{\top} = \Sigma_*$ such that $|Y \cap X| \geq \alpha n$. Then, with probability at least $1-\nu$ over only the randomness of the algorithm, there is an $i \leq k$ such that:

\begin{align}
\Sigma_* \preceq \hat{\Sigma}_i \preceq  O\left(\frac{1}{\alpha^{6t+18}}\right) \Sigma_* + O(\alpha^{2t-28}) I \mper 
\end{align}
For $t=20$, this gives a list $\cL$ of size $O(1/\alpha^{22})$ containing a candidate $\hat{\Sigma}_i$ satisfying:
\begin{align}
\Sigma_* \preceq \hat{\Sigma}_i \preceq  O\left(\frac{1}{\alpha^{138}}\right) \Sigma_* + O(\alpha^{12}) I \mper 
\end{align}

The algorithm runs in time $(Bn)^{O(1/\alpha^{12})}O(\log 1/\nu)$ where $B$ is the bit complexity of entries of $y_i$s. 
\end{theorem}

\subsection{Algorithm}
Our algorithm approximately solves and rounds a sum-of-squares relaxation of an appropriate polynomial system. Our polynomial constraint system encodes finding a set of $n$ points $X' \subseteq \R^{d\times n}$ such that $X'$ (intended to be variables for $X$) satisfies the properties of the original sample $X$ for some covariance matrix $\Sigma \in \R^{d \times d}$ (intended to be $\Sigma_*$). Our polynomial system has the following indeterminates.
\begin{enumerate}
\item $x_i'$ for $1 \leq i \leq n$: $d$-dimensional vector valued indeterminates forming $X'$. 
\item $R$, $\Sigma$, $U$, $Z$: $d \times d$ matrix-valued indeterminates. Here, $\Sigma$ encodes the empirical covariance matrix of $X$, $R$ stands for a matrix square root of $\Sigma$ and $U$ forces $R$ to be positive semidefinite. 
\item $w_i$ for $1 \leq i \leq n$: scalar indeterminates encoding that $Y$ intersects $X'$ in $\geq \alpha n$ points.
\end{enumerate}

We impose the following constraints on the indeterminates above (categorized for exposition).
\begin{equation}
    \text{Covariance Constraints: $\cA_1$} = 
      \left \{
        \begin{aligned}
          &
          &R
          &= UU^{\top}\\
          &
          &R^2
          &=\Sigma\\
          &
         &\Paren{\frac{8}{\alpha^2} I-\Sigma}
         &= ZZ^{\top}
        \end{aligned}
      \right \}
    \end{equation}

    \begin{equation}
    \text{Subset Constraints: $\cA_2$} = 
      \left \{
        \begin{aligned}
          &\forall i\in [n]
          & w_i^2
          & = w_i\\
          &&
          \textstyle\sum_{i\in[n]} w_i
          &= \alpha n \\
          &\forall i \in [n]
          &w_i(y_i-x'_i)
          &=0
        \end{aligned}
      \right \}
    \end{equation}

    \begin{equation}
    \text{Parameter Constraints: $\cA_3$} = 
      \left \{
        \begin{aligned}
          &
          &\frac{1}{n}\sum_{i = 1}^n w_i x_i' x_i'^{\top}
          &= \Sigma\\
        \end{aligned}
      \right \}
    \end{equation}
\begin{equation}
  \text{Cert. Anti-Concentration : $\cA_4$} =
    \left\{
      \begin{aligned}
        &
        &\frac{1}{n}\sum_{i=1}^n  q_{\delta,\Sigma}^2\left(x_i',v\right)
        &\leq C\delta \Paren{v^{\top}\Sigma v}^{s(\delta)}
       \end{aligned}
      \right\}
   \end{equation}

    \begin{equation}
    \text{Cert. Hypercontractivity : $\cA_5$} = 
      \left \{
        \begin{aligned}
         &\forall j \leq s(\delta),
         &\frac{1}{n} \sum_{i=1}^{n} \Paren{{x_i'}^{\top}Qx_i'- \frac{1}{n} \sum_{i \leq n} x_i'^{\top}Qx_i'}^{2j}
         &\leq (2Cj)^{2j}\Norm{R Q R}_F^{2j}
        \end{aligned}
      \right \}
    \end{equation}

        

\subsection{Algorithm}
We next describe our algorithm. 

\begin{mdframed}
      \begin{algorithm}[List-Decoding for Coarse Spectral Recovery]
        \label{algo:coarse-spectral-recovery}\mbox{}
        \begin{description}
        \item[Given:]
            $Y = \{y_1, y_2, \ldots, y_n\} \subseteq \bbQ^d$ such that $\frac{1}{n}\sum_{i = 1}^n y_i y_i^{\top} = (1 \pm 2^{-d})I$ and $\alpha, \eta > 0$.
        \item[Output:]
          A list $\cL$ of $\cO(1/\alpha^{t+2})$ positive semidefinite matrices in $\bbQ^{d\times d}$ for $t=20$.
        \item[Operation:]\mbox{}
        \begin{enumerate}
          \item For $\delta = \alpha^3/2C$, find a pseudo-distribution $\tilde{\zeta}$ of degree $O(s(\delta)+2t + 1)$ that approximately satisfies the constraint system $\cA$ and minimizes $||\pE [w]||_2$ with an error $\leq 2^{-(Bd)^{O(s(\delta))}}$.
          \item For any multiset $S \subseteq [n]$ of size $2t+1$ such that $\pE[w_S] = \pE[\Pi_{i \in S}w_i] > 0$, let $\tilde{\Sigma}_S = \frac{\pE[w_S\Sigma]}{\pE[w_S]}$.
          \item For $\cO(1/\alpha^{t+2})$ times: 1) pick a multiset $S \subseteq [n]$ of size $2t+1$ with probability proportional to $\pE[w_S]$ and 2) add $O(\frac{1}{\alpha^8}) \tilde{\Sigma}_S$ to $\cL$. 
          \item Return $\cL$.
        \end{enumerate}
        \end{description}
      \end{algorithm}
    \end{mdframed} 

\subsection{Deriving Key Properties Via Low-Degree Sum-of-Squares Proofs}
The goal of the next few lemmas is to derive a key consequence of our constraint system $\cA$. Informally speaking, we show that if $X'$ -- the algorithm's ``guess'' for the unknown $X$ -- intersects with $X$ non-trivially, then, the quadratic form of the empirical 2nd moments of $X$ and $X'$ on any vector $v$ must be close. The closeness is quantified by an error term with a multiplicative part and an additive part.

\begin{notation}Let $w(X') = \frac{1}{n} \sum_{i = 1}^n w_i \cdot \1(x_i = y_i)$ be the linear polynomial in the indeterminates $w_i$s. Note that $w(X')$ measures the fraction of points $X'$ has in common with the (unknown) good set $X$. 
\end{notation}

\begin{lemma}[Spectral Recovery Guarantee -- Lower Bound] \label{lem:spectral-recovery}
  Under the hypothesis of Theorem~\ref{thm:recover-large-evs}, 
  \[
  \cA \sststile{4s}{\Sigma,w,v,X'} \Biggl\{ \frac{1}{\delta^2} (v^{\top}\Sigma v) (v^{\top}\Sigma_* v)^{s-1} + C\delta (v^{\top} \Sigma_* v)^s \geq  w(X') (v^{\top} \Sigma_* v)^s 
    \Biggr\}\mper
  \]
  Further, if entries of $y_i$s have bit complexity $\leq B$, then, the bit-complexity of the SoS proof is $(Bd/\delta)^{O(st)}$. 
  
  \end{lemma}
  

  \begin{proof}
    One of the key properties of $w_i$ is that for \emph{any} polynomial $h(\cdot)$ of degree at most $O(s)$, it is true that 
    \begin{equation}\cA_2 \sststile{O(s)}{w,X'}\Set{w_i\1(x_i = y_i)h(x_i') = w_i\1(x_i = y_i)h(y_i) =  w_i\1(x_i = y_i)h(x_i)}.\label{eq:wixiyi}\end{equation}
  
    So, recalling the first certifiable anti-concentration constraint (Corollary~\ref{cor:cert-concentration-props-poly}) on $x_i$
    \[\sststile{O(s)}{v} \Set{\langle x_i,v\rangle^2 (v^\top \Sigma_\ast v)^{s-1} + \delta^2 q_{\delta,\Sigma_\ast}(x_i,v)^2 \ge \delta^2 (v^\top \Sigma_\ast v)^s}\]
    and applying Equation \eqref{eq:wixiyi} after multiplying by $w_i\1(x_i = y_i)$ gives
    \[\cA_4 \cup \cA_2 \sststile{O(s)}{w,\Sigma,x_i',v} \Set{w_i\1(x_i = y_i)\langle x_i',v\rangle^2 (v^\top \Sigma_\ast v)^{s-1} + \delta^2w_i\1(x_i = y_i)q_{\delta,\Sigma_\ast}(x_i,v)^2 \ge \delta^2w_i\1(x_i=y_i)(v^\top \Sigma_\ast v)^s}.\]
    Using next that $\cA_2 \sststile{2}{w}\Set{w_i\1(x_i = y_i)\le w_i}$ (and that this is also at most $1$) on the left hand side components and averaging over $i$ transforms this equation to
    \[\cA_4\cup \cA_2 \sststile{O(s)}{w,\Sigma,v,X'}\Set{\frac 1n \sum_{i=1}^{n}w_i\langle x_i',v\rangle^2(v^\top \Sigma_\ast v)^{s-1} + \frac{\delta^2}{n}\sum_{i=1}^{n}q_{\delta,\Sigma_\ast}(x_i,v)^2 \ge \frac{\delta^2}n\sum_{i=1}^{n}w_i\1(x_i = y_i)(v^\top \Sigma_\ast v)^s}.\]
    Now we wish to simplify each of these three terms: 
    \begin{itemize}
      \item By definition of $\Sigma$, we have  $\cA_3 \sststile{2}{v}\Set{\frac 1n \sum_{i=1}^{n}w_i\langle x_i',v\rangle^2 = v^\top \Sigma v}$.
      \item By certifiable anti-concentration of the true samples (Corollary~\ref{cor:cert-concentration-props-poly}) we obtain \[\sststile{O(s)}{v}\Set{\frac 1n \sum_{i=1}^{n}q_{\delta,\Sigma}(x_i,v)^2 \le C\delta (v^\top \Sigma_\ast v)^s}.\]
      \item Finally, by definition we have $\frac 1n\sum_{i=1}^{n}w_i\1(x_i=y_i) = w(X')$. 
    \end{itemize}
    Putting these three facts together yields the desired conclusion:
    \begin{equation}\cA_4\cup \cA_2\sststile{O(s)}{w,\Sigma,v,X'}\Set{\frac 1{\delta^2}(v^\top \Sigma v)(v^\top \Sigma_\ast v)^{s-1} + C\delta (v^\top \Sigma_\ast)^s\ge  w(X')(v^\top \Sigma_\ast v)^s}.\label{eq:inter-cdelta-1}\qedhere\end{equation}
  \end{proof}

  Our next lemma proves an upper-bound version of the spectral guarantee. Note that 

\begin{lemma}[Spectral Recovery -- Upper Bound] \label{lem:spectral-recovery-upper-bound}
Under the hypothesis of Theorem~\ref{thm:recover-large-evs}, for any $t \in \N$, 
\begin{equation}
\cA_4 \cup \cA_2\cup \cA_1\sststile{O(st)}{w,X',\Sigma,v} \Biggl\{ w(X')^{2ts} (v^{\top}\Sigma v)^{2s} \leq 2^{2s} \Paren{\frac{1}{\delta^{4s}} (v^{\top}\Sigma_*v)^{2s} + \Paren{\frac{4C^t\delta^t}{\alpha^2}}^{2s} \Norm{v}_2^{4s}}  \Biggr\} \mper 
\end{equation}

Further, if entries of $y_i$s have bit complexity $\leq B$, then, the bit-complexity of the SoS proof is $(Bd/\delta)^{O(st)}$. 
\end{lemma}


\begin{proof}
  We begin similarly to the proof of Lemma~\ref{lem:spectral-recovery}: in particular, by swapping the roles of $x_i,x_i'$ in the proof we obtain
  \begin{equation}\cA_4\cup \cA_2\sststile{O(s)}{w,\Sigma,v,X'}\Set{\frac 1{\delta^2}(v^\top \Sigma_\ast v)(v^\top \Sigma v)^{s-1} \ge  (w(X') - C\delta)(v^\top \Sigma v)^s}.\label{eq:inter-cdelta}\end{equation}
  
Now, note that we would like $w(X')^t$ to appear. To do so, we use that 
\[w(X')^t - (C\delta)^t = (w(X') - C\delta)\left(\sum_{i=0}^{t-1}w(X')^i (C\delta)^{t - 1 - i}\right)\]
and the fact that from $\sststile{2}{w}\Set{0\le w(X')\le 1}$, $C\delta \le \frac 12$ it follows that
\[\sststile{O(t)}{w}\Set{0\le \sum_{i=0}^{t-1}w(X')^i (C\delta)^{t - 1 - i} \le \sum_{i=0}^{t-1}(C\delta)^{t - 1 - i} \le \frac{1}{1-C\delta} \le 2}.\]
Therefore, multiplying both sides of Equation~\eqref{eq:inter-cdelta} by $\sum_{i=0}^{t-1}w(X')^i (C\delta)^{t - 1 - i}$ yields
\begin{multline}\cA_4\cup \cA_2\sststile{O(s+t)}{w,\Sigma,v,X'}\Biggl\{\frac 2{\delta^2}(v^\top \Sigma_\ast v)(v^\top \Sigma v)^{s-1}\ge \left(\sum_{i=0}^{t-1}w(X')^i (C\delta)^{t - 1 - i}\right)\cdot \frac 1{\delta^2}(v^\top \Sigma_\ast v)(v^\top \Sigma v)^{s-1} \\\ge  (w(X')^t - (C\delta)^t)(v^\top \Sigma v)^s\Biggr\}
\end{multline}

We may now rearrange this to
\[\cA_4\cup \cA_2\sststile{O(s+t)}{w,\Sigma,v,X'}\Set{\left(\frac 2{\delta^2}(v^\top \Sigma_\ast v) + (C\delta)^t(v^\top \Sigma v)\right)(v^\top \Sigma v)^{s - 1}\ge w(X')^{t}(v^\top \Sigma v)^s}.\]
Multiplying both sides by $w(X')^{t(s - 1)}$ and applying Cancellation within SoS (Lemma~\ref{lem:sos-cancel}) with $A = w(X')^t (v^\top \Sigma v)$ and $C = \frac{2}{\delta^2}(v^\top \Sigma_\ast v) + (C\delta)^t (v^\top \Sigma v)$ then brings us closer to our end goal by proving
\[\cA_4\cup \cA_2 \sststile{O(st)}{w,\Sigma,v,X'}\Set{w(X')^{2ts}(v^\top \Sigma v)^{2s} \le \left(\frac{2}{\delta^2}(v^\top \Sigma_\ast v) + (C\delta)^t (v^\top \Sigma v)\right)^{2s}}.\]
Now, applying the SoS Almost-Triangle Inequality (Fact~\ref{fact:sos-almost-triangle}) on the right hand side separates these latter terms:
\[\cA_4\cup \cA_2 \sststile{O(st)}{w,\Sigma,v,X'}\Set{w(X')^{2ts}(v^\top \Sigma v)^{2s} \le 2^{2s}\left(\frac{2^{2s}}{\delta^{4s}}(v^\top \Sigma_\ast v)^{2s} + (C\delta)^{2ts} (v^\top \Sigma v)^{2s}\right)}.\]
To finish, recall that $\cA_1\sststile{2}{\Sigma}\Set{\frac 8{\alpha^2}I - \Sigma = ZZ^\top}$ so in particular $v^\top \Sigma v \le \frac{4}{\alpha^2}||v||_2^2$. Therefore, plugging this in gives the final desired bound of
\[\cA_4\cup \cA_2\cup \cA_1 \sststile{O(st)}{w,\Sigma,v,X'}\Set{w(X')^{2ts}(v^\top \Sigma v)^{2s} \le 2^{2s}\left(\frac{2^{2s}}{\delta^{4s}}(v^\top \Sigma_\ast v)^{2s} + \left(\frac{8C^t\delta^t}{\alpha^2}\right)^{2s} ||v||_2^{4s}\right)}.\]
\end{proof}

\paragraph{Bootstrapping Spectral Recovery via Frobenius Recovery} Our spectral recovery guarantees (from previous two lemmas) is actually insufficient to ensure that a rounded candidate multiplicatively approximates \emph{all} the large eigenvalues of the unknown $\Sigma_*$. One of the key innovations in our analysis is a ``boostrapping'' trick that relies on a stronger \emph{Frobenius norm} guarantee \emph{restricted to the subspace of small eigenvalues of $\Sigma_*$} which we prove below. We invite the reader to think of $\Proj$ as the (unknown) projector to the subspace of small eigenvalues of the (unknown) $\Sigma_*$. 



\begin{lemma}[Frobenius Recovery] \label{lem:frob-recovery-subspace} Under the hypothesis of Theorem~\ref{thm:recover-large-evs}, for any projection matrix $\Proj$ to a subspace of $\R^d$, and for any $t \in \N$, 
\begin{align}
\cA \sststile{O(st+s^2)}{w,\Sigma} &\Biggl\{ w(X')^{2t+1}\Norm{\Proj (\Sigma - \Sigma_*)\Proj}_F^{2} \leq O(s^2)  \Paren{\Paren{ \frac{4C^t\delta^t}{\alpha^2}}^{2}  + \frac{4}{\delta^{4}} \Norm{ \Proj\Sigma_*\Proj }_{op}^{2}}\Biggr\} \mper
\end{align}

Further, if entries of $y_i$s have bit complexity $\leq B$, then, the bit-complexity of the SoS proof is $(Bd/\delta)^{O(st)}$. 
\end{lemma}

We will use the following consequence of certifiable hypercontractivity in the proof of Lemma~\ref{lem:frob-recovery-subspace}. 
\begin{lemma}[Frobenius bound] \label{lem:frob-bound-helper-lemma} 
Under the hypothesis of Theorem~\ref{thm:recover-large-evs}, for any $h \in \N$, and $d \times d$ matrix-valued indeterminate $Q$,
\begin{align*}
\cA_5 \sststile{O(h^2)}{Q,w,X'} &\Biggl\{ w(X')^{2h} \Iprod{\Sigma - \Sigma_*,Q}^{2h} \leq w(X')^{2h-1}\cdot 2^{2h} (Ch)^{2h} \Paren{\Norm{\Sigma_*^{1/2} Q \Sigma_*^{1/2}}_F^{2h} +\Norm{R Q R}_F^{2h}} \Biggr\}
\end{align*}
Further, if entries of $y_i$s have bit complexity $\leq B$, then, the bit-complexity of the SoS proof is $(Bd)^{O(h^2)}$.

\end{lemma}

\begin{proof} 
We begin by rewriting
\begin{align*}
  w(X')\langle \Sigma - \Sigma_\ast, Q\rangle &= \frac 1n \sum_{i=1}^{n}w_i\1(x_i = y_i)\langle \Sigma - \Sigma_\ast, Q\rangle\\
  & = \frac 1n \sum_{i=1}^{n}w_i\1(x_i = y_i)\langle \Sigma - x_i' x_i'^\top, Q\rangle + \frac 1n \sum_{i=1}^{n}w_i\1(x_i = y_i)\langle x_i'x_i'^\top - \Sigma_\ast, Q\rangle\\
  &=  \frac 1n \sum_{i=1}^{n}w_i\1(x_i = y_i)\langle \Sigma - x_i' x_i'^\top, Q\rangle + \frac 1n \sum_{i=1}^{n}w_i\1(x_i = y_i)\langle x_ix_i^\top - \Sigma_\ast, Q\rangle
\end{align*}
where in the last step we applied $\cA_2 \sststile{2}{w}\Set{w_i\1(x_i = y_i)\langle x_i'x_i'^\top,Q\rangle = w_i\1(x_i = y_i)\langle x_ix_i^\top,Q\rangle}$ ala the remark at the beginning of Lemma~\ref{lem:spectral-recovery} (passing through $y_i$).

Therefore, by the SoS Almost Triangle Inequality (Fact~\ref{fact:sos-almost-triangle}) it follows that
\begin{equation}w(X')^{2h}\langle \Sigma-\Sigma_\ast,Q\rangle^{2h} \le 2^{2h}\left(\left(\frac 1n \sum_{i=1}^{n}w_i\1(x_i = y_i)\langle \Sigma - x_i' x_i'^\top, Q\rangle\right)^{2h} + \left(\frac 1n \sum_{i=1}^{n}w_i\1(x_i = y_i)\langle x_ix_i^\top - \Sigma_\ast, Q\rangle\right)^{2h}\right)\label{eq:sumbound}\end{equation}
so it suffices to bound each of these terms separately.

Beginning with the first term, we apply SoS \Holder's inequality (Fact~\ref{fact:sos-holder}) with $f_i = \langle \Sigma - x_i'x_i'^\top,Q\rangle$ and $g_i = w_i\1(x_i = y_i)$ (which is idempotent, that is $g_i^{k} = g_i$) to obtain
\begin{multline}\cA\sststile{O(h^2)}{w,\Sigma,X',Q}\Biggl\{\left(\frac 1n \sum_{i=1}^{n}w_i\1(x_i = y_i)\langle \Sigma - x_i' x_i'^\top, Q\rangle\right)^{2h} \le \left(\frac{1}{n}\sum_{i=1}^{n}w_i\1(x_i = y_i)\right)^{2h-1}\left(\frac 1n\sum_{i=1}^{n}\langle x_i'x_i'^\top - \Sigma,Q\rangle^{2h}\right)\\
  \le w(X')^{2h-1}(Ch)^{2h}||RQR||_F^{2h}\Biggr\}
  \label{eq:frob-fake}
\end{multline}
where we used the certifiable hypercontractivity ($\cA_5$) of $X'$.
Similarly, we may bound that
\begin{multline}\cA\sststile{O(h^2)}{w,\Sigma,X',Q}\Biggl\{\left(\frac 1n \sum_{i=1}^{n}w_i\1(x_i = y_i)\langle x_ix_i^\top - \Sigma_\ast, Q\rangle\right)^{2h} \le \left(\frac{1}{n}\sum_{i=1}^{n}w_i\1(x_i = y_i)\right)^{2h-1}\left(\frac 1n\sum_{i=1}^{n}\langle x_i'x_i'^\top - \Sigma_\ast,Q\rangle^{2h}\right)\\
  \le w(X')^{2h-1}(Ch)^{2h}\Norm{\Sigma_\ast^{\frac 12}Q\Sigma_\ast^{\frac 12}}_F^{2h}\Biggr\}
  \label{eq:frob-real}
\end{multline}
by instead using the true anticoncentration of $X$.
Finally, plugging Equations~\eqref{eq:frob-fake} and \eqref{eq:frob-real} into Equation~\eqref{eq:sumbound} yields the desired
\[\cA\sststile{O(h^2)}{w,\Sigma,X',Q}\Set{w(X')^{2h}\langle \Sigma-\Sigma_\ast,Q\rangle^{2h} \le w(X')^{2h-1}2^{2h}(Ch)^{2h}\left(||RQR||_F^{2h} + \Norm{\Sigma_\ast^{\frac 12}Q\Sigma_\ast^{\frac 12}}_F^{2h}\right)}.\]

The bit complexity bound on the SoS proof follows by accounting the bounds for each of the elementary inequalities used in the argument above.
\end{proof}

We now go on to prove Lemma~\ref{lem:frob-recovery-subspace}. 
\begin{proof}[Proof of Lemma~\ref{lem:frob-recovery-subspace}]

From the conclusion of Lemma~\ref{lem:frob-bound-helper-lemma}, setting $h = 2s$ and letting $\Proj$ be fixed gives:
\begin{multline}
\cA_5 \sststile{O(st+s^2)}{\Sigma,R,m} \Biggl\{w(X')^{4s}\Iprod{\Proj (\Sigma - \Sigma_*)\Proj ,Q}^{4s} = w(X')^{4s}\Iprod{\Sigma - \Sigma_*,\Proj Q \Proj}^{4s} \\\leq 2^{4s} w(X')^{4s-1} (2Cs)^{4s} \Paren{\Norm{ \Sigma_*^{1/2} \Proj Q\Proj \Sigma_*^{1/2}}_F^{4s} +\Norm{ R\Proj  Q  \Proj R}_F^{4s} } \Biggr\} \mper \label{eq:main-bound-subspace-f}
\end{multline}

Multiplying throughout by the SoS polynomial $w(X')^{4s(t-1)+1}$ yields:
\begin{multline}
\cA_5 \sststile{O(st+s^2)}{\Sigma,R,m} \Biggl\{w(X')^{4st + 1}\Iprod{\Proj (\Sigma - \Sigma_*)\Proj ,Q}^{4s} = w(X')^{4st+1}\Iprod{\Sigma - \Sigma_*,\Proj Q \Proj}^{4s} \\\leq 2^{4s} w(X')^{4st} (2Cs)^{4s} \Paren{\Norm{ \Sigma_*^{1/2} \Proj Q\Proj \Sigma_*^{1/2}}_F^{4s} +\Norm{ R\Proj  Q  \Proj R}_F^{4s} } \Biggr\} \mper \label{eq:main-bound-subspace}
\end{multline}


As before, let's analyze the two terms on the RHS separately. 

Using cyclic properties of Frobenius norm, we know that $||\Sigma_\ast^{1/2} \Proj Q\Proj \Sigma_\ast^{1/2}||_F^2 = ||\Proj \Sigma_\ast \Proj Q||_F^2$. Therefore, applying Lemma~\ref{lem:frob-op-norm}, we have that:
\begin{equation}
\cA\sststile{O(st+s^2)}{\Sigma,w} \Set{w(X')^{4st} \Norm{  \Sigma_*^{1/2}\Proj Q \Proj  \Sigma_*^{1/2}}_F^{4s} \leq w(X')^{4st} \Paren{\Norm{ \Proj \Sigma_* \Proj }_{op}^{4s} \Norm{Q}_F^{4s}}}\mper \label{eq:final-bound-subspace-1}
\end{equation}

For the second term, we start from the guarantee of Lemma~\ref{lem:spectral-recovery-upper-bound} applied to the vector $\Proj v$:
\begin{multline}
\cA \sststile{O(st+s^2)}{R,w} \Biggl\{ w(X')^{2st} (v^{\top}  \Proj \Sigma \Proj  v)^{2s} \leq 2^{2s}\Paren{\Paren{ \frac{8C^t\delta^t}{\alpha^2}}^{2s}||v||_2^{4s}  + \frac{2^{2s}}{\delta^{4s}} (v^\top \Proj^\top  \Sigma_\ast\Proj v)^{2s}}\\
\le 2^{2s}||v||_2^{4s}\left(\Paren{\frac{4C^t\delta^t}{\alpha^2}}^{2s} + \frac{2^{2s}}{\delta^{4s}}||\Proj \Sigma_\ast \Proj||_{op}^{2s}\right)
\Biggr\}\mper
\end{multline}

From Contraction within SoS (Fact~\ref{fact:sos-contraction}), we may derive that
\[\Set{\beta(v^\top A^\top Av)^t \le \Delta ||v||_2^{2t}}\sststile{}{} \Set{\beta^2 ||AQA^\top||_F^{2t} \le \Delta^2 t^{2t} ||Q||_F^{2t}}.\]

Applying this with $\beta = w(X')^{2st}$, $A = R \Proj$, $t = 2s$, and $\Delta = 2^{2s} \Paren{\Paren{ \frac{8C^t\delta^t}{\alpha^2}}^{2s}  + \frac{2^{2s}}{\delta^{4s}} \Norm{  \Proj \Sigma_* \Proj }_{op}^{2s}}$ yields
\begin{equation}\label{eq:final-bound-subspace-2}
  \cA \sststile{O(st+s^2)}{R,w,Q} \Set{ w(X')^{4st} \Norm{R\Proj  Q\Proj  R}_F^{4s} \leq 2\cdot (4s)^{4s} \Paren{\Paren{ \frac{8C^t\delta^t}{\alpha^2}}^{4s}  + \frac{2^{4s}}{\delta^{8s}} \Norm{  \Proj\Sigma_* \Proj }_{op}^{4s}}\Norm{Q}_F^{4s}  }
  \end{equation} 
  where we also crucially used an application of the SoS Almost-Triangle Inequality to expand $\Delta^2$.


Plugging back the estimates from \eqref{eq:final-bound-subspace-1} and \eqref{eq:final-bound-subspace-2} into \eqref{eq:main-bound-subspace} gives:
\begin{equation}
\cA_5 \sststile{O(st+s^2)}{\Sigma,w,Q} \Biggl\{ w(X')^{4st + 1}\Iprod{\Proj(\Sigma - \Sigma_*)\Proj,Q}^{4s} \leq (O(s))^{4s} \Paren{\Paren{ \frac{8C^t\delta^t}{\alpha^2}}^{4s}  + \frac{2^{4s}}{\delta^{8s}} \Norm{ \Proj \Sigma_*\Proj }_{op}^{4s}}\Norm{Q}_F^{4s}
 \Biggr\} \mper 
\end{equation}

Note that as $0\le w(X')\le 1$ we may shrink the LHS by multiplying it through further by $w(X')^{2s-1}$.
Substituting $Q = \Proj (\Sigma - \Sigma_*)\Proj$  and multiplying throughout by the SoS polynomial $w(X')^{4st+2s}$ now yields:
\begin{equation}
\cA_5 \sststile{O(st+s^2)}{\Sigma,w,Q} \Biggl\{ w(X')^{8st+4s}\Norm{\Proj(\Sigma - \Sigma_*)\Proj}_F^{8s} \leq w(X')^{4st+2s} (O(s))^{4s} \Paren{\Paren{ \frac{8C^t\delta^t}{\alpha^2}}^{4s}  + \frac{2^{4s}}{\delta^{8s}} \Norm{ P\Sigma_*P }_{op}^{4s}}\Norm{\Proj(\Sigma - \Sigma_*)\Proj}_F^{4s}
 \Biggr\} \mper 
\end{equation}




We now apply Lemma~\ref{lem:sos-cancel} (Cancellation within SoS) with $A = w(X')^{4st+2s} \Norm{\Proj(\Sigma-\Sigma_*)\Proj}_F^{4s}$ and the SoS Almost-Triangle Inequality to obtain that:

\begin{align}
\cA \sststile{O(st+s^2)}{w,\Sigma} &\Biggl\{ w(X')^{16st+8s}\Norm{\Proj(\Sigma - \Sigma_*)\Proj}_F^{16s} \leq (O(s))^{16s}  \Paren{\Paren{ \frac{8C^t\delta^t}{\alpha^2}}^{16s}  + \frac{2^{16s}}{\delta^{32s}} \Norm{ \Proj\Sigma_*\Proj }_{op}^{16s}} \Biggr\} \mper
\end{align}

We finally apply Cancellation with Constant RHS (Lemma~\ref{lem:sos-cancel-basic}) to conclude that:
\begin{align}
\cA \sststile{O(st+s^2)}{w,\Sigma} &\Biggl\{ w(X')^{2t+1}\Norm{\Proj(\Sigma - \Sigma_*)\Proj}_F^{2} \leq (O(s))^{2}  \Paren{\Paren{ \frac{8C^t\delta^t}{\alpha^2}}^{2}  + \frac{4}{\delta^{4}} \Norm{ \Proj\Sigma_*\Proj }_{op}^{2}}\Biggr\} \mper
\end{align}


\end{proof}

\subsection{Analysis of Rounding} 
We first show that setting $X'=X$ (and using the naturally induced assignments for other indeterminates) yields a feasible solution for the relaxation. 
\begin{lemma}[Feasibility of the Relaxation] \label{lem:feasibility-coarse-spectral-recovery}
For $\delta \le \frac{\alpha^2}{2C}$, suppose there exists a $(C,\delta)$-good set $X \subseteq \R^d$ of size $n$ such that $x_i = y_i$ for $\alpha n$ different $i \leq n$. Then, there is a pseudo-distribution of degree $O(s(\delta)+q)$ consistent with the constraint system $\cA$ and as a consequence, Step 1 of Algorithm~\ref{algo:coarse-spectral-recovery} succeeds. 
\end{lemma}
\begin{proof}
We give a solution to the constraint system $\cA$ to prove that the constraints are satisfiable. We set $X' = X$ and $w_i = 1$ if and only if $x_i = y_i$. Since $X$ is a $(C,\delta)$-good set, we immediately obtained that $\cA_2, \cA_3, \cA_4,\cA_5$ are satisfied. Therefore, it remains to check $\cA_1$.

For $\cA_1$, we set $R = \Sigma_*^{1/2}$ -- the PSD square root of $\Sigma_*$ and $U$ to be a matrix such that $UU^{\top} = \Sigma_*^{1/2}$ (which exists). Let us now prove that there is a $Z$ such that $\cA_1$ is feasible. To do so, we make use of the isotropic position of $Y$. 

To begin, note via $|X\cap Y| = \alpha n$ that
\[\E_{x\sim X\cap Y}xx^\top = \frac 1{|X\cap Y|}\sum_{x\in X\cap Y}xx^\top = \frac 1\alpha \left(\frac 1n \sum_{x\in X\cap Y}xx^\top\right)\preceq \frac 1\alpha \left(\frac 1n \sum_{x\in Y}xx^\top\right) = \frac 1\alpha I.\]

Now, recall that true anti-concentration of $X$ tells us that $\Pr_{x\sim X}[\langle x,v\rangle^2 \le \frac {\beta}2 v^\top \Sigma_\ast v] \le \beta$. Therefore, at most $\frac \alpha 2n$ points in $X$ have $\langle x,v\rangle^2 \le \frac {\alpha}{4}v^\top \Sigma^\ast v$, and hence at least $\alpha n - \frac{\alpha}2n = \frac{\alpha}2n$ points in $X\cap Y$ have $\langle x,v\rangle^2 \ge \frac {\alpha}4v^\top \Sigma^\ast v$. Therefore, this implies that
\[\E_{x\sim X\cap Y}\langle x,v\rangle^2 = \frac 1{|X\cap Y|}\sum_{x\in X\cap Y}\langle x,v\rangle^2 \ge \frac{1}{\alpha n}\cdot \left(\frac{\alpha n}{2}\cdot \frac{\alpha}{4}v^\top \Sigma_\ast v\right) = \frac{\alpha}{8}v^\top \Sigma_\ast v\]
and thus we have $\frac{\alpha}{8}\Sigma_\ast \preceq \frac{1}{\alpha}I$. Rearranging gives that $\frac{8}{\alpha^2}I - \Sigma_\ast=ZZ^\top$ for some $Z$ and hence we have feasibility of $\cA_1$. 

This completes the proof. 
\end{proof}


Next, we prove that in expectation, $w(X')$ -- the normalized intersection of $X'$ and the (unknown) sample $X$ must be at least $\alpha^2$. This is a consequence of our relaxation hunting for a maximum entropy (or more precisely, max ``collision probability'') solution. 

\begin{lemma}[High-entropy Pseudo-distributions Intersect $X$] \label{lem:high-entropy-pseudo-distributions}
Under the hypothesis of Theorem~\ref{thm:recover-large-evs}, let $\tzeta$ be a pseudo-distribution of degree $\geq O(s(\delta))$ satisfying $\cA$ that minimizes $\Norm{\pE_{\tzeta} w}_2$.
  Then, $\pE_{\tzeta}[ w(X')]\geq \alpha^2$.

  \end{lemma}
  \begin{proof}
Towards a contradiction, we will show that if the conclusion does not hold then there exists a pseudo-distribution with smaller value of $\Norm{\pE_{\tzeta} w}_2$. 

Suppose $\tau = \frac 1{\alpha n}\sum_{i=1}^{n}\pE_{\tzeta}[w_i\1(x_i = y_i)] =  \frac 1{\alpha n}\sum_{i=1}^{n}\pE_{\tzeta}[w_i] \cdot \1(x_i = y_i)  < \alpha$ where we used that $\1(x_i = y_i)$ is a constant (as opposed to an indeterminate in our polynomial program). To show a contradiction, we will exhibit a pseudodistribution $\nu$ that is 1) feasible for our SoS relaxation and 2) has a smaller value of $\Norm{\pE_{\nu}[w]}_2$.
  
    Toward this, notice that we can write 
    \begin{align}
      \Norm{\frac 1{\alpha n}\pE_{\tzeta}[w]}_2^2  = \frac 1{\alpha^2 n^2}\sum_{i=1}^{n}\pE_{\tzeta}[w_i]^2 = \frac 1{\alpha^2 n^2}\sum_{i=1}^{n}\Paren{\pE_{\tzeta}[w_i]\1(x_i = y_i)^2}^2 + \frac 1{\alpha^2 n^2}\sum_{i=1}^{n}\Paren{\pE_{\tzeta}[w_i](1 - \1(x_i = y_i))^2}^2\label{eq:pe-inlier-outlier}
    \end{align}
    where we used that $\tzeta$ satisfies $w_i^2=w_i$ for every $i$.
  
    By Cauchy Schwarz inequality for pseudo-distributions, we observe  
    \begin{align}
      \tau^2 &= \Paren{\frac 1{\alpha n}\sum_{i=1}^{n}\pE_{\tzeta}[w_i]\1(x_i = y_i)^2}^2 \nonumber \\
      &\le \Paren{\frac 1{\alpha^2 n^2}\sum_{i=1}^{n}\Paren{\pE_{\tzeta}[w_i]\1(x_i = y_i)}^2}\Paren{\sum_{i=1}^{n}\1(x_i = y_i)^2} \nonumber\\
      &= \alpha n\Paren{\frac 1{\alpha^2 n^2}\sum_{i=1}^{n}\Paren{\pE_{\tzeta}[w_i]\1(x_i = y_i)}^2} \label{eq:pe-inlier-lower-bound}
    \end{align}
    and similarly 
    \begin{align}
      (1-\tau)^2 &= \Paren{\frac 1{\alpha n}\sum_{i=1}^{n}\pE_{\tzeta}[w_i](1-\1(x_i = y_i))^2}^2 \nonumber\\
      &\le \Paren{\frac 1{\alpha^2 n^2}\sum_{i=1}^{n}\Paren{\pE_{\tzeta}[w_i](1-\1(x_i = y_i))}^2}\Paren{\sum_{i=1}^{n}(1-\1(x_i = y_i))^2} \nonumber\\
      &= (1 - \alpha)n\Paren{\frac 1{\alpha^2 n^2}\sum_{i=1}^{n}\Paren{\pE_{\tzeta}[w_i](1-\1(x_i = y_i))}^2}. \label{eq:pe-outlier-lower-bound}
    \end{align}
    We can apply Equations~\ref{eq:pe-inlier-lower-bound} and~\ref{eq:pe-outlier-lower-bound} to~\ref{eq:pe-inlier-outlier} to obtain 
    \begin{align}\Norm{\frac 1{\alpha n}\pE_{\tzeta}[w]}_2^2 \ge \frac{\tau^2}{\alpha n} + \frac{(1 - \tau)^2}{(1 - \alpha)n} = \frac 1{\alpha n}\Paren{\tau^2 + (1 - \tau)^2\frac{\alpha}{1 - \alpha}}.\label{eq:pe-lower-bound}
    \end{align}
  
    Now we proceed to construct our $\nu$ achieving lower norm. Let $\zeta^\ast$ be the (true) distribution supported on a single point $w$, where $w_i = \1(x_i = y_i)$ ($i$ is an inlier). Due to $X$ being a good sample, it follows that $\zeta^\ast$ satisfies the constraint system.
    
    Then, we let $\nu = (\alpha - \tau)\zeta^\ast + (1 + \tau - \alpha)\tzeta$, and set $\lambda = \alpha - \tau$. Note that as $0\le \tau< \alpha< 1$, this ``mix'' is indeed a pseudodistribution. From here, notice that 
    \begin{align*}
      \Norm{\frac 1{\alpha n}\pE_{\nu}[w]}_2^2 &= \frac 1{\alpha^2 n^2}\sum_{i=1}^{n}\Paren{\lambda\pE_{\zeta^\ast}[w_i] + (1 - \lambda)\pE_{\tzeta}[w_i]}^2\\
      &= \frac 1{\alpha^2 n^2}\sum_{i=1}^{n}\Brac{\lambda^2\pE_{\zeta^\ast}[w_i]^2 + (1 - \lambda)^2\pE_{\tzeta}[w_i]^2 + 2\lambda(1 - \lambda)\pE_{\zeta^\ast}[w_i]\pE_{\tzeta}[w_i]}\\
      &= \frac{\lambda^2}{\alpha n} + (1 - \lambda^2)\Norm{\frac 1{\alpha n}\pE_{\tzeta}[w]}_2^2 + 2\lambda(1 - \lambda)\frac{\tau}{\alpha n}.
    \end{align*}
  
    Therefore, it follows that
    \begin{align*}\Norm{\frac 1{\alpha n}\pE_{\tzeta}[w]}_2^2 - \Norm{\frac 1{\alpha n}\pE_{\nu}[w]}_2^2 &= (2\lambda - \lambda^2)\Norm{\frac 1{\alpha n}\pE_{\tzeta}[w]}_2^2 - \frac{\lambda^2}{\alpha n} - 2\lambda(1 - \lambda)\frac{\tau}{\alpha n}\\
    &\ge \frac{\lambda}{\alpha n}\Paren{(2 - \lambda)\Paren{\tau^2 + (1 - \tau)^2\frac{\alpha}{1 - \alpha}} - \lambda - 2\tau(1 - \lambda)}.
    \end{align*}
    Plugging in $\lambda = \alpha - \tau$ yields 
    \[\Norm{\frac 1{\alpha n}\pE_{\tzeta}[w]}_2^2 - \Norm{\frac 1{\alpha n}\pE_{\nu}[w]}_2^2 \ge \frac{(1 - \tau^2)(\alpha - \tau)^2}{(1 - \alpha)\alpha n} > 0\]
    when $0\le \tau < \alpha < 1$. Hence, we have found a pseudodistribution $\nu$ satisfying the constraints and of lower norm, contradiction. Therefore it must be the case that $\tau \ge \alpha$.
  \end{proof}

We will now use the above fact along with our spectral and Frobenius recovery guarantees to derive properties of the rounded solutions. Our rounding prcedure depends on sampling multisets.

\begin{notation}[Sampling from Pseudo-distribution]
  Let $k\in \N$ and $\tzeta$ be a pseudodistribution. Consider the distribution $D(k,\tzeta)$ on multisets $S\subseteq [n]$ of size $k$ chosen via the following process:
  \begin{enumerate}
    \item Pick $S = (i_1,i_2,\ldots,i_k)\subseteq [n]$ with probability $\pE_{\tzeta}[\prod_{j\in S}w_j] / (\alpha n)^k$.
    \item Output $S$. 
  \end{enumerate}

  Further define the conditional distribution $D'(k,\tzeta) = D(k,\tzeta)\mid x_i = y_i \forall i\in S$.
\end{notation}

Note that $\sum_{|S| = k}\pE_{\tzeta}[\prod_{j\in S}w_j] = (\alpha n)^k$ to show that this is a well defined probability distribution.

\begin{lemma}[Analysis of Rounding, Lower-Bound]
  \label{lem:analysis-lower-bound}
Assume the hypothesis of Theorem~\ref{thm:recover-large-evs}.
Let $\tzeta$ be a pseudo-distribution of degree $\ge O(s(\delta))$ consistent with $\cA$ minimizing $\norm{\pE_{\tzeta} w}_2$ and fix $\delta = \alpha^3/2C \le \alpha^2/2C$. 
Then, 
\[
\Pr_{S \sim D'(2t+1,\tzeta)} \left[\Paren{\frac{16C^2}{\alpha^{8}}} \frac{\pE_{\tzeta}[w_S \Sigma]}{\pE_{\tzeta}[w_S]} \succeq \Sigma_*\right] \geq \alpha^2/4 \mper
\]

\end{lemma}
  \begin{proof}
We know that for any $S \subseteq [n]$, $\pE[w_{S} \Sigma] = \pE[w_S RR^{\top}] \succeq 0$. Next, from Lemma~\ref{lem:spectral-recovery} we have that
    \[\cA\sststile{4s}{\Sigma,w}\Set{(v^\top \Sigma v)(v^\top \Sigma_\ast v)^{s-1} \ge \delta^2(w(X') - C\delta)(v^\top \Sigma_\ast v)^s}.\]
Let's multiply both sides by the SoS polynomial $w_{S}=w_S^2$ and take pseudo-expectations with respect to $\tzeta$.  Then, by Fact~\ref{fact:partial-pseudo-expecting} acting on $w_S,\Sigma,X'$, we must have that \textbf{for every $v \in \R^d$},  

\[
(v^\top \pE[w_S \Sigma]v)(v^\top \Sigma_\ast v)^{s-1} \ge \delta^2 \Paren{\pE[w(X')w_S] - C\delta\pE[w_S]} \Paren{v^\top \Sigma_\ast v}^s\mper
\]
  
For any $S$ such that $\pE[w_{S}] > 0$, let $\hat{\Sigma}_S = \pE[w_S \Sigma]/\pE[w_S]$. Then, dividing through by $\pE[w_S]$ yields that for every $v \in \R^d$, we have
    \begin{equation}(v^\top\hat{\Sigma}_S v)(v^\top \Sigma_\ast v)^{s-1}\ge \delta^2 \Paren{\frac{\pE[w(X')w_S]}{\pE[w_S]} - C\delta}(v^\top \Sigma_\ast v)^s\mper\label{eq:bound-pe-ws}\end{equation}

    
    
Let us analyze the random variable $a_S = \Paren{\frac{\pE[w(X')w_S]}{\pE[w_S]} - C\delta}$ when $S \sim D'(2t+1,\tzeta)$ (note that the only randomness here is over the choice of $S$, the pseudo-distribution $\tzeta$ is fixed). 


Let $G = \{i\in [n]: x_i = y_i\} \subseteq [n]$ be the set of ``good'' indices. Then, noting that $\pE[w(X')^{2t+1}] = \sum_{S = (i_1, i_2, \ldots, i_{2t+1}) \in G^{2t+1}} [\pE[w_S]]$ (by moving out $\1(x_i = y_i)$), we have:
    \begin{align*}
      \E_{S \sim D'(2t+1,\tzeta)} [a_S] &= \frac 1{\pE[w(X')^{2t+1}]} \sum_{S\in G^{2t+1}} \pE[w_S] \frac{\pE\left[w(X') w_S\right]}{\pE[w_S]}-C\delta\\
      &= \frac 1{\pE[w(X')^{2t+1}]} \sum_{S} \pE[w(X')w_S]-C\delta\\
      &= \frac 1{\pE[w(X')^{2t+1}]} \cdot \pE\left[w(X')\sum_{S}w_S\right]-C\delta\\
      &= \frac {\pE[w(X')^{2t+2}]}{\pE[w(X')^{2t+1}]} -C\delta
    \end{align*}
    We claim then that $\pE[w(X')^{2t+2}] \ge \pE[w(X')^{2t+1}]\pE[w(X')]$ (which is true in real expectations by monotonicity). Indeed, $\pE[w(X')] \le \pE[w(X')^{2t+2}]^{1/(2t+2)}$ by application of \Holder's inequality (Fact~\ref{fact:pseudo-expectation-holder}) with $f = 1$, $g = w(X')$ and $\pE[w(X')^{2t+1}] \le \pE[w(X')^{2t+2}]^{(2t+1)/(2t+2)}$ by application of \Holder's with $f = w(X')$ and $g = 1$. Multiplying these gives the result. Hence, we have
    \[\E_{S\sim D'(t,\tzeta)}[a_S] = \frac {\pE[w(X')^{2t+2}]}{\pE[w(X')^{2t+1}]} -C\delta \ge \pE[w(X')] - C\delta \ge \alpha^2 - C\delta \ge \frac{\alpha^2}{2}.\]
    where the second to last step is by Lemma~\ref{lem:high-entropy-pseudo-distributions}.

Now, since $a_S \le 1$ for every $S$, we may apply a Reverse Markov Bound to obtain that
\[\Pr_{S\sim D'}[a_S \le \alpha^2 / 4] = \Pr_{S\sim D'}[1 - a_S \ge 1 - \alpha^2 / 4]\le \frac{1 - \E[a_S]}{1 - \alpha^2/4} \le \frac{1 - \alpha^2 / 2}{1 - \alpha^2 / 4}\]
and hence $\Pr_{S\sim D'}[a_S > \alpha^2 / 4]\ge \frac{\alpha^2 / 4}{1 - \alpha^2 / 4}\ge \frac{\alpha^2}{4}$.

Thus, with probability at least $\alpha^2/4$ over the choice of $S\sim D'$, it must hold that $\pE[w_S w(X')]/\pE[w_S] - C\delta \geq \alpha^2/4$. Under this event and letting $\delta = \frac{\alpha^3}{2C}$, it follows by \eqref{eq:bound-pe-ws} that

\[\left(v^\top\frac{\pE_{\tzeta}[w_S \Sigma]}{\pE_{\tzeta}[w_S]} v\right)(v^\top \Sigma_\ast v)^{s-1}\ge \delta^2 \cdot \frac{\alpha^2}{4}(v^\top \Sigma_\ast v)^s \ge \frac{\alpha^8}{16C^2}(v^\top \Sigma_\ast v)^s\mper\]

Dividing through by $(v^\top \Sigma_\ast v)^{s-1}$ (and the upcoming guarantee holds even if this is $0$), we obtain $v^\top\frac{\pE_{\tzeta}[w_S \Sigma]}{\pE_{\tzeta}[w_S]} v \geq  \frac{\alpha^8}{16C^2} v^\top \Sigma_\ast v$ for every $v \in \R^d$. Therefore, this implies that $\frac{16C^2}{\alpha^8} \frac{\pE_{\tzeta}[w_S \Sigma]}{\pE_{\tzeta}[w_S]}  \succeq \Sigma_\ast$ with probability at least $ \alpha^2/4$ over the choice of $S\sim D'(t,\tzeta)$ and we are done.
    \end{proof}

\begin{lemma}[Analysis of Rounding, Upper-Bound]
  \label{lem:analysis-upper-bound}
Fix $\delta = \alpha^3/2C$. 
Let $\tzeta$ be a pseudo-distribution of degree $\geq O(s(\delta)^2)$ consistent with $\cA$ minimizing $\norm{\pE_{\tzeta} w}_2$. 
Let $\Proj$ be the projection matrix to the subspace spanned by all the eigenvectors of $\Sigma_*$ with eigenvalues at most $O(\alpha^{6t+8})$. Then, with probability at least $1 - \alpha^2 / 10$ over the choice of $S\sim D'(2t+1,\tzeta)$, we have that for $\hat{\Sigma}_S = \pE[w_S \Sigma]/\pE[w_S]$:
\begin{align}
\Proj \hat{\Sigma}_S \Proj \preceq  O(\alpha^{2t-20}) I + \Proj \Sigma_* \Proj \mper \label{eq:main-bound}
\end{align}

\end{lemma}

\begin{proof}[Proof of Lemma~\ref{lem:analysis-upper-bound}]
From Lemma~\ref{lem:frob-recovery-subspace}, we have:

\begin{align}
\cA \sststile{O(st)}{w,\Sigma} &\Biggl\{ w(X')^{2t+1}\Norm{\Proj(\Sigma - \Sigma_*)\Proj}_F^{2} \leq O(s^2) \Paren{\Paren{\frac{8C^t\delta^t}{\alpha^2}}^2 +  \frac{4}{\delta^4} \Norm{\Proj \Sigma_* \Proj }_{op}^2  } \Biggr\} \mper 
\end{align}

Using that $s(\delta) = O(1/\delta^2)$, it follows that $O(s^2)\Paren{\frac{8C^t\delta^t}{\alpha^2}}^2 = O(C^{2t})\cdot \delta^{2t - 4}\cdot \alpha^{-4}$, and using $\delta = \alpha^3 / (2C)$ yields that this term is at most $O(C^{4}\cdot\alpha^{6t - 16})$. The second term is at most $O(1/\delta^8)\cdot ||\Proj \Sigma_\ast \Proj||_{op}^2$ using again that $s(\delta)= O(1/\delta^2)$. Thus, altogether, we have: 

\begin{align}
\cA \sststile{O(st)}{w,\Sigma} &\Biggl\{ w(X')^{2t+1}\Norm{\Proj(\Sigma - \Sigma_*)\Proj}_F^{2} \leq O(\alpha^{6t-16}) +  O\left(\frac{1}{\delta^8}\right) \Norm{\Proj \Sigma_* \Proj }_2^2 \Biggr\} \mper 
\end{align}

Taking pseudo-expectations with respect to $\tzeta$ yields:
\begin{align}
\pE[w(X')^{2t+1}\Norm{\Proj (\Sigma - \Sigma_*)\Proj }_F^{2}] \leq O(\alpha^{6t-16}) +  O\left(\frac{1}{\delta^8}\right) \Norm{\Proj \Sigma_* \Proj }_2^2  \mper
\end{align}

Expanding, using that $\tzeta$ satisfies $w_S^2 = w_S$ and denoting once more $G = \{i\in [n] \mid x_i = y_i\}\subseteq [n]$ as the ``good'' set yields:
\begin{align}
\frac{1}{\pE[w(X')^{2t+1}]} \sum_{S = (i_1, i_2, \ldots, i_{2t+1})\in G^{2t+1}} \pE[w_S] \Norm{\Proj \Paren{\frac{\pE[w_S \Sigma] }{\pE[w_S]}- \Sigma_*}\Proj }_F^{2} \leq \frac{O(\alpha^{6t-16}) +  O\left(\frac{1}{\delta^8}\right) \Norm{\Proj \Sigma_* \Proj }_2^2 }{\pE[w(X')^{2t+1}] \cdot } \mper
\end{align}

From Lemma~\ref{lem:high-entropy-pseudo-distributions} and an application of \Holder's, we know that $\pE[w(X')^{2t+1}] \geq \alpha^{4t+2}$. Thus, we conclude that 
\begin{align}
\frac{1}{\pE[w(X')^{2t+1}]} \sum_{S = (i_1, i_2, \ldots, i_{2t}) \in G^{2t}} \pE[w_S] \Norm{\Proj \Paren{\frac{\pE[w_S \Sigma] }{\pE[w_S]}- \Sigma_*}\Proj }_F^{2} \leq O(\alpha^{2t-18}) +  O\left(\frac{1}{\alpha^{4t+2}\delta^8}\right) \Norm{\Proj \Sigma_* \Proj }_2^2  \mper
\end{align}

Since $\Proj$ is the projection to the subspace where $\Sigma_*$ has eigenvalues smaller than $O(\alpha^{6t-16}\delta^8)= O(\alpha^{6t+8})$, the second term is $O(\alpha^{2t-18})$. So altogether, the RHS above is $O(\alpha^{2t-18})$. Thus, 
\begin{align}
\frac{1}{\pE[w(X')^{2t+1}]} \sum_{S = (i_1, i_2, \ldots, i_{2t}) \in G^{2t}} \pE[w_S] \Norm{\Proj \Paren{\frac{\pE[w_S] \Sigma }{\pE[w_S]}- \Sigma_*}\Proj }_F^{2} \leq O(\alpha^{2t-18})  \mper
\end{align}
Note that by definition of $D'$, we have that is exactly
\[\E_{S\sim D'(2t + 1,\tzeta)}\Brac{\Norm{\Proj(\hat{\Sigma}_S - \Sigma_\ast)\Proj}_F^2} \le O(\alpha^{2t - 18})\]
Hence, by Markov's Inequality, it follows that
\[\Pr_{S\sim D'(2t + 1, \tzeta)}\Brac{\Norm{\Proj(\hat{\Sigma}_S - \Sigma_\ast)\Proj}_F^2 \le O(\alpha^{2t - 20})}\ge 1 - \frac{\alpha^2}{10}.\]
For such an $S$, we must then have that
\[
\Proj \hat{\Sigma}_S \Proj \preceq O(\alpha^{2t-20}) I + \Proj \Sigma_* \Proj \mper
\]

This completes the proof. 
\end{proof}

\paragraph{Putting Things Together} 
We now put the upper and lower bounds above together to prove Theorem~\ref{thm:recover-large-evs}. We will use the following simple bound:

\begin{lemma}[Splitting on Projections] \label{lem:projection-to-global}
Let $\Proj$ be a projection matrix to a subspace of $\R^d$. Let $A$ be any $d \times d$ PSD matrix. Then, we have:
\[
A \preceq 2 \Proj A \Proj  + 2 (I-\Proj) A (I-\Proj)\mper
\]
\end{lemma}
\begin{proof} For any vector $v \in \R^d$, we have:
\begin{align*}
v^{\top}A v &= v^{\top} (\Proj+I-\Proj) A( \Proj + I - \Proj) v\\
&= v^{\top} \Proj A \Proj v + v^{\top} (I-\Proj) A (I-\Proj) v + 2 v^{\top} \Proj A (I-\Proj)v\mper
\end{align*}
We now have using the Cauchy-Schwarz followed by the AM-GM inequality:
\begin{align*}
v^{\top} \Proj A (I-\Proj) v &= v^{\top} \Proj A^{1/2} A^{1/2}(I-\Proj) v\\ 
&\leq \sqrt{v^{\top} \Proj A \Proj v} \sqrt{v^{\top}(I-\Proj)A(I-\Proj) v}\\ 
&\leq \frac{v^{\top} \Proj A \Proj v + v^{\top} (I-\Proj)A(I-\Proj)v}{2}\mper
\end{align*}
\end{proof}

\begin{proof}[Proof of Theorem~\ref{thm:recover-large-evs}]
From the constraints $\cA$ and that the fact that $\tzeta$ is consistent with $\cA$ along with Fact~\ref{fact:partial-pseudo-expecting}, we must have that for every $2t+1$-tuple $S$,
\[
\pE[w_S \Sigma] \preceq \pE[w_S] \frac{8}{\alpha^2}\mper
\]
Thus,
\[
\hat{\Sigma}_S = \frac{\pE[w_S \Sigma]}{\pE[w_S]} \preceq \frac{8}{\alpha^2}\mper
\]

Next, from Lemma~\ref{lem:analysis-upper-bound}, we know that with probability $1 - \alpha^2 / 10$ over sampling $S\sim D'(2t + 1, \tzeta)$ we must have: 
\[
\Proj \hat{\Sigma}_S \Proj \preceq \Proj \Sigma_* \Proj + O(\alpha^{2t-20}) I \mper
\]

Thus, applying Lemma~\ref{lem:projection-to-global} and noting that $(I - \Proj)^2 = I - \Proj$, we have that for such an $S$:
\[
\hat{\Sigma}_S \preceq 2 \Proj \Sigma_* \Proj + O(\alpha^{2t-20}) I + \frac{16}{\alpha^2} (I-\Proj)\mper
\]

Next, observe that $(I-\Proj) \Sigma_* (I-\Proj) \succeq O(\alpha^{6t+8}) (I-\Proj)$ since $\Proj$ is the subspace of all eigenvectors of $\Sigma_*$ with eigenvalues $\leq O(\alpha^{6t+8})$.

Note also that $\Sigma_\ast = \Proj \Sigma_\ast \Proj + (I - \Proj)\Sigma_\ast (I - \Proj)$ (for any vector $v$, $v^\top (I - \Proj)\Sigma_\ast v = 0$ by orthonormality of eigenvectors and the fact that eigenvectors must lie in either $\Proj$ or $I - \Proj$).

Thus, we must have that with probability at least $1-\alpha^2/10$ over the choice of $S\sim D'(2t + 1, \tzeta)$:

\[
\hat{\Sigma}_S \preceq 2 \Proj \Sigma_* \Proj +  O(\alpha^{2t-20}) I + O\left(\frac{1}{\alpha^{6t+8}}\right)\cdot \frac{16}{\alpha^2} (I-\Proj) \Sigma_* (I-\Proj) \preceq O\left(\frac{1}{\alpha^{6t+10}}\right) \Sigma_* + O(\alpha^{2t-20}) I\mper
\]

Next, from Lemma~\ref{lem:analysis-lower-bound}, we have that with probability at least $\alpha^2/4$ over the choice of $S$ conditioned on $S$ satisfying $x_i = y_i$ for $i \in S$:

\[
O(1/\alpha^8) \hat{\Sigma}_S \succeq \Sigma_* \mper
\]
By a union bound, we obtain that with probability at least $\alpha^2/10$ over the choice of $S\sim D'(2t + 1,\tzeta)$:
\begin{equation} \label{eq:rounding-succeeds}
\Sigma_* \preceq O(1/\alpha^8) \hat{\Sigma}_S \preceq O(\frac{1}{\alpha^{6t+18}}) \Sigma_* + O(\alpha^{2t-28}) I \mper 
\end{equation}

Finally, note that since $\pE[w(X')^t] \geq \alpha^{2}$, the chance that $S$ satisfies $x_i = y_i$ for every $i \in S$ (as in, our draw from $D$ is also from $D'$) is at least $\pE[w(X')^t]/(\alpha^t) \geq \alpha^{t}$. Thus, together, we obtain that with probability at least $\alpha^{t+2}/10$, the candidate $O(1/\alpha^8)\hat{\Sigma}_S$ corresponding to the set $S$ output by the rounding algorithm satisfies \eqref{eq:rounding-succeeds}. This completes the proof.

\end{proof}

\section{Subgaussian Restriction}
In this section, we describe and analyze a subroutine that effectively allows us assume that the corrupted sample $Y$, after a pruning step, itself has subgaussian moments with respect to its covariance. To do so, we use a new definition of a good set as compared to the previous section.

\begin{definition}[Well-behaved Set]
  For $d\in \N$, we say that a subset $X\subseteq \QQ^d$ with mean $\mu_\ast$ and covariance $\Sigma_\ast$ is a $(C,\delta,t)$-well behaved set if:
  \begin{enumerate}
    \item {\bf Small Mean: } $\mu_\ast\mu_\ast^\top \preceq 0.1\Sigma_\ast$.
    \item {\bf Anticoncentration: } $X$ is $(\delta,C\delta)$-anticoncentrated.
    \item {\bf Subgaussianity: } For indeterminate $v$ and all $s\le 2t$, \[\sststile{O(s)}{v}\Set{\E_{x\sim X} \langle x,v\rangle^{2s} \le(Cs)^{2s}\cdot \left(\E_{x\sim X} \langle x,v\rangle^2\right)^s}.\]
  \end{enumerate}
  
\end{definition}

With this definition in mind, we will prove the following main theorem: 

\begin{theorem}[Subgaussian Restriction] \label{thm:subgaussian-restriction}
Fix $1\geq \alpha >0$ and $d,n,t \in \N$. Let $\tau > 0$. 
Let $X$ be $(C,\delta,t)$-well behaved set and let $Y$ be an arbitrary collection of $n$ points in $\R^d$ such that $|Y \cap X| \geq \alpha n$ and $\frac{1}{n} \sum_i y_i y_i^{\top} = (1\pm 2^{-d})I$. 
Then, for any $0 < \tau \le O(\alpha^{2t + 11})$ there is a $Y' \subseteq Y$ with $|X\cap Y'|\ge \max(\alpha|Y'|,(\alpha - \alpha^{10})n)$ and satisfying 
\[
    \sststile{O(t)}{v} \Set{\E_{y' \sim Y'} \iprod{y',v}^{2t} \leq  \frac{1}{\tau} \Paren{\frac{16Ct}{\alpha^2}}^t \Paren{\E_{y' \sim Y'} \iprod{y',v}^2}^t }
    \mper\]
Further, given $\alpha,Y$ as above, a $Y' \subseteq Y$ satisfying the above properties can be found in time $(Bn)^{O(t)}$ where $B$ is the bit complexity of entries of $Y$.
\end{theorem}

Our algorithm to establish the above theorem uses a subroutine that shows that if $Y$ is not subgaussian with the parameters above, then, we can find a subset of $Y$ that is ``outlier-heavy'', that is, contains significantly bigger fraction of outliers than in all of $Y$. Removing these points can only ``increase'' the density of the inliers so cannot hurt us. 

We will use the following simple inequality that relates the tails of a distribution to its mean.

\begin{lemma}[Expectation vs Tail] \label{lem:exp-vs-cdf}
Let $Z$ be a non-negative real-valued random variable. Then, 
\[
\E[Z] \leq \frac{1}{2} + \int_{1/2}^{\infty} \Pr[Z \geq L] dL \mper 
\]
\end{lemma}
\begin{proof}
$\E[Z] = \int_{0}^{\infty} \Pr[Z \geq L] dL = \int_{0}^{1/2} \Pr[Z \geq L] dL + \int_{1/2}^{\infty} \Pr[Z \geq L] dL \leq 1/2 + \int_{1/2}^{\infty} \Pr[Z \geq L] dL$. 
\end{proof}

We will use the above bound to show that if $Y$ is not certifiably subgaussian, then, the outliers must make an outsized contribution to one of a few natural sections of the tail of $Y$. 
\begin{lemma}[Outliers Must Dominate Some Portion of Tail]
Fix $1\geq \alpha >0$ and $d,n,t \in \N$. Let $\tau > 0$ and $X$ be $(C,\delta,t)$-well behaved set.

Let $Y$ be an arbitrary collection of $n$ points in $\R^d$ such that $|Y \cap X| \geq \alpha' n$ for $\alpha' \geq \alpha/2$ and $\frac{1}{n} \sum_i y_i y_i^{\top} = (1 \pm 2^{-d}) I$. 
Suppose there is a pseudo-distribution $\tzeta$ of degree $\geq 4t$ over $d$-dimensional vector-valued indeterminate $v$ such that $\pE_{\tzeta}[\E_{y \sim Y} \iprod{v,y}^{2t}] > \Gamma$ for $\Gamma = \frac{1}{\tau} (16Ct/\alpha^2)^{t}$. Then, there is a $L >0$ such that $\Pr_{y \sim Y}[ \pE_{\tzeta}[\iprod{v,y}^{4t}] \geq L] > \Gamma/4L$. In contrast, for every $L>0$, $\Pr_{x \sim X}[ \pE_{\tzeta}[\iprod{x,v}^{4t}]\geq L] \leq \tau^2 \Gamma^2/L$.  \label{lem:pruning-subroutine}
\end{lemma}
\begin{proof}
Let's prove the first claim. Note that $\Gamma > 1$ and suppose for the sake of contradiction that for every $L>0$, $\Pr_{y \sim Y}[ \pE_{\tzeta}[\iprod{v,y}^{4t}] \geq L] \leq \Gamma/4L$. Observe that by Cauchy-Schwarz inequality for pseudo-distributions and the fact that $\tzeta$ has degree $\geq 4t$, we must have that $\pE_{\tzeta}[\iprod{v,y}^{4t}] \geq  \Paren{\pE_{\tzeta}[\iprod{v,y}^{2t}]}^2$. Thus, whenever $\pE_{\tzeta}[\iprod{v,y}^{2t}] \geq L$, $\pE_{\tzeta}[\iprod{v,y}^{4t}] \geq L^2$.

Applying Lemma~\ref{lem:exp-vs-cdf} to the random variable $\pE_{\tzeta}[\iprod{v,y}^{2t}]$ as $y$ is chosen uniformly at random from $Y$ (notice the randomness here is simply $y$ being chosen uniformly from $Y$), we have: 
\begin{align*}
\E_{y \sim Y} [\pE_{\tzeta}[\iprod{v,y}^{2t}]] &= 1/2+ \int_{1/2}^{\infty} \Pr[\pE_{\tzeta}[\iprod{v,y}^{2t}]\geq L] dL\\
&\leq 1/2 + \int_{1/2}^{\infty} \Pr[\pE_{\tzeta}[\iprod{v,y}^{4t}]\geq L^2] dL\\
&\leq 1/2+ \int_{1/2}^{\infty} \frac{\Gamma}{4L^2} dL \leq 1/2 + \Gamma/2 < \Gamma \mcom
\end{align*}
since $\Gamma>1$. This is a contradiction and proves the first claim.

Next, let $\Sigma_* = \frac{1}{n} \sum_{i =1}^n x_i x_i^{\top}$. From the argument in Lemma~\ref{lem:feasibility-coarse-spectral-recovery}, we have that $\Sigma_\ast\preceq \frac{8}{\alpha^2}I$.

By $C$-certifiable $4t$-subgaussianity of $X$ and the fact that $\tzeta$ is a pseudo-distribution of degree $\geq 4t$, we must have that $\pE_{\tzeta}[ \E_{x\sim X}[ \iprod{x,v}^{4t}]] \leq (2Ct)^{2t} \Paren{\E_{x \sim X} \iprod{x,v}^2}^{2t} = (2Ct)^{2t} (\frac{8}{\alpha^2})^{2t} = (\tau\Gamma)^2$. By Markov's inequality, $\Pr_{x \sim X}[ \pE_{\tzeta}[\iprod{x,v}^{4t}] > L] \leq \tau^2 \Gamma^2/L$. 

\end{proof}

We can now describe the subgaussian restriction algorithm. 

\begin{mdframed}
      \begin{algorithm}[Subgaussian Restriction Algorithm]
        \label{algo:subgaussian-restriction}\mbox{}
        \begin{description}
        \item[Given:]
            A set of points $Y = \{y_1, y_2, \ldots, y_n\} \subseteq \bbQ^d$ and $\alpha, \tau >0$.
        \item[Output:]
          A subset $Y' \subseteq Y$ that is certifiably subgaussian: 
          \[ 
          \sststile{2t}{v} \Set{ \E_{y \sim Y'} \iprod{y, v}^{2t} \leq \frac 1\tau\Paren{\frac{16Ct}{\alpha^2}}^t \Paren{\E_{y \sim Y'} \iprod{y, v}^{2}}^{t}}\mper
          \] 
        \item[Operation:]\mbox{} Initialize $Y' = Y$. While true, do: 
        \begin{enumerate}
        \item \textbf{Isotropize: } By a linear transformation of all $y \in Y'$, ensure that $\E_{y \sim Y'} yy^{\top} = (1 \pm 2^{-d})I$. 
        \item \textbf{Subgaussianity Check: } Find a degree $4t$ pseudo-distribution $\tzeta$ over $d$-dimensional vector-valued indeterminate $v$ satisfying $\norm{v}_2^{2} = 1$ and maximizing $\E_{y \sim Y'} [\pE_{\tzeta}[ \iprod{v,y}^{2t}]]$. If the objective value is $\leq \Gamma$, halt and return $Y'$. 
        \item \textbf{Find Outlier-Heavy Level Set: } If not, by binary search, find $L$ such that $\Pr_{y \sim Y'}[ \pE_{\tzeta}[\iprod{v,y}^{4t}] \geq L] > \Gamma/4L$. Such an $L$ is guaranteed to exist by Lemma~\ref{lem:pruning-subroutine}. Let $R$ be the set of all such points in $Y$. 
        \item \textbf{Prune: } Set $Y' = Y \setminus R$. Go to step 1. 
        \end{enumerate}
        \end{description}
      \end{algorithm}
    \end{mdframed}

\begin{proof}[Proof of Theorem~\ref{thm:subgaussian-restriction}]
Starting with $Y' = Y$, Algorithm~\ref{algo:subgaussian-restriction} repeatedly (approximately) isotropizes $Y'$ and whenever $Y'$ does not satisfy the desired subgaussianity condition, it finds an ``outlier-heavy'' subset of $Y'$ to prune away. 

Let's analyze this algorithm. We will prove the following two facts about the run of the algorithm to analyze it. 
\begin{enumerate}
\item In each iteration either we halt and return or remove at least one point from $Y'$. This immediately implies that the algorithm terminates in $n$ iterations and thus, in time $(Bn)^{O(t)}$ (where $B$ is the bit complexity of entries of $Y$).
\item Let $R \subseteq Y'$ be the subset of points removed in any iteration. Then, at most $\alpha^{10} |R|$ of the points belong to $X \cap Y$. Notice that this immediately implies that the total number of ``inliers'' (i.e. points in $X \cap Y$) removed are at most $\alpha^{10} n$.
\end{enumerate} 

We will prove both the claims by induction over the iterations of the algorithm. 
Consider any iteration starting with $Y'$. 
Let $|X \cap Y'| = \alpha'$. Then, by inductive assumption, we know that $\alpha' \geq (\alpha -\alpha^{10}) \geq \alpha/2$. By an argument similar to the one in the proof of Lemma~\ref{lem:feasibility-coarse-spectral-recovery}, we can infer that (after the isotropic linear transformation), the transformed $\Sigma_* \preceq \frac{8}{\alpha^2} I$. If $Y'$ does not pass the subgaussianity check, then, by duality of pseudo-distributions and SoS proofs, we can find a pseudo-distribution $\tzeta$ on $v$ such that $\E_{y' \sim Y'} \pE_{\tzeta}[ \iprod{v,y'}^{2t}]>\Gamma$. Thus, by Lemma~\ref{lem:pruning-subroutine}, the binary search to find $L$ for an outlier-dense level set must succeed. So we must remove at least one point from $Y'$ (proving (1)). We know then that the fraction of $R$ contained in $X\cap Y'$ is at most $\cfrac{\frac{\tau^2\Gamma^2}{L}}{\frac{\alpha/2\cdot \Gamma}{8L}}$ where we used that $|Y'| \ge \frac \alpha 2 \cdot n$. Rearranging this yields that we need $\frac{128Ct\tau }{\alpha^{2t + 1}}\le \alpha^{10}$ implying that we can take $\tau = O(\alpha^{2t + 11})$ to finish (2). This completes the proof.

\end{proof}

\section{Splitting via Paley-Zygmund Anti-Concentration} \label{sec:splitting}
In this section, we describe and analyze a subroutine that takes input a corrupted sample $Y$ and prunes away a constant fraction of $Y$. Crucially, our splitting algorithm requires only a mild anti-concentration property that can be inferred from only moment \emph{upper bounds}. This is in contrast to strong (and certifiable) anticoncentration needed in our coarse spectral recovery algorithm. 

Our algorithm below first performs a simple check on a candidate obtained from coarse spectral recovery step. If a candidate passes the check, we obtain a certificate that a coarse spectral estimate in fact must in fact be a multiplicative spectral estimate. If not, our algorithm efficiently splits $Y$ into two approximately balanced parts such that the most of $X \cap Y$ is included on one side. Thus, we either ``finish'' by finding a good candidate covariance or can recurse and make progress be decreasing the fraction of the sample that is corrupted. 

Our algorithm itself is simple:
\begin{mdframed}
\begin{algorithm}[Spectral Splitting]
  \label{algo:spectral-splitting}

  \begin{enumerate}

    \item \textbf{Input:} $Y \subseteq \R^d$, a vector $v \in \R^d$.
    \item \textbf{Operation:} Return $Y_2$ -- the set of all points $y_i$ such that $\iprod{y_i,v}^2>0.5$. 
    \end{enumerate}
\end{algorithm}
\end{mdframed}

We first prove that the splitting algorithm always makes progress.
\begin{lemma}[Splitting Algorithm: Progress] \label{lem:splitting-algo-progress}
Let $Y$ be a set of points in $\R^d$ such that 1) $\frac{1}{n} \sum_i y_i y_i^{\top} = (1 \pm 2^{-d})I$ and 2) $\frac{1}{n} \sum_{i =1}^n \iprod{y_i,v}^{2t} \leq \Delta \Norm{v}_2^{2t}$ for $\Delta = \frac{1}{\tau} \Paren{\frac{16Ct}{\alpha'^2}}^t$ and all $v\in \R^d$. Let $Y_2 = \{y \in Y \mid \iprod{y,v}^2> 1/2\}$. Then, $|Y_2| > O(\alpha^{10})|Y|$. 
\end{lemma}

Next, we prove that if $Y$ intersects with the original uncorrupted sample $X$ appreciably, then the splitting algorithms prunes away only a small fraction of points from $Y \cap X$.  

\begin{lemma}[Splitting Algorithm: Correctness] \label{lem:splitting-algo-correctness}
Suppose that $X$ is a set of $n$ points satisfying $2t$-certifiable $C$-subgaussianity: $\frac{1}{n} \sum_{i=1}^n \iprod{x_i,v}^{2t}\leq (Ct)^t \Paren{\frac{1}{n} \sum_{i=1}^n \iprod{x_i,v}^{2}}^t$ for $t = O(1/\alpha)$. 
For some $\alpha, \tau > 0$, suppose $Y$ is a set of $n$ points in $\R^d$ satisfying $|Y \cap X|\geq \alpha' n \geq \alpha n/2$. Let $\widehat{\Sigma}$ be a PSD matrix such that $\Sigma_* \preceq \widehat{\Sigma}$ 
 and suppose there is a unit vector $v$ such that $v^{\top} \widehat{\Sigma} v \leq \eta$ for $\eta < O(\alpha^6)$. 
Then, there is a $\poly(n)$ time algorithm to split $Y$ into $Y_1 \cup Y_2$ such that $|X \cap Y_1|\geq (\alpha' - O(\alpha^{12})|Y|$. 
\end{lemma}

We will use the following basic consequence of the Paley-Zygmund anti-concentration inequality:
\begin{lemma}[Mild Anti-Concentration via Paley-Zygmund]
Let $Y$ be a set of $n$ points in $\R^d$ such that $\frac{1}{n} \sum_i y_i y_i^{\top} = (1 \pm 2^{-d})I$. 
Suppose further that $\frac{1}{n} \sum_{i =1}^n \iprod{y_i,v}^4 \leq \Delta \Norm{v}_2^4$ for every $v \in \R^d$. 
Then, for any $v \in \R^d$, the fraction of $y_i$s such that $|\iprod{y_i,v}|>\frac{1}{2}$ is at least $\frac{1}{4\Delta}$. \label{lem:paley-zygmund}
\end{lemma}

\begin{proof}
Observe that the contribution of $y_i$s such that $|\iprod{y_i,v}|\leq \frac{1}{2}$ to $\frac{1}{n}\sum_i \iprod{y_i,v}^2 \leq \frac{1}{4}$. Thus, by Cauchy-Schwarz inequality, we have:
\begin{align*}
\frac{1}{2} \Norm{v}_2^2 \leq \frac{1}{n} \sum_i \iprod{y_i,v}^2 \1(|\iprod{y_i,v}|>1/2) &\leq  \sqrt{\frac{1}{n} \sum_i \iprod{y_i,v}^4} \sqrt{\frac{1}{n} \sum_{i = 1}^n \1(|\iprod{y_i,v}|>1/2)}\\
&\leq \sqrt{\Delta} \Norm{v}_2^2 \sqrt{\Paren{\frac{1}{n} \sum_{i = 1}^n \1(|\iprod{y_i,v}|>1/2)}} \mper
\end{align*}
Rearranging yields that:
\[
\frac{1}{n} \sum_{i = 1}^n \1(|\iprod{y_i,v}|>1/2)  \geq \frac{1}{4 \Delta}\mper
\]

\end{proof}

\begin{proof}[Proof of Lemma~\ref{lem:splitting-algo-progress}]
Let us analyze $Y_2$. By Hölder's inequality, $\frac{1}{n} \sum_{i =1}^n \iprod{y_i,v}^{4} \leq \Paren{\frac{1}{n} \sum_{i =1}^n \iprod{y_i,v}^{2t}}^{2/t} \leq \frac{1}{\tau^{2/t}} \Paren{\frac{16Ct}{\alpha'^2}}^2 \norm{v}_2^4$. Thus, by Lemma~\ref{lem:paley-zygmund}, we must have that at least a $O(\alpha^4) \tau^{2/t} /t^2$ fraction of $y_i$s satisfy $\iprod{y_i,v}^2 \geq \frac{1}{2}$. This fraction, for $t = O(1/\alpha)$ and $\tau = O(\alpha^{2t+11})$ is at least $O(\alpha^{10})$. Thus, $|Y_2| \geq O(\alpha^{10})|Y|$ as desired. 
\end{proof}

\begin{proof}[Proof of Lemma~\ref{lem:splitting-algo-correctness}]
Observe that $v^{\top} \Sigma_* v \leq v^{\top} \widehat{\Sigma}v \leq \eta$. Thus, using $C$-subgaussianity of $4$th moments of $X$ 
we have that
\[\E_{x_i\sim X}[\langle x_i,v\rangle^4] \le (2C)^2(v^\top \Sigma_\ast v)^2 \le 4C^2\eta^2.\] So, if a $\kappa$ fraction of $x_i$ had $\langle x_i,v\rangle^2\ge \frac 12$ we would have that
\[\Paren{\frac 12}^2\kappa \le 4C^2\eta^2\] which upon rearranging yields that the fraction of $x_i$s such that $\iprod{x_i,v}^{2} \geq 1/2$ is at most $16 C^2 \eta^2$. Thus, $Y_1$ contains at least $\alpha' - O(\eta^2) = \alpha' - O(\alpha^{12})$ fraction of the points in $Y$ as desired. 
\end{proof}

\section{List-Decodable Covariance Estimation with Spectral Accuracy} \label{sec:spectral-full}
In this section, we put the three components from the previous sections together to obtain an algorithm for list-decodable covariance estimation with multiplicative spectral recovery guarantee.

\begin{mdframed}
      \begin{algorithm}[List-Decodable Covariance Estimation with Spectral Recovery Guarantee]
        \label{algo:list-decoding-covariance-spectral}\mbox{}
        \begin{description}
        \item[Given:]
             $Y = \{y_1, y_2, \ldots, y_n\} \subseteq \bbQ^d$ such that $\frac{1}{n}\sum_{i = 1}^n y_i y_i^{\top} = (1 \pm 2^{-d})I$ and $\alpha>0$. 
        \item[Output:]
          A list $\cL$ of positive semidefinite matrices in $\bbQ^{d\times d}$. 
        \item[Operation:]\mbox{}  Maintain a list $\cL'$ of candidates with ``witness subsets'' of $Y$. Initialize $\cL'$ with $(Y,I,0)$. During the course of the algorithm, some candidates in $\cL'$ will become ``final''.
        \begin{enumerate}
        \item Set $t_1 = O(1/\alpha)$, $t_2 = 20$, and $\tau = O(\alpha^{2t_1 + 11})$. For $g = 0,1,\ldots,$ do:
        \item \textbf{Process $g$-th Generation Candidates: } While there is $(Y^{(i)}, \widehat{\Sigma}^{(i)},g)$ in $\cL'$ that is not marked final, remove $(Y^{(i)}, \hat{\Sigma}^{(i)},g)$ from $\cL'$ and run the following steps:
         \begin{enumerate}
         \item \textbf{Subgaussian Restriction: } Run the Subgaussian Restriction Algorithm (Algorithm~\ref{algo:subgaussian-restriction}) to find $Y' \subseteq Y^{(i)}$ satisfying\\ $\sststile{2t}{v} \Set{ \E_{y \sim Y'} \iprod{y,v}^{2t} \leq \frac{1}{\tau} \Paren{\frac{16Ct}{\alpha^2}}^t \Paren{\E_{y \sim Y'} \iprod{y,v}^{2}}^t}$ for $t = t_1$ and $\tau = \tau$ above.
         \item \textbf{Isotropization: } By a linear transform, ensure that $Y'$ satisfies $\E_{y' \sim Y'} y'y'^{\top} = I$.
         \item \textbf{Coarse Spectral Recovery: } Apply Coarse Spectral Recovery Algorithm (Algorithm~\ref{algo:coarse-spectral-recovery}) to input $Y'$ with $t = t_2$, fraction of inliers set to $\alpha/2$ and failure probability $\nu = \alpha^{10}/100$. If the algorithm outputs infeasible, go back to the beginning of the loop. Otherwise: 
         \item \textbf{Check Certificate of Spectral Approximation: } For each of the candidates $\hat{\Sigma}_i$s produced, check if the minimum eigenvalue of $\hat{\Sigma}_i$ is at least $\eta$ for $\eta = 2\alpha^{6}$. If yes, add $(\hat{\Sigma}_i,Y',g+1)$ to $\cL'$ after undoing the isotropic transformation from Step (b) above and label the candidate ``final''.
         \item \textbf{Apply Spectral Splitting: } For each candidate $\hat{\Sigma}_i$ that is not labeled final, find an eigenvector $v$ of $\hat{\Sigma}_i$ with eigenvalue $< \eta$ and apply Spectral Splitting Algorithm (Algorithm~\ref{algo:spectral-splitting}) with respect to $\hat{\Sigma}_i$ to $Y'$ to obtain $Y''$. If $|Y''|< \alpha n/2$, reject the candidate and continue to loop. Otherwise, add $(\hat{\Sigma}_i, Y'',g+1)$ to the list $\cL$ after undoing the isotropic linear transformation on both $Y'_i$ and $\hat{\Sigma}_i$ from Step 2 above. 
         \end{enumerate}
         \item When the for loop exits, add all the candidate covariances $\hat{\Sigma}$ from $\cL'$ to $\cL$.
        \end{enumerate}
        \end{description}
      \end{algorithm}
    \end{mdframed} 

We first explain the idea of the algorithm and the analysis. We start with the input corrupted sample $Y$ along with the $0$-th generation candidate $I$ and apply the subgaussian restriction procedure. Intuitively, the goal of the algorithm is to make progress on getting a good candidate covariance in the list as the generations $g$ progress. Theorem~\ref{thm:subgaussian-restriction} ensures that in this process, we only increase the density of inliers (i.e. points intersecting with the original good set of points $X$) and that the number of inliers is at least $(\alpha-\alpha^{10})n$.. We then make the resulting $Y$ approximately isotropic and apply the coarse spectral recovery algorithm to obtain a list of $O(1/\alpha^2)$ size that is guaranteed to contain an $\hat{\Sigma}_i$ such that $\Sigma_* \preceq \hat{\Sigma}_i \preceq \poly(1/\alpha) + \alpha^{6} I$ (where we have set $\eta = \alpha^6$). If all eigenvalues of every candidate $\hat{\Sigma}_i$ are at least $2\alpha^6$, then, the coarse guarantee implies a multiplicative spectral approximation (see Lemma~\ref{lem:certificates-for-spectral-recovery}) as we'd like and so we are done (our algorithm labels such candidates ``final''). 

Otherwise, there must be a $\hat{\Sigma}_i$ that has an eigenvalue smaller than $2 \alpha^6$. In this case, because $Y$ satisfies certifiable subgaussianity (as a result of our subgaussian restriction subroutine), spectral splitting can prune out $\geq O(\alpha^{10})$ fraction of points from $Y$ while removing at most $O(\alpha^{12})n$ points from $X \cap Y$. We call the pruned set $Y'$ a ``witness subset'' for $\hat{\Sigma}_i$. If $\hat{\Sigma}_i$ happened to be a ``candidate good estimate'' for the unknown covariance, then, the witness subset is sufficient to work with from this point on in the algorithm. Of course we do not know whether $\hat{\Sigma}_i$ is the ``right'' candidate and in general, the witness subset is different for different candidates $\hat{\Sigma}_i$. Thus, our algorithm maintains the the current witness subset for each potential candidate $\hat{\Sigma}_i$ in our list. In the subsequent runs of the algorithm, we repeat the algorithm on witness subset for each candidate that is not marked final. At any point of time, each member $(\hat{\Sigma}^{(i)}, Y^{(i)})$ in the list $\cL'$ is obtained by a sequence of subgaussian-restrictions, coarse-spectral-recovery and spectral splitting applied to the initial input set $Y$. And each time a candidate is processed, we must decrease the size of its witness set by at at least $1-O(\alpha^{10})$ factor while increasing the list size by $O(1/\alpha^{22})$. Since we know that for the ``correct candidate'', we never throw away more than $\alpha n/2$ inliers, if the size of a witness set drops below $\alpha n/2$, we can comfortably reject the corresponding candidate. Thus the number of generations in the algorithm cannot be more than $\tilde{O}(1/\alpha^{10})$ giving us the bound on the list-size.

We now formally argue the guarantees of the algorithm. The following theorem summarizes the guarantees of Algorithm~\ref{algo:list-decoding-covariance-spectral}. 
\begin{theorem} \label{thm:list-decoding-covariance-spectral-guarantee}
Let $1 \geq \alpha > 0$. 
Suppose $X = \{x_1, x_2, \ldots, x_n\} \subseteq \bbQ^d$ is a \emph{good} set (Definition~\ref{def:good-set-coarse}) of $n$ points satisfying $\E_{x \sim X} xx^{\top} = \Sigma_*$ such that $|Y \cap X| \geq \alpha n$. 
Let $Y \subseteq \bbQ^d$ be a set of $n$ points such that $\frac{1}{n} \sum_i y_i y_i^{\top} = I$. 
Then, Algorithm~\ref{algo:list-decoding-covariance-spectral} on input $Y$, 1) runs in time $(Bn/\alpha)^{O(1/\alpha^{12})}$, 2) outputs a list $\cL$  of size $\alpha^{\tilde{O}(1/\alpha^6)}$, with the guarantee that with probability at least $0.99$ only over the randomness of the algorithm, 3) $\cL$ contains a $\hat{\Sigma}$ satisfying $\Sigma_* \preceq \hat{\Sigma} \preceq O(1/\alpha^{150})\Sigma_*$. 
\end{theorem}

We prove the theorem in the following sequence of lemmas.

The first set of claims below analyze the size of the list $\cL$ output by the algorithm.  
\begin{lemma}[Size of witness sets in $g$-th generation] \label{lem:witness-size-generation}
Let $(\hat{\Sigma}_i,Y^{(i)},g)$ be a $g$-th generation candidate that is not marked final. Then, $|Y^{(i)}|\leq (1-O(\alpha^{10}))^g n$. 
\end{lemma}
\begin{proof}
We prove this by induction on the generation iterator $g$. For the base case, observe that at the beginning, $g=0$ and $|Y^{(0)}| = |Y|=n$. Next, for the inductive case, observe that a $g$-th generation candidate for $g\geq 1$ is obtained by taking a 1) $g-1$-th generation candidate $(\hat{\Sigma}_i,Y^{(i)})$, 2) applying subgaussian restriction to $Y^{(i)}$, applying coarse-spectral recovery to obtain a list of $g$-th generation candidates by applying Algorithm~\ref{algo:coarse-spectral-recovery} to input $Y^{(i)}$. By Theorem~\ref{thm:subgaussian-restriction}, $Y^{(i)}$ satisfies \[
    \sststile{O(t)}{v} \Set{\E_{y \sim Y^{(i)}} \iprod{y,v}^{2t} \leq  \frac{1}{\tau} \Paren{\frac{16Ct}{\alpha^2}}^t \Paren{\E_{y \sim Y^{(i)}} \iprod{y,v}^2}^t }
    \mper\] 
If a $g$-th generation candidate covariance $\hat{\Sigma}_j$ is not marked final, then, there must be an unit length eigenvector $v$ of $\hat{\Sigma}_j$ with an eigenvalue of at most $\eta =\alpha^{6} \E_{y \sim Y^{(i)}} \iprod{y,v}^2$. Thus, the assumptions of Lemma~\ref{lem:splitting-algo-progress} are met and the splitting algo must prune away at least $O(\alpha^{10})|Y^{(i)}|$ points from $Y^{(i)}$ before producing a $g$-th generation candidate $(\hat{\Sigma}_j,Y^{(j)},g)$. This completes the proof. 
\end{proof}

\begin{lemma}[Bounding the Number of Generations] \label{lem:bounding-generations}
The maximum value of $g$ during the run of the algorithm is $O(\log 1/\alpha)/\alpha^{10}$. 
\end{lemma}
\begin{proof}
From Lemma~\ref{lem:witness-size-generation}, the size of the witness set drops as $(1-O(\alpha^{10}))^g$ in the $g$-th generation. If $g> O(\log 1/\alpha)/\alpha^{10}$, then, the above size is $\leq \alpha n/2$ in which case, Step 2(e) of the algorithm exits the loop disallowing further generations. This completes the proof.   
\end{proof}

As an immediate corollary, we obtain a bound on the size of the list obtained by the algorithm above:

\begin{lemma}[List Size Bound]
The size of the list $\cL$ of covariances output by the algorithm is at most $\alpha^{\tilde{O}(1/\alpha^{10})}$. 
\end{lemma}
\begin{proof}
Every candidate in the list $\cL$ corresponds to a candidate from $\cL'$ that the algorithm marks ``final''. Each candidate in $\cL'$ marked final is at most of $O(\log 1/\alpha)/\alpha^{10}$ generation from Lemma~\ref{lem:bounding-generations}. Each $g$-th generation candidate in $\cL'$ produces at most $O(1/\alpha^{20})$ $g+1$-th generation candidates from the guarantees of Theorem~\ref{thm:recover-large-evs}. This immediately yields the upper bound on the list size as desired. 
\end{proof}

Next, we bound the running time of the algorithm. 

\begin{lemma}[Running Time]
The running time of Algorithm~\ref{algo:list-decoding-covariance-spectral} is $(Bn/\alpha)^{O(1/\alpha^{12})}$. 
\end{lemma}
\begin{proof}
Given our parameters, the running time of each iteration is dominated by the running time of coarse spectral recovery. The number of iterations is upper bounded by $\alpha^{-\tilde{O}(1/\alpha^{11})}$ by an argument similar to the one bounding the size of the list output by the algorithm. This gives the final running time bound as desired.  
\end{proof}

Finally, we prove the correctness -- that one of the candidates in the list gives a multiplicative approximation to the unknown covariance. Our proof will rely on the following simple observation that we will use to infer that if all eigenvalues of every candidate $\hat{\Sigma}$ are not too small relative to its witness set then the set of candidates must contain a multiplicative spectral approximation to the unknown covariance $\Sigma_*$.

\begin{lemma}[Certificates of Spectral Recovery]  \label{lem:certificates-for-spectral-recovery}
Suppose $\widehat{\Sigma}$ satisfies $\Sigma_* \preceq \widehat{\Sigma} \preceq O\Paren{\frac{1}{\alpha^{150}}} \Sigma_* + \eta I$. Further, suppose that for all unit vectors $v$, $v^{\top} \widehat{\Sigma} v > 2 \eta$. Then, 
\[
\Sigma_* \preceq \widehat{\Sigma} \preceq O\Paren{\frac{1}{\alpha^{150}}} \Sigma_*\mper
\]
\end{lemma}
\begin{proof}
Let $v$ be a unit vector such that $\lambda = v^{\top} \widehat{\Sigma} v > 2\eta$. Then, we have:
\[
\frac{v^{\top}\Sigma v}{v^{\top} \Sigma_* v} \leq O\Paren{\frac{\lambda}{(\lambda-\eta) \alpha^{150}}} \leq O\Paren{1/\alpha^{150}} \mper 
\]
This completes the proof.
\end{proof}

\begin{lemma}[Correctness]
Let $1 \geq \alpha > 0$. 
Suppose $X = \{x_1, x_2, \ldots, x_n\} \subseteq \bbQ^d$ is a \emph{good} set of $n$ points satisfying $\E_{x \sim X} xx^{\top} = \Sigma_*$ such that $|Y \cap X| \geq \alpha n$. 
Let $Y \subseteq \bbQ^d$ be a set of $n$ points such that $\frac{1}{n} \sum_i y_i y_i^{\top} = I$. 
Then, Algorithm~\ref{algo:list-decoding-covariance-spectral} on input $Y$ outputs a list of $\alpha^{-\tilde{O}(1/\alpha^{10})}$ covariance matrices such that there is a candidate $\hat{\Sigma}$ in the list satisfying:
\begin{align}
\Sigma_* \preceq \hat{\Sigma} \preceq  O\Paren{\frac{1}{\alpha^{150}}} \Sigma_* \mper \label{eq:main-bound-spectral-rec}
\end{align} 
\end{lemma}



\begin{proof}
Let us call a $g+1$-th generation candidate $(\hat{\Sigma}_j,Y^{(j)},g)$ good if it satisfies 1) $\Sigma_* \preceq \hat{\Sigma}_i \preceq O(1/\alpha^{150})\Sigma_* + O(\alpha^{16})\E_{y \sim Y^{(i)}} yy^{\top}$, 2) $|Y^{(i)} \cap X| \geq (\alpha - g \cdot O(\alpha^{10}))n$, 3) $|Y^{(i)}| \leq (1-O(\alpha^{10}))^g n$. Here $(\hat{\Sigma}_i,Y^{(i)},g)$ is the $g$-th generation candidate processing of which generated $\hat{\Sigma}_j$ as a candidate in the $g$-th iteration of the while loop in Algorithm~\ref{algo:list-decoding-covariance-spectral}. 

We will prove the following by induction on $g$: suppose there is a $g$-th generation candidate that is good and not marked final. Then, there is a $g+1$th generation candidate that is good.     

We first observe that this claim is enough to complete the proof. To see why, observe that the number of generations $g$ is no more than $O(\log 1/\alpha)/\alpha^{10}$. So $|Y^{(i)} \cap X| \geq (\alpha/2) n$ for all $g$ encountered in the run of the algorithm. Now consider a good candidate $(\hat{\Sigma}_i,Y^{(i)},1)$ in generation $g =1$ and let's track the sequence of good candidates guaranteed by the inductive claim above for each of the generations $g>1$ starting with $(\hat{\Sigma}_i,Y^{(i)},1)$. Let $g_*$ be the largest $g$ such that the good candidate in generation $g$ is not marked final. Since $g_* \leq \tilde{O}(1/\alpha^{10})$ and $\nu = \alpha^{10}/100$ and the assumptions on such $Y^{(i)}$ for Algorithm~\ref{algo:coarse-spectral-recovery} succeeding are met, each run of of Algorithm~\ref{algo:coarse-spectral-recovery} along such a path succeed with probability at least $1-\nu$. By a union bound, and that $g_* \leq \tilde{O}(1/\alpha^{10})$, all the runs succeed with probability at least $0.99$. Let's condition on this event. Then, in iteration $g_*$, starting with such a good candidate, the coarse spectral recovery algorithm must produce a generation $g_*+1$ candidate that is marked final. In which case, we must have that all eigenvalues of $\hat{\Sigma}_j$ are at least $\eta = O(\alpha^{6})$ relative to $\E_{y \sim Y^{(j)}} yy^{\top}$. Since $|Y^{(j)} \cap X| \geq \alpha/2$, by Lemma~\ref{lem:certificates-for-spectral-recovery}, we can conclude that $\hat{\Sigma}_j \preceq O(1/\alpha^{150})\Sigma_*$. This completes the proof modulo the inductive claim. 

Let us now prove the inductive claim to finish the proof.

For the base case, observe that the first iteration runs with the only $0$-th generation candidate in the list $\cL'$, namely, $(\frac{1}{\alpha^2}I, Y,0)$ and $|Y \cap X| \geq \alpha n$. Let us now analyze the steps of the algorithm when $Y^{(i)} = Y$ is processed in $g=1$st iteration. In the first (subgaussian restriction) step, we apply Algorithm~\ref{thm:subgaussian-restriction} which, from Lemma~\ref{thm:subgaussian-restriction}, allows us to obtain a $Y' \subseteq Y$ such that $Y'$ satisfies certifiable subgaussianity and $Y'$ satisfies $|Y' \cap X| \geq (\alpha - \alpha^{10})n$. Thus, our coarse spectral recovery algorithm (Algorithm~\ref{algo:coarse-spectral-recovery}) gets input a $(1-\alpha')$ corrupted sample for $\alpha' \geq \alpha - \alpha^{10}$ fraction of $Y$ and as a result of Theorem~\ref{thm:recover-large-evs}, with probability at least $1-\nu$, returns a list of candidates one of which, say $\hat{\Sigma}_j$, satisfies the first conclusion of the lemma. If this candidate is not marked final, then, by Lemma~\ref{lem:splitting-algo-correctness}, we know that $(\hat{\Sigma}_j, Y^{(j)})$ is added to $\cL'$ with $Y^{(j)}$ satisfying (2) and (3) as desired. The analysis of the inductive case is entirely analogous. 
\end{proof}

\section{List-Decodable Mean and Covariance Estimation}
In this section, we prove that given a good spectral estimate of the covariance (i.e., guaranteed by Theorem~\ref{thm:list-decoding-covariance-spectral-guarantee}) and an $(1-\alpha)$-corruption of a good set of points $X$, we can obtain a list of $O(1/\alpha^{\poly(1/\alpha)})$-candidates that contains an estimate of the mean that is accurate within $\poly(1/\alpha)$-Mahalanobis distance and covariance that is accurate in (stronger) relative Frobenius distance. The stronger guarantee immediately implies our main theorem on list-decoding mean and covariance for Gaussian distributions with a total variation error guarantee. 

We first start by describing the strong relative Frobenius error guarantee for covariance estimation.

\subsection{Covariance Recovery in Relative Frobenius Error} Our algorithm builds on ideas in the prior work~\cite{DBLP:conf/soda/BakshiK21} on list-decodable subspace recovery (which can be thought of as the special case where the unknown covariance is allowed to have eigenvalues that are either $0$ or $1$). 

Our notion of \emph{good} set for Frobenius error guarantee (under additional spectral closeness hypothesis) is significantly weaker:

\begin{definition}[Good set for relative Frobenius Recovery for known spectral approximation] \label{def:good-set-rel-frob}
We say that a subset $X \subseteq \R^d$ is a $C$-good with mean $\E_{x \sim X} x = \mu_*$ and 2nd moment $\E_{x \sim X} xx^{\top} = \Sigma_*$ if 1) $\mu_* \mu_*^{\top} \preceq 0.1 \E_{x \sim X} (x-\mu_*)(x-\mu_*)^{\top}$ and 2) $X$ has $O(1)$-certifiably $C$-hypercontractive degree $2$ polynomials.
\end{definition}

\begin{theorem}[List-decoding covariances with relative Frobenius error guarantee] \label{thm:cov-estimation-rel-frob}
There is an algorithm that takes input $Y\subseteq \R^d$ of size $n$, runs in time $n^{O(1)}$ and outputs a list of PSD matrices $\hat{\Sigma}_1, \hat{\Sigma}_2, \ldots, \hat{\Sigma}_k$ for $k =O(1/\alpha^4)$ with the following guarantees.
Let $X$ be a $C$-\emph{good} set of $n$ points in $\R^d$ such that $\E_{x \sim X} xx^{\top} = \Sigma_*$ satisfying $I \preceq \Sigma_* = \E_{x \sim X} xx^{\top} \preceq O(1/\alpha^{150}) I$. 
Suppose $Y$ be an $(1-\alpha)$-corruption of $X$, i.e., $Y \subseteq \R^d$ of size $n$ satisfying $|Y \cap X| = \alpha n$. 
Then, there is an $i$ such that $\Norm{\Sigma_*^{-1/2}(\widehat{\Sigma}_i - \Sigma_*)\Sigma_*^{-1/2}}^2_F \leq O(1/\alpha^{304})$.
\end{theorem}

\paragraph{Algorithm}Our algorithm solves the SoS relaxation of the constraints $\cA_1' \cup \cA_2 \cup \cA_3 \cup \cA_4$ defined in Section~\ref{sec:coarse-spectral-recovery} -- we have dropped the anti-concentration constraints $\cA_5$ and will modify $\cA_1$ to replace the third constraint with $(\Sigma-I) = VV^{\top}$ and $(O(\frac{1}{\alpha^{150}})I-\Sigma) = ZZ^{\top}$ for matrix valued indeterminates $V$ and $Z$ that encode the additional information that $I \preceq \Sigma \preceq O(1/\alpha^{150})I$. 

The proof of Lemma~\ref{lem:feasibility-coarse-spectral-recovery} extends to show that setting $X' = X$ and $w_i$ to be the indicator of $i$ such that $y_i = x_i$ along with appropriate values to $U,V,Z$ gives a feasible solution to the polynomial constraints $\cA_1' \cup \cA_2 \cup \cA_3 \cup \cA_4$. 

Our full algorithm is as follows:
\begin{enumerate}
  \item Find a pseudo-distribution $\tzeta$ of degree $O(1)$ consistent with $\cA = \cA_1' \cup \cA_2 \cup \cA_3 \cup \cA_4$ and minimizing $\Norm{\pE_{\tzeta}[w]}_2^2$. 
  \item \textbf{Rounding: } Repeat $O(1/\alpha^4)$ times: \begin{enumerate} \item choose $(i_i,i_2)$ with probability proportional to $\pE_{\tzeta}[w_{i_1}w_{i_2}]$. 
  \item Output $\pE_{\tzeta}[ w_{i_1}w_{i_2} \Sigma]/\pE_{\tzeta}[w_{i_1}w_{i_2}]$.
\end{enumerate}
\end{enumerate}

The main claim in our analysis is the following lemma (analogous to Lemma~\ref{lem:frob-recovery-subspace}). 

\begin{lemma}[Deriving Frobenius Error Bounds within Low-Degree SoS] \label{lem:frob-error-bound-final}
\begin{align}
\cA \sststile{O(1)}{w,\Sigma} &\Biggl\{ w(X')^{2}\Norm{\Sigma_*^{-1/2}(\Sigma - \Sigma_*)\Sigma_*^{-1/2}}_F^{2} \leq O(1/\alpha^{300}) \Biggr\} \mper
\end{align}

\end{lemma}
\begin{proof}

From the conclusion of Lemma~\ref{lem:frob-bound-helper-lemma} (with $h=1$):
\begin{multline}
\cA_5 \sststile{O(1)}{\Sigma,R,Q,w} \Biggl\{w(X')^{4}\Iprod{\Sigma_*^{-1/2}(\Sigma - \Sigma_*)\Sigma_*^{-1/2}, Q }^{4} = w(X')^{4}\Iprod{\Sigma - \Sigma_*, \Sigma_*^{-1/2}Q \Sigma_*^{-1/2} }^{4}  \\\leq O(1) w(X')^{4} \Paren{\Norm{Q}_F^{4} +\Norm{ R  \Sigma_*^{-1/2} Q  \Sigma_*^{-1/2} R}_F^{4} } \Biggr\} \mper
\end{multline}


For the second term, we start from $\cA_1'$ using that $I \preceq \Sigma_*$ and thus $\Norm{\Sigma_*^{-1/2}}_2 \leq 1$, we must have:
\begin{multline}
\cA \sststile{O(1)}{R,w} \Biggl\{ w(X')^{2} (v^{\top}  \Sigma_*^{-1/2}R^{\top} R \Sigma_*^{-1/2} v)^{2} =  w(X')^{2} (v^{\top}  \Sigma_*^{-1/2}\Sigma \Sigma_*^{-1/2}  v)^{2} \\\leq O(1/\alpha^{600})  w(X')^{2} \Norm{\Sigma_*^{-1/2} v}_2^{4} \leq O(1/\alpha^{600})  w(X')^{2} \Norm{v}_2^4 \Biggr\}\mper
\end{multline}

Using Contraction within SoS (Fact~\ref{fact:sos-contraction}) with $\beta = w(X')^{2}$, $A = R$, $t = 2$, and $\Delta =O(1/\alpha^{600})$ twice: 
\begin{equation}
  \cA \sststile{O(h)}{R,w,Q} \Set{ w(X')^{4} \Norm{R \Sigma_*^{-1/2} Q \Sigma_*^{-1/2} R}_F^{4} \leq O(1/\alpha^{600}) \Norm{Q}_F^{4}}
  \end{equation}

Plugging back the estimate from \eqref{eq:final-bound-subspace-2} in \eqref{eq:main-bound-subspace} gives:
\begin{equation}
\cA_5 \sststile{O(1)}{\Sigma,w,Q} \Biggl\{ w(X')^{4}\Iprod{\Sigma_*^{-1/2}(\Sigma - \Sigma_*)\Sigma_*^{-1/2},Q}^{4} \leq O(1/\alpha^{600}) \Norm{Q}_F^4\Biggr\} \mper 
\end{equation}

Substituting $Q = \Sigma_*^{-1/2}(\Sigma - \Sigma_*)\Sigma_*^{-1/2}$ and multiplying by the SoS polynomial $w(X')^4$ yields:
\begin{equation}
\cA_5 \sststile{O(1}{\Sigma,w} \Biggl\{ w(X')^{8}\Norm{\Sigma_*^{-1/2}(\Sigma - \Sigma_*)\Sigma_*^{-1/2}}_F^{8} \leq w(X')^4 \cdot O(1/\alpha^{600}) \Norm{\Sigma_*^{-1/2}(\Sigma - \Sigma_*)\Sigma_*^{-1/2}}_F^4
 \Biggr\} \mper 
\end{equation}

We now apply Lemma~\ref{lem:sos-cancel} (Cancellation within SoS) with $A = w(X')^{4} \Norm{\Sigma-\Sigma_*}_F^{4}$ to obtain that:

\begin{align}
\cA \sststile{O(1)}{w,\Sigma} &\Biggl\{ w(X')^{16}\Norm{(\Sigma_*^{-1/2}\Sigma - \Sigma_*)\Sigma_*^{-1/2}}_F^{16} \leq O(1/\alpha^{2400}) \Biggr\} \mper
\end{align}

We finally apply Cancellation with Constant RHS (Lemma~\ref{lem:sos-cancel-basic}) to conclude that:
\begin{align}
\cA \sststile{O(1)}{w,\Sigma} &\Biggl\{ w(X')^{2}\Norm{\Sigma_*^{-1/2}(\Sigma - \Sigma_*)\Sigma_*^{-1/2}}_F^{2} \leq O(1/\alpha^{300}) \Biggr\} \mper
\end{align}

\end{proof}

\begin{proof}
Observe that by Lemma~\ref{fact:partial-pseudo-expecting}, we know that $I \preceq \frac{\pE[w_{i_1} w_{i_2} \Sigma_*]}{\pE[w_{i_1}w_{i_2}]} \preceq O(1/\alpha^{150})I$. Next, let $G \subseteq [n]$ be the set of indices $i$ such that $x_i = y_i$ (unknown to the algorithm). 
Taking pseudo-expectations with respect to $\tzeta$ of the conclusion of Lemma~\ref{lem:frob-error-bound-final}, we obtain that:
\[
\sum_{i_1, i_2 \in G} \pE_{\tzeta}[w_{i_1}w_{i_2}\Norm{\Sigma_*^{-1/2}(\Sigma-\Sigma_*)\Sigma_*^{-1/2}}_F^2 \leq O(1/\alpha^{300})]\mper
\]
Dividing both sides by $\pE[w(X')^2]$, using that $\sum_{i_1, i_2 \in G} \pE[w_{i_1}w_{i_2}] = \pE[w(X')^2]$ and the conclusion of Lemma~\ref{lem:high-entropy-pseudo-distributions} along with Cauchy-Schwarz inequality for pseudo-distributions that yields that $\pE[w(X')^2] \geq \pE[w(X')]^2 \geq \alpha^4$, we obtain:
\[
\frac{1}{\pE[w(X')^2]} \sum_{i_1, i_2 \in G} \pE_{\tzeta}[w_{i_1}w_{i_2}]\Norm{\Sigma_*^{-1/2}(\frac{\pE_{\tzeta}[w_{i_1}w_{i_2}\Sigma]}{\pE_{\tzeta}[w_{i_1}w_{i_2}]}-\Sigma_*)\Sigma_*^{-1/2}}_F^2 \leq O(1/\alpha^{304})\mper
\]
The left hand side can now be interpreted as expectation over the choice of $i_1, i_2$ with probability proportional to $\pE_{\tzeta}[w_{i_1}w_{i_2}]$ conditioned on $i_1, i_2, \in G$. By Markov's inequality, with probability at least $0.99$, a draw from this distribution of $(i_1, i_2)$ must satisfy: 
\[
\Norm{\Sigma_*^{-1/2}(\frac{\pE_{\tzeta}[w_{i_1}w_{i_2}\Sigma]}{\pE_{\tzeta}[w_{i_1}w_{i_2}]}-\Sigma_*)\Sigma_*^{-1/2}}_F^2 \leq O(1/\alpha^{304}) \mper
\]

Further, since $\pE[w(X')] \geq \alpha^2$, the probability that $(i_1,i_2)$ chosen with probability proportional to $\pE_{\tzeta}[w_{i_1}w_{i_2}]$ satisfy that $i_1,i_2 \in G$ is at least $\alpha^4$. Thus, altogether, the chance that a random draw of $(i_1, i_2)$ yields an estimate $\frac{\pE_{\tzeta}[w_{i_1}w_{i_2}\Sigma]}{\pE_{\tzeta}[w_{i_1}w_{i_2}]}$ satisfying the relative Frobenius error bound above is at least $O(\alpha^4)$. Thus, repeating the sampling process $O(1/\alpha^4)$ times is sufficient to ensure that the list contains a candidate close in relative Frobenius distance as desired with probability at least $0.99$.

\end{proof}

\subsection{Mean Estimation Given Spectral Approximation to Covariance} Given a multiplicative spectral approximation to the covariance, one can apply algorithms from prior works (or a significantly simpler variant of our list-decoding algorithm for covariance estimation above) to obtain good estimates of the mean. We only note the consequence here. 

Technically speaking, Theorem 1.2 in~\cite{KothariSteinhardt17} works for the weaker model of list-decodable estimation where $Y$ must contain as a subset $\alpha n$ i.i.d. points from a certifiably subgaussian distribution $D$. For our stronger model, we observe that any subst of size $\alpha n$ of a good set (Definition~\ref{def:good-set-coarse}) satisfies $O(C/\alpha)$-certifiable subgaussianity. To do this, we only need the following basic observation (applied to centered version of a good set):

\begin{lemma}
Suppose $X \subseteq \bbQ^d$ is a set of $n$ points with mean $0$ such that $\Pr_{x \sim X}[ |\iprod{x,v}| \leq \alpha/100 \E_{x \sim X} \iprod{x,v}^2] \leq \alpha/2$. Let $Z \subseteq X$ be any subset of $X$ of size $\alpha n$. Then, $\E_{z \in Z} zz^{\top} \succeq (\alpha/200) \E_{x \sim X} xx^{\top}$. 
\end{lemma}
\begin{proof}
Fix $v \in \R^d$. Apply the anti-concentration of the set $X$ to conclude that at most $\alpha n/2$ points in $X$ can satisfy $\iprod{x,v}^2 \leq \alpha/100 \E_{x \sim X} \iprod{x,v}^2$. Thus, at least $\alpha n/2$ points in $Z$ satisfy $\iprod{x,v}^2 > \alpha/100 \E_{x \sim X} \iprod{x,v}^2$ and as a result, $\E_{z \sim Z} zz^{\top} \succeq \alpha/100 \E_{x \sim X} xx^{\top}$. 
\end{proof}

The above lemma immediately yields that for any subset $Z \subseteq X$ of size at least $\alpha n$, whenever $X$ is $2t$-certifiably $C$-subgaussian, $Z$ itself is $2t$-certifiably $O(C/\alpha)$-subgaussian:
\[
\sststile{2t}{} \Set{ \E_{z \sim Z} \iprod{z,v}^{2t} \leq \frac{1}{\alpha} \E_{x \sim X} \iprod{x,v}^{2t} \leq \frac{1}{\alpha} (Ct)^t \Paren{\E_{x \sim X} \iprod{x,v}^2}^t \leq \frac{1}{\alpha} (200Ct/\alpha)^t \Paren{\E_{z \sim Z} \iprod{z,v}^2}^t }\mper
\]

We can thus apply Theorem 1.2 in~\cite{KothariSteinhardt17} (at the cost of loss of an additional factor of $O(1/\alpha)$ in the error).

\begin{theorem}[List-decoding mean estimation given spectral approximation to covariance, Theorem 1.2 in~\cite{KothariSteinhardt17}] \label{fact:list-decodable-mean-est}
There is an algorithm that takes input $Y\subseteq \R^d$ of size $n$, runs in time $n^{O(\log(1/\alpha))}$ and outputs a list of $d$-dimensional vectors $\hat{\mu}_1, \hat{\mu}_2, \ldots, \hat{\mu}_k$ for $k =O(1/\alpha)$ with the following guarantees.
Let $X$ be a $C$-\emph{good} set of $n$ points in $\R^d$ such that $\E_{x\sim X} x = \mu_*$ and $\E_{x \sim X} xx^{\top} = \Sigma_*$ satisfying $I \preceq \Sigma_* = \E_{x \sim X} xx^{\top} \preceq O(1/\alpha^{150}) I$. 
Suppose $Y$ be an $(1-\alpha)$-corruption of $X$, i.e., $Y \subseteq \R^d$ of size $n$ satisfying $|Y \cap X| = \alpha n$. 
Then, there is an $i$ such that for every $u \in \R^d$, $\Norm{\widehat{\mu}_i - \mu_*,u}^2_2 \leq O(\log(1/\alpha)/\alpha^{152}) u^{\top}\Sigma_* u^2$.
\end{theorem}




\subsection{Proof of Main Theorem} 
We now have all the components to prove Theorem~\ref{thm:main-intro}. 

\begin{theorem}[Main Theorem] \label{thm:list-decodable-cov-mean-estimation-section}
Fix $\alpha >0$. For any $B$, there is a $(Bn)^{\tilde{O}(1/\alpha^{12})}$ time algorithm that takes input a collection of $n$ points $Y \subseteq \bbQ^d$ with entries of bit-complexity at most $\poly(Bd)$ and either ``rejects'' or outputs a list of parameters $\{(\hat{\mu}_i,\hat{\Sigma}_i)\}_{i \leq k}$ for $k = \alpha^{-\poly(1/\alpha)}$ with the following guarantee: suppose that for some absolute constant $C>0$, $D$ is a distribution on $\R^d$ with mean $\mu_*$ and covariance $\Sigma_*$ with rational entries of bit complexity $\leq B$ that is 1) $s(\delta)$-certifiably $(C,\delta)$-anti-concentrated for $\delta = O(\alpha^3)$,  $s(\delta) = O(1/\delta^2)$ and 2) has $2t$-certifiable $C$-hypercontractive degree $2$ polynomials for $t \geq O(1/\alpha)$. Suppose $Y$ is a $\poly(d)$-bit rational truncation of an $\epsilon$-corrupted sample from $D$ of size $n \geq n_0 = d^{\tilde{O}(1/\alpha^6)}$. 

Then, with probability at least $0.99$ over the draw of $X$ and over the random choices of the algorithm, the algorithm does not reject and outputs a list of parameters of size $k$ such that there exists an $i$ such that for every $u \in \R^d$:
\[
\iprod{\hat{\mu}_i - \mu_*, u} \leq \tilde{O}(\log (1/\alpha)/\alpha^{152}) \sqrt{u^{\top} \Sigma_* u} \mcom
\]
\[
\Sigma_* \preceq \hat{\Sigma} \preceq O(1/\alpha^{152}) \Sigma_*\mcom
\]
and,
\[
\Norm{\Sigma_*^{\dagger/2} (\hat{\Sigma} - \Sigma_*) \Sigma_*^{\dagger/2}}_F \leq O(1/\alpha^{304}) \mper
\]
\end{theorem}

Combined with the characterization of total variation distance in terms of the three parameter distance bounds (Proposition A.1 in ~\cite{bakshi2020mixture}), we immediately obtain Theorem~\ref{thm:main-intro}.
\begin{corollary}[List-decoding Gaussian with total variation error guarantee]
There is a $n^{\poly(1/\alpha)}$ time algorithm that takes input a $(1-\alpha)$-corrupted sample of size $n \geq d^{\alpha^{-O(1)}}$ from a $d$-dimensional Gaussian distribution with mean $\mu_*$ and covariance $\Sigma_*$ and outputs a list of $2^{O(1/\alpha^{O(1)})}$-parameters such that there is a $(\hat{\mu},\hat{\Sigma})$ in the list satisfying:
\[
\dtv(\cN(\hat{\mu},\hat{\Sigma}), \cN(\mu_*,\Sigma_*)) \leq 1-\exp(-O(1/\alpha^{304})) \mper
\] 
\end{corollary}
\begin{proof}[Proof of Theorem~\ref{thm:list-decodable-cov-mean-estimation-section}]
Let $X$ be an i.i.d. sample from $D$ and let $\tilde{X}$ is a truncation of $D$ to $\poly(Bd)$ bits. If $\Sigma_* \succeq 2^{-\poly(Bd)}I$, then, Fact~\ref{fact:sampling} shows that $\tilde{X}$ has small mean (i.e. $\E_{x \sim \tilde{X}} x = \tilde{\mu}$ such that $\tilde{\mu} \tilde{\mu}^{\top} \preceq 0.1 \E_{x \sim \tilde{X}} (x-\tilde{\mu})(x-\tilde{\mu})^{\top}$), covariance in $[0.99,1.01] \cdot \Sigma_*$, satisfies the properties of good set required in Definition~\ref{def:good-set-coarse}.

Our algorithm works in three steps. 

In the first step, we list-decode covariances with spectral guarantee using the Algorithm from Theorem~\ref{thm:list-decoding-covariance-spectral-guarantee}. Observe that Theorem~\ref{thm:list-decoding-covariance-spectral-guarantee} requires that input sample be from a \emph{small-mean} good set $X$. To meet these guarantees, we work with the ``pairwise difference'' version of $Y$. 

Assume $n$ is even, randomly permute the points in $Y$ (and assume, for the sake of simplicity that $y_1, y_2,\ldots, y_n$ is the permuted version) and let $Y'$ be the set of $n/2$ points $\frac{y_1-y_2}{\sqrt{2}}$, $\frac{y_3-y_4}{\sqrt{2}}$,  \ldots $\frac{y_{n-1}-y_n}{\sqrt{2}}$. Then, observe that with probability at least $1-1/n$ over the choice of the random permutation, there are at least $\alpha^2 n/2$ pairs $(2i-1,2i)$ such that $y_{2i-1}, y_{2i}$ are both in the intersection $Y \cap X$. Without loss of generality, let's say that $y_{2i-1} = x_{2i-1}$ and $y_{2i} = x_{2i}$. Thus, $Y'$ can be thought of as an $(1-\alpha^2/2)$-corruption of randomly paired and $\frac{1}{\sqrt{2}}$-scaled differences, say, $\tilde{X}$ from $X$. Thus, Theorem~\ref{thm:list-decoding-covariance-spectral-guarantee} guarantees that with probability at least $0.99$, there is an element in the list of size $\alpha^{-\poly(1/\alpha)}$ output by it that contains a PSD matrix $\hat{\Sigma}$ that satisfies $\Sigma_* \preceq \hat{\Sigma} \preceq O(1/\alpha^{150})\Sigma_*$. 

In the 2nd step, we take $\alpha^{-150} \hat{\Sigma}_i$ for each element $\hat{\Sigma}_i$ of the list obtained in the first step and transform $Y'$ by applying the linear transformation $y_i' \rightarrow \hat{\Sigma}_i^{-1/2}$. When applied to $\tilde{X}$, this transformation ensures that the covariance of $\tilde{X}$ is sandwiched (in Löwner order) between $I$ and $O(1/\alpha^{150})$. Further, since $\tilde{X}$ satisfies the properties of a good set in Definition~\ref{def:good-set-coarse} and certifiable hypercontractivity of degree $2$ polynomials is invariant under linear transformation, the linearly transformed $\tilde{X}$ continues to satisfy the requirements of Definition~\ref{def:good-set-rel-frob}. Running the algorithm from Theorem~\ref{thm:cov-estimation-rel-frob} for each possible candidate $\hat{\Sigma}_i$ from the list obtained in the first step enlarges the list by a factor of $O(1/\alpha^4)$ and when starting with a good candidate $\hat{\Sigma}_i$ from the first step, Theorem~\ref{thm:cov-estimation-rel-frob} guarantees with probability at least $0.99$, that there is a candidate $\hat{\Sigma}_j$ in the enlarged list satisfying $\Norm{\Sigma_*^{-1/2}\hat{\Sigma}_j\Sigma_*^{-1/2}-I}_F^2 \leq O(1/\alpha^{304})$. 

In the final step, we take $\alpha^{-150} \hat{\Sigma}_i$ for each element $\hat{\Sigma}_i$ of the list obtained in the first step (we do not need Frobenius guarantees for mean estimation) and transform $Y'$ by applying the linear transformation $y_i' \rightarrow \hat{\Sigma}_i^{-1/2}$ and observe that by the same argument as in the analysis of the 2nd step above, the assumptions of Fact~\ref{fact:list-decodable-mean-est} are met and thus, we obtain an list with estimates of means of size $O(1/\alpha)$ factor larger than the one in Step 1 guaranteed with probability at least $0.99$ to contain a candidate $\hat{\mu}_i$ satisfying $\iprod{\hat{\mu}_i - \mu_*,u} \leq O(\log (1/\alpha)/\alpha^{152}) \sqrt{u^{\top} \Sigma_* u}$ for every $u \in \R^d$. 

All 3 steps succeed with probability at least $0.9$ by a union bound.  Returning every possible paired combination of the covariances from Step 2 and means from step 1 then satisfies the requirements of the theorem with a list of size $2^{\poly(1/\alpha)}$. 

\end{proof}

\section{Applications}

In this section, we derive improved algorithms for list-decodable linear regression, subspace recovery and clustering of non-spherical mixtures as immediate consequences of our algorithm for list-decodable covariance estimation. 

\subsection{Linear Regression}
For list-decodable linear regression, given a accuracy parameter $\eta$, the best known prior works~\cite{DBLP:conf/nips/KarmalkarKK19,RY19} obtain a list containing an $\eta$-accurate estimate of the unknown vector with a running time of $n^{O(1/(\eta\alpha)^{4})}$ and sample complexity $d^{O(1/(\eta\alpha)^{4})}$. An error reduction technique from~\cite{DBLP:conf/soda/BakshiK21} allows improving both the exponents to $O(\log 1/\eta)/\alpha^4$. But all these bounds depend exponentially on the target accuracy $\eta$. As a result whenever $\eta \rightarrow 0$ as $d \rightarrow \infty$, the sample complexity and the running time are both super-polynomial in the underlying dimension $d$.

Our list-decodable covariance estimation algorithm allows obtaining the first \emph{exact} algorithm for list-decodable linear regression. As a consequence, we can obtain an error of $\eta$ time $n^{\tilde{O}(1/\alpha^{12})} \poly \log (1/\eta)$ and sample complexity $d^{\tilde{O}(1/\alpha^6)}$. In particular, the sample complexity does not depend on the target accuracy and the running time scales polylogarithmically in $1/\eta$. As a result, our algorithm allows obtaining exponentially small error in the underlying dimension $d$ in polynomial time. The size of the list recovered, while still an absolute constant depending only on $\alpha$, is larger and grows as $\alpha^{-\poly(1/\alpha)}$. Our algorithm works without knowing (any upper bound on) the length of the unknown vector $\ell_*$ and extends easily to the setting of unknown arbitrary non-spherical covariance (if we simply apply our list-decodable covariance estimation as a preprocessing step). It also succeeds in the strong contamination model for list-decodable learning as opposed to the additive model studied in the prior works.

\begin{corollary} \label{cor:improved-list-decodable-linear-regression}
For every $\eta>0$, there is an algorithm that takes input a collection of $n$ equations $Y \subseteq \bbQ^d \times \bbQ$, runs in time $n^{\tilde{O}(1/\alpha^{12})} \poly \log (1/\eta)$ and either outputs ``reject'' or a list $\cL$ of size $k =\alpha^{-\poly(1/\alpha)}$ of candidate vectors $\hat{\ell}_1, \hat{\ell}_2, \ldots, \hat{\ell}_k$ with the following guarantees: suppose $X$ is a set of $n \geq n_0 = O(d^{O(1/\alpha^6)}/\alpha)$ linear equations $\langle a, \ell_* \rangle = b$ where $\ell_*$ is an unknown arbitrary vector and $a \sim D$ on $\R^d$ such that $D$ has mean $0$, a full-rank covariance $\Sigma_*$, is $s(\delta)$-certifiably $(C,\delta)$-anti-concentrated and has $2t$-certifiable $C$-hypercontractivity of degree $2$ polynomials for all $t$. Suppose $Y$ is an $(1-\alpha)$-corruption of $X$. Then, with probability at least $0.99$ over the draw of $X$ and the random choices of the algorithm, the algorithm does not reject and outputs a list that contains a candidate $\hat{\ell}_k$ satisfying:
\[
\Norm{\Sigma_*^{-1/2}(\hat{\ell}_k - \ell_*)}_2 \leq \eta \mper
\]
\end{corollary}

\begin{proof}[Proof Sketch]
We observe that if $(a,b) \in \bbQ^{d} \times \bbQ$ are the coefficient vectors and ``right-hand-sides'' of the equations in $X$, then, the distribution $D'$ of $(a,b)$ satisfies the conditions in Definition~\ref{algo:coarse-spectral-recovery}. Further, the covariance of $D'$ has rank exactly $d$ (in $d+1$ dimensional ambient space) with the kernel in the direction $(\ell_*,-1)$. Thus, finding a multiplicative spectral approximation to the covariance of $D'$ and using the kernel of the estimate to obtain $\hat{\ell}$ immediately gives the required guarantee. 

Observe that we only pay (and poly logarithmically so) in the running time for the target accuracy. There is no cost in sample complexity as a function of the target accuracy. 
\end{proof}

\subsection{Subspace Recovery}
A similar argument also upgrades the guarantees for list-decodable subspace recovery obtained in prior works. The best known prior work~\cite{DBLP:conf/soda/BakshiK21} obtained an algorithm that runs in fixed polynomial time (exponent independent of $\alpha$) and gets an Frobenius estimation error of $O(1/\alpha)$. The independent work~\cite{raghavendra2020list} obtains a worse error guarantee that grows as $\sqrt{r}$ (where $r$ is the dimension of the unknown subspace). In particular, for obtaining an arbitrary target error $\eta>0$, the algorithm from~\cite{DBLP:conf/soda/BakshiK21} runs in time $n^{\log (1/\alpha \eta) O(1/\alpha^4)}$ that is superpolynomial for any $\eta \rightarrow 0$.

Our result below obtains an algorithm that runs in time $n^{\tilde{O}(1/\alpha^{12})} \poly\log (1/\eta)$. This, in particular, allows obtaining error $\eta$ as small as $2^{-d}$ in polynomial time in the dimension $d$. Our list-size however is $\alpha^{-\poly(1/\alpha)}$ compared to $O(1/\alpha^{ \log (1/\alpha) + \log (1/\eta)})$ in ~\cite{DBLP:conf/soda/BakshiK21}.

\begin{corollary} \label{cor:improved-list-decodable-subspace-recovery}
For any $\eta>0$, there is an algorithm that takes input a collection of $n$ points $Y \subseteq \bbQ^d$, runs in time $n^{\tilde{O}(1/\alpha^{12})} \poly \log (1/\eta)$ and either outputs ``reject'' or a list $\cL$ of size $k =\alpha^{-\poly(1/\alpha)}$ of candidate projection matrices $\hat{\Pi}_1, \hat{\Pi}_2, \ldots, \hat{\Pi}_k$ with the following guarantee: suppose $X$ is a set of $n \geq n_0 = O(d^{O(1/\alpha^6)}/\alpha)$ i.i.d. draws from a distribution $D$ on $\R^d$ such that $D$ has mean $0$, covariance $\Pi_*$ -- a projection matrix to a subspace of $\R^d$, is $s(\delta)$-certifiably $(C,\delta)$-anti-concentrated and has $2t$-certifiable $C$-hypercontractivity of degree $2$ polynomials for all $t$. Suppose $Y$ is an $1-\alpha)$-corruption of $X$. Then, with probability at least $0.99$ over the draw of $X$ and the random choices of the algorithm, the algorithm does not reject and outputs a list that contains a candidate $\hat{\Pi}_k$ satisfying:
\[
\Norm{\hat{\Pi}_k - \Pi_*}_F \leq \eta \mper
\]
\end{corollary}

\subsection{Clustering Non-Spherical Mixtures} \label{sec:clustering-non-spherical}
Let $M = \sum_i p_i D_i$ be a mixture of $D_1, D_2, \ldots, D_k$ such that for each $i$, $D_i$ is $s(\delta)$-certifiably $C$-anti-concentrated distributions with $2t$-certifiably $C$-hypercontractive degree $2$ polynomials for all $t \in \N$ and $p_i \geq p_{min}$ for each $i$. Then, so long as $\epsilon < p_{min}/2$, an $\epsilon$ corrupted sample from $M$ intersects with an i.i.d. sample from any $D_i$ in at least $p_{min}/4$ points. Thus, we can immediately apply our list-decodable mean and covariance estimation algorithm (Theorem~\ref{thm:list-decodable-cov-mean-estimation-section}) with $\alpha = p_{min}/4$ runs in time $d^{\poly(1/p_{min})}$ and get a list of $(1/p_{min})^{\poly(1/p_{min})}$ candidates such that there is a $\Delta$-close (in parameter distance) mean-covariance pair to $(\mu_i,\Sigma_i)$ for every $i$ for $\Delta = \poly(1/p_{min})$. Note that this consequence does not require any separation assumptions.

If the component $D_i$s are guaranteed to have well-separated parameters (as in the main result in~\cite{bakshi2020mixture}), then we can cluster the input corrupted sample $Y$ with at most $O(\epsilon/p_{min})$ fraction of misclassified points in any cluster. This, in particular, also allows obtaining estimates of the parameters up to $\tilde{O}(\epsilon/p_{min})$ in $\dist$. 

Let us briefly explain this procedure before supplying a more detailed proof sketch. The main idea is simple: suppose we were given the parameters of each $D_i$ \emph{exactly}. We can then apply a natural clustering procedure based on the parameters. We will argue that this natural clustering procedure continues to function even if we have an estimate of the parameters that is accurate to within $\poly(1/p_{min})$ factor in $\dist$ as long as the pairwise separation (again, in $\dist$) between the parameters of $D_i$s is at least $1/p_{min}^{O(k)}$. For the sake of keeping the exposition in this section simple, we only describe the algorithm for the case when $D_i$s are Gaussian distributions.

\begin{theorem}[Robust Clustering of TV-Separated Gaussian Mixtures] \label{thm:clustering-non-spherical-section}
Fix $p_{min}>0, k \in \N, B \in \N$. For every large enough $d \in \N$, there is an algorithm that takes input an $\epsilon$-corruption $Y$ of an i.i.d. sample $X = C_1 \cup C_2 \ldots C_k$ of size $n\geq n_0 = d^{\poly(1/p_{min})}$ from a $d$-dimensional mixture $\sum_i p_i D_i(\mu_i,\Sigma_i)$ of Gaussians with parameters $\mu_i,\Sigma_i$ having rational entries of bit complexity at most $B$ and runs in time $(Bn)^{\poly(1/p_{min})}$. If $\epsilon \leq c p_{min}$ for a small enough constant $c>0$ and $\dist(\cN(\mu_i,\Sigma_i), \cN(\mu_j,\Sigma_j)) \geq p_{min}^{-O(k)}$, the algorithm, with probability at least $0.99$ over the draw of the original uncorrupted sample $X$ and the random choices of the algorithm, outputs a clustering $Y = \hat{C}_1 \cup \hat{C}_2 \cup \ldots \hat{C}_k$ of $Y$ with the property $\min_{\pi:[k] \rightarrow [k]} \max_{i \leq k} (1-|\hat{C}_i \cap C_{\pi(i)}|/|C_{\pi(i)}|) \leq \tilde{O}(\epsilon/p_{min})$. 
\end{theorem}

Notice that in comparison, the algorithm from~\cite{bakshi2020mixture} is presented only for the equiweighted case (i.e., $p_{min} =1/k$), needs $n=d^{k^{O(k)}}$ samples and $n^{k^{O(k)}}$ time, and works only when the fraction of outliers $\epsilon \ll k^{-O(k)}$. The parameters in the algorithm of~\cite{DHKK20} are worse (with, roughly speaking, every instance of $k^{O(k)}$ replaced by a $\poly(k)$ size tower of exponential in $k$ for equiweighted mixture of Gaussians). 

\begin{proof}[Proof Sketch]

First, we apply the algorithm from Theorem~\ref{thm:list-decodable-cov-mean-estimation-section} to $Y$ with parameter $\alpha = p_{min}-\epsilon \geq 0.99p_{min}$. By thinking of the uncorrupted points in the any true cluster $C_i$ as the inliers and all the rest of $Y$ as outliers, we observe that $Y$ is an $1-(p_{min}-\epsilon)$ corruption of $D_i(\mu_i,\Sigma_i)$ and thus, Theorem~\ref{thm:list-decodable-cov-mean-estimation-section} provides an algorithm that runs in time $(Bn)^{\poly(1/p_{min})}$ and generates a list of size $p_{min}^{-\poly(1/p_{min})}$ such that for every $i$, there is a $(\hat{\mu}_i,\hat{\Sigma}_i)$ in the list that is $\Delta = \poly(1/p_{min})$-close in $\dist$ to $(\mu_i,\Sigma_i)$. Call a $k$-tuple of parameters from this list \emph{good} if for every $(\mu_i,\Sigma_i)$, there is a $\Delta$-close $(\hat{\mu}_i,\hat{\Sigma}_i)$ in the $k$-tuple. Notice that by the approximate triangle inequality for $\dist$, every pair of parameters in such a good $k$-tuple must be $\Delta^{100k}$ apart in $\dist$.  

The goal of the algorithm now is to use this list to cluster the input points. Here's how the algorithm proceeds: the algorithm enumerates over all $k$-tuples of parameters in the list that satisfy the property that every pair in the $k$-tuple is at a distance of at least $\Delta^{100k}$ from each other. We will give a procedure to cluster assuming the $k$ set of parameters are good. Running a cluster verification procedure (Lemma~\ref{fact:verification-subroutine}) similar to the one emploiyed in~\cite{bakshi2020mixture} on each cluster so constructed allows verifying whether each cluster satisfies hypercontractivity and anti-concentration properties finishing the algorithm. Thus, the key remaining piece is to establish that if we chose a subset of $k$ parameters from the list that are good, then we can efficiently construct an approximate clustering of $Y$. 

Let's now describe the clustering procedure assuming a good set of known $k$ parameters, say $\{(\hat{\mu}_i,\hat{\Sigma}_i)\}_{i \leq k}$. Our ideas rely on partial cluster recovery procedure employed in~\cite{bakshi2020mixture} that exploits the three kinds of separations possible between mixtures of reasonable distributions (Definition~\ref{def:param-distances}).

First, suppose there is a unit vector $v$ such that $v^{\top} \hat{\Sigma}_i v \leq \Delta^{O(k)} v^{\top} \hat{\Sigma}_j v$ for some $1\leq i <j \leq k$. Such a vector, if it exists, can be found by going over all pairs $i,j$, taking the top eigenvector of $\hat{\Sigma}_j^{-\dagger/2} \hat{\Sigma}_i \hat{\Sigma}_j^{-\dagger/2}$ and applying $\hat{\Sigma}_j^{1/2}$ to it. In this case, we will use the following variance clustering procedure. Observe that there is a partition of $[k]$ into $S, \bar{S}$ such that $v^{\top} \hat{\Sigma}_i v \leq \beta$ for all $i \in S$ and $v^{\top} \hat{\Sigma}_i v \geq \Delta^{O(1)} \beta$ for all $i \in \bar{S}$. Since $(\mu_i,\Sigma_i)$ are $\Delta$-close to $(\hat{\mu}_i,\hat{\Sigma}_i)$, by Lemma~\ref{lem:approx-triangle-dist}, we must thus have that $v^{\top} \Sigma_i v \leq \beta'$ for all $i \in S$ and $v^{\top} \hat{\Sigma}_i v \geq \beta' \Delta^{O(1)}$ where $\beta' = \beta \Delta^{O(1)}$. Our clustering algorithm does the following: for each $y \in Y$, we include $y$ in cluster $L$ if there is an $i \in S$ such that $\frac{\eta}{2 \Delta} v^{\top} \hat{\Sigma}_i v \leq \iprod{y-\hat{\mu}_i, v}^2 \leq O(\log 1/\eta) \Delta^2 v^{\top} \hat{\Sigma}_i v$ for $\eta = \epsilon$. If there is no such $i$, we include $y \in R$. We now claim that $|\cup_{i \in S} C_i \cap L| \geq |\cup_{i \in S} C_i|-2\epsilon n$, and, $|\cup_{i \not \in S} C_i \cap R| \geq |\cup_{i \not \in S} C_i|-2\epsilon n$. To see why, observe first that $v^{\top} \hat{\Sigma}_i v \ll \Delta^{O(1)} v^{\top} \hat{\Sigma}_j v$ for all $i \in S, j \not \in S$. The claim then immediately follows by observing that from Lemma~\ref{lem:approx-isotropization}, at most $\eta$ fraction of $x \in \cup_{i \not \in S }$ get put in $L$ and similarly, at most an $\eta$ fraction of $x \in \cup_{i \in S}$ get put in $R$ -- in particular, we have achieved a partial clustering of the input samples with an error of at most $2\epsilon n$ points on either side. We can repeat the above ``variance clustering'' procedure until there's no vector $v$ satisfying $v^{\top} \hat{\Sigma}_i v \leq \Delta^{O(k)} v^{\top} \hat{\Sigma}_j v$ for some $1\leq i <j \leq k$. Thus, in the following, we can assume that for every $v$ and every $i$, $\Delta^{-O(k)} \frac{1}{k} (\sum_i \hat{\Sigma}_i)v \leq v^{\top} \hat{\Sigma}_i v \leq \Delta^{O(k)} \frac{1}{k} (\sum_i \hat{\Sigma}_i)v$. 

Next, suppose there is a unit vector $v$ such that $\iprod{\hat{\mu}_i-\hat{\mu}_j,v}^2 \geq \Delta^{O(k)} v^{\top} \frac{1}{k} \sum_i \hat{\Sigma}_i v ) v$. Such a vector, if it exists, can be found by going over all pairs $i,j$, and checking if $v = (\frac{1}{k} \sum_i \hat{\Sigma}_i)^{-\dagger/2}(\hat{\mu}_i -\hat{\mu}_j)$ satisfies the inequality above. Given such a vector $v$, we can again partition $[k]$ into two groups, $S$ and $\bar{S}$ such that for every $i \in S, j \not \in S$, $\iprod{\hat{\mu}_i - \hat{\mu}_j,v}^2 \geq \Delta^{O(1)} v^{\top} \frac{1}{k} \sum_i \hat{\Sigma}_i v ) v$. We now do a ``mean-clustering'' procedure as follows: we put $y\in L$ iff there is an $i \in S$ such that $O(\Delta^2 \log 1/\eta) v^{\top} \Sigma' v \geq \iprod{y-\hat{\mu}_i,v}^2 \geq \frac{\eta}{2\Delta} v^{\top} \Sigma' v$ for $\eta = \epsilon$. By an analysis similar to the above, we arrive at a partial clustering as before this time ensuring that every group consists of clusters with mean-close parameters. 

Finally, suppose there is a pair $i,j$ such that $\Norm{\hat{\Sigma}_i - \hat{\Sigma}_j}_F \geq \Delta^{O(k)}$. Then, in particular, for $A = \hat{\Sigma}_i - \hat{\Sigma}_j$, it holds that $\tr(A \cdot (\hat{\Sigma}_i - \hat{\Sigma}_j)) \geq \Delta^{O(k)}$. As before, we find a partitions of $[k]$ into two groups $S$ and $\bar{S}$ such that for every $i \in S, j \not \in S$, $\tr(A \cdot (\hat{\Sigma}_i - \hat{\Sigma}_j) \geq \Delta^{O(1)}$. We now apply a ``Frobenius clustering'' procedure that puts $y \in L$ if there is an $i \in S$ such that $|(y-\hat{\mu}_i)^{\top} A (y-\hat{\mu}_i) - \tr(A \hat{\Sigma}_i)| \geq O(\Delta \log 1/\eta) \Norm{\hat{\Sigma}^{1/2} A \hat{\Sigma}^{1/2}}_F$. Using Lemma~\ref{lem:approx-isotropization} and a similar analysis in the above two cases, we arrive at a partial clustering as before ensuring that every group consists of relative Frobenius close parameters. 

At the end of the three modes of clustering, we end up with partial clustering from the original data with at most $O(k \epsilon n)$ points misclassified in total. Further, within each group, we every pair of clusters that contribute must be within a $\Delta^{O(k)}$ distance in each of the three possible ways of separation and thus, also $\Delta^{O(k)}$-close in $\dist$. Since every pair of $(\mu_i,\Sigma_i)$s are $\gg \Delta^{O(k)}$-far in $\dist$, the resulting groups must in fact be an approximate clustering of the data with at most $O(k\epsilon n)$ misclassified points as desired.

\end{proof}

\begin{lemma}[Approximate Isotropization] \label{lem:approx-isotropization}
For $\Delta\geq 1, \eta> 0$ and two sets of parameters $(\mu,\Sigma), (\mu',\Sigma')$ of $d$-dimensional Gaussian distributions, let $\dist((\mu,\Sigma), (\mu',\Sigma')) \leq \Delta$. Let $x \sim \cN(\mu,\Sigma)$. Then, 
\[
\Pr[ \frac{\eta}{2\Delta} v^{\top} \Sigma' v\leq \iprod{x-\mu',v}^2 \leq O(\Delta^2 \log 1/\eta) v^{\top} \Sigma' v] \geq 1-\eta\mcom
\]
and,
\[
\Pr[ \Abs{ (x-\mu')^{\top}A (x-\mu') - \tr(A \Sigma')} \geq O(\Delta \log 1/\eta) \Norm{{\Sigma'}^{1/2} A {\Sigma'}^{1/2}}_F ] \geq 1-\eta\mper
\]
\end{lemma}
\begin{proof}
We know by subgaussianity and anti-concentration of Gaussian random variables that 
\[
\Pr[ \eta/2 v^{\top} \Sigma v\leq \iprod{x-\mu,v}^2 \leq O( \log 1/\eta) v^{\top} \Sigma' v] \geq 1-\eta\mper
\]

Since $\dist((\mu,\Sigma), (\mu',\Sigma')) \leq \Delta$, we also know that $\frac{1}{\Delta} v^{\top} \Sigma v\leq v^{\top} \Sigma' v \leq \Delta v^{\top} \Sigma v$.

Thus,
\[
\Pr[ \frac{\eta}{2 \Delta} v^{\top} \Sigma' v \leq \iprod{x-\mu,v}^2 \leq O(\Delta \log 1/\eta) v^{\top} \Sigma' v] \geq 1-\eta\mper
\]

Further, $\iprod{\mu-\mu',v}^2 \leq \Delta v^{\top} (\Sigma + \Sigma') v \leq (1+\Delta^2) v^{\top} \Sigma' v$. Thus, 

\[
\Pr[ \frac{\eta}{2 \Delta} v^{\top} \Sigma' v \leq \iprod{x-\mu',v}^2 \leq O(\Delta^2 \log 1/\eta) v^{\top} \Sigma' v] \geq 1-\eta\mper
\]

Next, by tail bounds for degree $2$ polynomials of hypercontractive distributions, we have:
\begin{equation} \label{eq:basic-tail-poly}
\Pr[ \Abs{ (x-\mu)^{\top}A (x-\mu) - \tr(A \Sigma)} \geq O(\log 1/\eta) \Norm{\Sigma^{1/2} A \Sigma^{1/2}}_F ] \geq 1-\eta\mper
\end{equation}
Now, observe that using $\dist((\mu,\Sigma),(\mu',\Sigma')) \leq \Delta$, we have:
\[
\tr(A (\Sigma-\Sigma')) = \tr( {\Sigma'}^{1/2} A {\Sigma'}^{1/2} \cdot ({\Sigma'}^{\dagger/2}\Sigma {\Sigma'}^{\dagger/2}-I))\leq \Norm{{\Sigma'}^{1/2} A {\Sigma'}^{1/2}}_F \Norm{I-{\Sigma'}^{\dagger/2}\Sigma {\Sigma'}^{\dagger/2})}_F \leq \Delta \Norm{{\Sigma'}^{1/2} A {\Sigma'}^{1/2}}_F\mper
\] 
Further, again using $\dist((\mu,\Sigma),(\mu',\Sigma')) \leq \Delta$, we have:
\[
\Norm{\Sigma^{1/2} A \Sigma^{1/2}}_F = \Norm{\Sigma^{1/2}{\Sigma'}^{\dagger/2} ({\Sigma'}^{1/2} A{\Sigma'}^{1/2}){\Sigma'}^{\dagger/2}) \Sigma^{1/2}}_F \leq \Norm{\Sigma^{1/2} {\Sigma'}^{\dagger/2}}_2^2 \Norm{{\Sigma'}^{1/2} A{\Sigma'}^{1/2}}_F \leq \Delta \norm{{\Sigma'}^{1/2} A{\Sigma'}^{1/2}}_F \mper
\] 
Finally, since $\dist((\mu,\Sigma),(\mu',\Sigma')) \leq \Delta$, we have that $\Norm{{\Sigma'}^{1/2}(\mu-\mu')}_2^2 \leq \Delta$ and that $\Norm{{\Sigma'}^{\dagger/2}(x-\mu)^{\top}}^2_2 \leq O(\log 1/\eta) \Delta$ by subgaussianity of $\cN(\mu,\Sigma)$. Thus, for $A' = {\Sigma'}^{1/2} A {\Sigma'}$,  
\begin{align*}
&\Abs{{\Sigma'}^{\dagger/2}(x-\mu')^{\top}A' {\Sigma'}^{\dagger/2}(x-\mu') - {\Sigma'}^{\dagger/2}(x-\mu)^{\top}A' {\Sigma'}^{\dagger/2}(x-\mu)} \\
&\leq |{\Sigma'}^{\dagger/2}(\mu-\mu')^{\top}A' {\Sigma'}^{\dagger/2}(x-\mu)| + |{\Sigma'}^{\dagger/2}(x-\mu)^{\top} A' {\Sigma'}^{\dagger/2}(\mu-\mu')| + |{\Sigma'}^{\dagger/2}(\mu-\mu')^{\top} A' {\Sigma'}^{\dagger/2}(\mu-\mu')|\\
&\leq O(\Delta \log 1/\eta) \Norm{A'}_F\mper
\end{align*} 

Thus, combined with \eqref{eq:basic-tail-poly}, we have:
\[
\Pr[ \Abs{ (x-\mu')^{\top}A (x-\mu') - \tr(A \Sigma')} \geq O(\Delta \log 1/\eta) \Norm{{\Sigma'}^{1/2} A {\Sigma'}^{1/2}}_F ] \geq 1-\eta\mper
\]

\end{proof}

We will use the following cluster verification algorithm from ~\cite{bakshi2020mixture}.
\begin{fact}[Verifying Clusters, analogous to Lemma 6.5 in~\cite{bakshi2020mixture}] \label{fact:verification-subroutine}
There is an algorithm that takes input a set of $n$ $d$-dimensional points $Y$  and a subset $\hat{C} \subseteq Y$, runs in time $n^{\poly(1/p_{min})}$, and outputs acccept or reject with the following guarantee: Suppose $X$ is a good sample from a $\Delta$-separated mixture of Gaussian distributions $\sum_i p_i \cN(\mu_i,\Sigma_i)$ with weights $p_i \geq p_{min}$ for every $i$. Let $Y$ be a $\tau$-corruption of $X$. Let $\hat{C} \subseteq Y$ be such that $|\hat{C} \cap C_i| \leq (1-O(\tau/p_{min})) |C_i|$ for every $i$. Then, the algorithm rejects with probability at least $1-1/\poly(n)$ over the draw of $X$. If, on the other hand, there exists an $i$ such that $|\hat{C} \cap C_i| \geq (1-O(\tau/p_{min})) \max \{ |C_i|, |\hat{C}|\}$, then, the algorithm accepts with probability at least $1-1/\poly(n)$ over the draw of $X$. 
\end{fact}

\begin{lemma}[Approximate Triangle Inequality for Parameter Distance] \label{lem:approx-triangle-dist}
Suppose $\dist((\hat{\mu},\hat{\Sigma}),(\mu_i,\Sigma_i)), \dist((\hat{\mu},\hat{\Sigma}),(\mu_j,\Sigma_j)) \leq \Delta$  for $\Delta >1$. Then, $\dist((\mu_i,\Sigma_i), (\mu_j,\Sigma_j)) \leq 3\Delta^2$. 
\end{lemma}
\begin{proof}
For any vector $v$, we know that $v^{\top} \Sigma_j v \leq \Delta v^{\top} \hat{\Sigma} v \leq \Delta^2 v^{\top} \Sigma_i v$. Similarly, $v^{\top} \Sigma_i v \leq \Delta^2 v^{\top} \Sigma_j v$. This establishes the multiplicative spectral part of the guarantee in $\dist$. 

Next, let's consider the relative Frobenius guarantee. Towards that first observe that $\Norm{\Sigma_i^{-1/2}\hat{\Sigma}^{1/2}v}_2^2 \leq \Norm{\Sigma_i^{-1/2} \hat{\Sigma} \Sigma_i^{-1/2}}_2 \Norm{v}_2^2 \leq (1+\Delta) \Norm{v}_2^2$. Next, because of the multiplicative spectral guarantee, we can assume that $\Sigma_i, \hat{\Sigma}, \Sigma_j$ all have the same range space. We can thus assume that they are all full rank WLOG (as otherwise, we can simply work in their common range space instead). 
\begin{align*}
\Norm{\Sigma_i^{-1/2} \Sigma_j \Sigma_i^{-1/2} -I}_F &\leq \Norm{\Sigma_i^{-1/2} \hat{\Sigma}^{1/2} (\hat{\Sigma}^{-1/2}\Sigma_j-I) \hat{\Sigma}^{-1/2} \hat{\Sigma}^{1/2}\Sigma_i^{-1/2} + \Sigma_i^{-1/2} \hat{\Sigma} \Sigma_i^{-1/2} -I}_F\\
&\leq \Norm{\Sigma_i^{-1/2} \hat{\Sigma}^{1/2} (\hat{\Sigma}^{-1/2}\Sigma_j-I) \hat{\Sigma}^{-1/2} \hat{\Sigma}^{1/2}\Sigma_i^{-1/2}}_F + \Norm{\Sigma_i^{-1/2} \hat{\Sigma} \Sigma_i^{-1/2} -I}_F\\
&\leq \Norm{\Sigma_i^{-1/2} \hat{\Sigma}^{1/2}}_2^2  \Norm{\hat{\Sigma}^{-1/2}\Sigma_j \hat{\Sigma}^{-1/2}-I}_F + \Norm{\Sigma_i^{-1/2} \hat{\Sigma} \Sigma_i^{-1/2} -I}_F\\
&\leq (1+\Delta)\Delta + \Delta \leq 3\Delta^2\mper 
\end{align*} 
Here, in the first inequality, we used the triangle inequality for Frobenius norm and in the second inequality, used the contraction principle for Frobenius norms twice: for any matrices $A,B$, $\Norm{AB}_F \leq \Norm{A}_2 \Norm{B}_F$ along with the fact that $\Norm{\Sigma_i^{-1/2}\hat{\Sigma}^{1/2}v}_2^2 \leq (1+\Delta) \Norm{v}_2^2$. 

Finally, observe that for any vector $v$:
\begin{align*}
\iprod{\mu_i - \mu_j,v}^2 &\leq 2\iprod{\mu_i - \hat{\mu},v}^2 + 2\iprod{\hat{\mu}-\mu_j,v}^2 \\
&\leq \Delta v^{\top} (\Sigma_i + \Sigma_j + 2\hat{\Sigma}) v \leq 2\Delta v^{\top} (\Sigma_i + \Sigma_j) v\mper
\end{align*}

\end{proof}

\phantomsection
  \addcontentsline{toc}{section}{References}
  \bibliographystyle{amsalpha}
  \bibliography{bib/mathreview,bib/dblp,bib/custom,bib/scholar,bib/custom2}  

\appendix
\newpage
\section{Deferred Proofs} \label{sec:deferred-proofs}

\begin{proof}[Proof of Corollary~\ref{cor:cert-concentration-props-poly}]
  Observe that the polynomial inequality in indeterminate $z$, $z^2 + 2\delta^2p_\delta^2(z)- \delta^2 \geq 0$ holds. To see this, consider the following two cases: 1) $z^2 \geq \delta^2$: in this case, we are done because $p_{\delta}^2$ is non-negative. 2) $z^2 < \delta$: in this case, we use the fact that $p_{\delta}^2(z) \geq (1-\delta)^2$ which, for $\delta < 0.1$ implies that $2 \delta^2 p_{\delta}^2(z) \geq \delta^2$. 

  Now, using Fact~\ref{fact:univariate}, we know that $\sststile{\cO(s(\delta))}{z} \{ z^2 + 2 \delta^2 p_{\delta}^2(z) - \delta^2 \geq -\eta\}$.
  
  As a result, we know that $z^2 + 2 \delta^2 p_{\delta}^2(z) - \delta^2 + \eta = r(z)$ for some sos polynomial $r$ in $z$ with coefficients upper-bounded by $2\cdot (4z)^s$ and degree $= s(\delta)$ (because a polynomial is identically $0$ on reals if and only if all its coefficients are $0$). Further, we observe that $r$ is an even polynomial because $p$ is even. 
  
  Now, let's substitute $z = \frac{\iprod{x,v}}{\sqrt{v^{\top} \Sigma v}}$ in $p$ and $r$. Since $p,r$ are even, all monomials in $z$ appearing with non-zero coefficients in $p,r$ are even powers and are thus monomials in $z^2$. As a result, $(v^{\top} \Sigma v)^{s(\delta)} (z^2 + 2 \delta^2 p_{\delta}^2(z) - \delta^2 + \eta -r (z))$ is a polynomial in indeterminate $v$ for any given $x$. Further, $(v^{\top} \Sigma v)^{s(\delta)} r$ is SoS in $v^{\top} \Sigma v$ and therefore, also SoS in indeterminate $v$. So, putting this all together upon rearranging gives $\sststile{\cO(s) + \poly\log(1/\eta)}{v} \Set{(v^\top\Sigma v)^{s-1}\langle x,v\rangle^2+2\delta^2q_{\delta,\Sigma}^2(x, v)\ge (\delta^2 - \eta)(v^\top\Sigma v)^s}$ for all $\eta > 0$. Taking $\eta = 0.01\delta^2$ suffices for the conclusion.
  
  The second inequality follows from $\sststile{\cO(s(\delta))}{z}\{\bbE[p_\delta^2(\langle x,v\rangle)]\le \cO(\delta)\}$ upon substitution.
\end{proof}

\section{Bit Complexity Analysis} \label{sec:bit-complexity}

We will use the following basic observations in our bit complexity analysis to analyze the rank deficient covariance $\Sigma_*$ of Theorem~\ref{thm:list-decodable-cov-mean-estimation-section}. 

Recall that in this case, we assume that the target matrix $\Sigma_*$ has rational entries with bit complexity at most $B$. The following proposition shows that in this case, the smallest non-zero eigenvalue of $\Sigma_*$ cannot be too small.

\begin{proposition}[Smallest non-zero singular value of a rational matrix] \label{prop:smallest-non-zero-sing}
For $d \in \N$, let $A \in \bbQ^{d \times d}$ be a non-zero matrix with each entry of bit length at most $B$. Then, every non-zero eigenvalue of $A$ has absolute value at least $2^{-3Bd^3}$.
\end{proposition}

\begin{proof}
Let $L$ be the least common multiple of the denominators of the rational numbers appearing in entries of $A$. Then, since each denominator is upper-bounded by $2^B$, $L \leq 2^{Bd^2}$. Thus, $A' = L A$ is a matrix with integer entries. Observe further that by the Gershgorin circle theorem, the spectral norm of $A'$ (and thus, the eigenvalue of largest magnitude) is at most $n 2^{Bd^2} 2^{B} \leq 2^{2Bd^2}$.

Let $r \leq d$ be the rank of $A'$. 
Consider the characteristic polynomial $char(A')$ of $A$ in indeterminate $\lambda$. Then, $char(A')$ is monic and has integer coefficients. Consider the coefficient of $\lambda^r$ in $char(A')$. Then, this coefficient equals the sum of $r$-wise products of eigenvalues of $A'$. Since $A'$ has rank $r$, it has exactly $r$ non-zero eigenvalues and thus, the coefficient of $\lambda^r$ is the product of the non-zero eigenvalues of $A'$. Since the coefficients of $A'$ is a non-zero integer, this product is at least $1$ in magnitude. Since all eigenvalues of $A'$ are of magnitude at most $2^{2Bd^2}$, the smallest magnitude of any eigenvalue must thus be at least $2^{-2Bd^2r} \geq 2^{-2Bd^3}$. 

Thus, every non-zero eigenvalue of $A$ has magnitude at least $L^{-1} 2^{-2Bd^3} \geq 2^{-3Bd^3}$.

\end{proof}

We will also need the following basic facts about the classical algorithm for lattice basis reduction due to Lenstra, Lenstra and Lov\'asz~\cite{MR682664-Lenstra82}. 

\paragraph{Preliminaries on Integer Lattices} Let $A \in \bbQ^{d \times d}$ be a matrix of rationals. The lattice defined by $A$ is the discrete additive subgroup $\cL (a_1, a_2, \ldots, a_d) = \sum_{i = 1}^d z_i a_i$ as $z_i$s vary over $\Z$ and $a_i$s are the columns of $A$. We write $\lambda_1(\cL)$ to be the length of the smallest non-zero vector in $\cL$. More generally, let $\lambda_i(\cL)$ be the minimum of the maximum length of any vector from among all linearly independent sets of $i$ vectors $v_1, v_2, \ldots, v_i \in \cL$. 

The determinant of a lattice $\cL$ is defined as $\det(\cL) = \det(A^{\top}A)$ if $A$ is full-rank. 
Notice that $\det(\cL)$ is independent of the basis used for $\cL$. 

\begin{fact}[Minkowski's Theorems]\label{fact:minkowski}
$\lambda_1(\cL) \leq \det(\cL)$. More generally, $\prod_{i \leq d} \lambda_i(\cL) \leq d^{\sqrt{d}} \det(\cL)$. 
\end{fact}

Given a matrix $A \in \bbQ^{d \times k}$, the orthogonal lattice $\cL^{\perp}(A)$ defined by $A$ is the set of all integer vectors $v$ such that $Av = 0$. We can relate the size of the basis for $\cL^{\perp}(A)$ to that of $\cL(A)$ via Hadamard's inequality:

\begin{fact}[Hadamard's Inequality] \label{fact:hamadard}
$\det(\cL^{\perp}(A)) \leq \det(\cL(A))$. 
\end{fact}

Finally, we recall the guarantees of the lattice basis reduction algorithm of~\cite{MR682664-Lenstra82}. 

\begin{fact}[LLL Algorithm] \label{fact:LLL}
Let $\cL$ be a lattice defined by a $d \times d$ matrix $A$. There is a polynomial time algorithm that takes input the Gram matrix $A^{\top}A$ and outputs a basis $b_1, b_2,\ldots, b_d$ of $\cL$ such that $\Norm{b_i}_2 \leq 2^{O(D)} \lambda_i(\cL)$. 
\end{fact}

\paragraph{Analyzing the algorithm in Theorem~\ref{thm:list-decodable-cov-mean-estimation-section} in the word RAM model} 
We begin by setting $\lambda$ to be $2^{-B\cdot d^{C/\alpha}}$ for a sufficiently large constant $C$. We start by running the algorithm described in the proof of Theorem~\ref{thm:list-decodable-cov-mean-estimation-section} on input sample after adding a ($\poly(Bd)$ bit rational truncation of) an independent sample from $N(0,\lambda I)$ to each $y_i \in Y$. This allows us to effectively assume that the smallest eigenvalue of the unknown covariance is $\lambda$ and thus our analysis applies. 

As a result, we obtain a list of candidates one of which gives a good approximation (in $\dist$) to the $\Sigma_* + \lambda I$. Observe that if $\Sigma_*$ had $2^{-\poly(d)}$ large smallest eigenvalue, then the resulting list is already a good approximation in $\dist$ to $\Sigma_*$. If not, then, since $\Sigma_*$ has rational entries of bit complexity $\leq B$,  the determinant of the sublattice of which $\Sigma_*$ is a gram matrix is at most $(Bd)^d$. Thus, by Minkowski's theorem, there must an \emph{integer} basis $v_1, v_2, \ldots,$ with entries of bit complexity $\leq O(Bd^2)$ for the orthogonal lattice of $\Sigma_*$. 

Let $d-r$ be the rank of $\Sigma_*$. Since we ran the algorithm above on $\Sigma_* + \lambda I$, for small enough $\lambda \ll 2^{-\poly(Bd^2)}$, we must have there is a candidate $\hat{\Sigma}$ in the list with $r$ eigenvalues at most $2 \lambda$. 

We take every such candidate $\hat{\Sigma}$ and consider the quadratic form on integer vectors $v$: $Q(v) = v^{\top} \hat{\Sigma} v + \sqrt{\lambda} \norm{v}_2^2$. If $\Sigma_*$ has an integer vector $v$ in its kernel of length $\leq 2^{O(Bd^2)}$, then, the same $v$ must satisfy $Q(v) \leq \lambda^{1/4}$ for the ``good'' candidate $\hat{\Sigma}$. Thus, using the LLL  Algorithm (Fact~\ref{fact:LLL}), we can find a reduced basis for all such vectors $v$ where $Q(v)$ is within a $2^{O(d)}$ factor from the minimum possible value of $Q(v)$ over all non-zero integer vectors $v$. If for such a $v$, $Q(v) > 2^{\Omega(d)}\lambda^{1/4}$, we know that the the corresponding $\hat{\Sigma}$ couldn't possibly be a good candidate. On the other hand, for any $v$ such that $Q(v) \leq 2^{O(d)}\lambda^{1/4}$, we must have $v^T \hat{\Sigma} v < 2^{O(d)} \lambda^{1/4}$ and $\norm{v}_2 \leq 2^{O(d)} \lambda^{-1/2}$ by our choice of $Q(v)$. Further, the projection of any such $v$ on to $\ker(\Sigma_*)$ is either $0$ or has magnitude at least $2^{O(Bd)}$ because of Proposition~\ref{prop:smallest-non-zero-sing}. But if it were the latter, $Q(v) \gg 2^{\Omega(d)}\lambda^{1/4}$. Thus, if $\hat{\Sigma}$ were a good candidate, it must be that every $v$ in our reduced basis is in the kernel of $\Sigma_*$. We can now project the candidate $\hat{\Sigma}$ on to the orthogonal complement of the subspace defined by the reduced basis. This projection can be done exactly over rationals of bit complexity $\poly(Bd)$. For any ``good'' candidate $\hat{\Sigma}$, this will ensure that the range space of $\hat{\Sigma}$ is \emph{exactly} equal to that of $\Sigma_*$ as we desired.

\end{document}